\documentclass[11pt,a4paper]{article}

\def\Anonymity{0}
\def\IddoSys{1}
\def\CommentsOn{1}

\usepackage{comment}   
\excludecomment{comment}

\usepackage[T1]{fontenc}
\usepackage{lmodern}
\usepackage{microtype}

\usepackage{xcolor}
\usepackage{mathtools}
\usepackage{amssymb}
\usepackage{amsthm}
\usepackage{tabularx}
\usepackage[normalem]{ulem}
\usepackage{thmtools}
\usepackage{enumitem}
\usepackage{array}
\usepackage{tcolorbox}
\tcbuselibrary{skins,breakable}

\usepackage[pagebackref=true,colorlinks=true]{hyperref}
\usepackage[top=1in,bottom=1in,left=0.92in,right=0.92in]{geometry}

\ifnum\IddoSys=0
\usepackage{zref-clever} 
\zcsetup{cap,nameinlink,noabbrev}
\newcommand{\cref}{\zcref}
\newcommand{\Cref}{\zcref}
\zcRefTypeSetup{claim}{
    Name-sg = Claim ,
    name-sg = claim ,
    Name-pl = Claims ,
    name-pl = claims ,
}
\else
\usepackage[capitalize,nameinlink,noabbrev]{cleveref}
\fi

\usepackage[inline,final]{showlabels}

\definecolor{airforceblue}{rgb}{0.36, 0.54, 0.76}
\definecolor{darkcerulean}{rgb}{0.03, 0.27, 0.49}
\definecolor{darkelectricblue}{rgb}{0.33, 0.41, 0.77}
\hypersetup{
  linkcolor  = airforceblue,
  urlcolor   = blue,
  citecolor  = darkelectricblue,
  pdfstartview=FitH
}

\theoremstyle{plain}
\newtheorem{theorem}{Theorem}[section]
\newtheorem{lemma}[theorem]{Lemma}
\newtheorem{proposition}[theorem]{Proposition}
\newtheorem{corollary}[theorem]{Corollary}

\theoremstyle{definition}
\newtheorem{definition}[theorem]{Definition}
\newtheorem{remark}[theorem]{Remark}
\newtheorem{claim}[theorem]{Claim}

\newcommand{\probP}{\text{I\kern-0.15em P}}

\DeclareMathOperator{\depth}{depth} 
\DeclareMathOperator{\sdeg}{sdeg} 
\DeclareMathOperator{\coeff}{coeff} 
\DeclareMathOperator{\pow}{pow} 
\DeclareMathOperator{\sgn}{sgn} 
\DeclareMathOperator{\relrk}{relrk} 
\DeclareMathOperator{\poly}{poly} 
\DeclareMathOperator{\Rank}{Rank} 
\DeclareMathOperator{\perm}{perm} 

\newcommand{\IPS}{\ensuremath{\mathsf{IPS}}}

\newcommand{\PC}{\ensuremath{\mathsf{PC}}}
\newcommand{\PCR}{\ensuremath{\mathsf{PCR}}}

\newcommand{\IMM}{\ensuremath{\mathrm{IMM}}} 

\newcommand{\calL}{\mathcal{L}} 
\newcommand{\calC}{\mathcal{C}}
\newcommand{\x}{\mathbf{x}}
\newcommand{\Z}{\mathbb{Z}}

\newcommand{\F}{\mathbb{F}}
\newcommand{\calP}{\mathcal{P}}
\newcommand{\calN}{\mathcal{N}}

\newcommand{\calT}{\mathcal{T}} 
\newcommand{\out}{\mathsf{out}}

\newcommand{\Qbb}{\mathbb{Q}}
\newcommand{\Rbb}{\mathbb{R}}

\newcommand{\SLP}{\ensuremath{\mathsf{SLP}}}
\newcommand{\AUB}{\mathsf{AUB}} 
\newcommand{\RankP}{\mathsf{RankP}} 
\newcommand{\IRankP}{\mathsf{IRankP}} 
\newcommand{\IRankPE}{\mathsf{IRankPE}} 
\newcommand{\TRankP}{\mathsf{TRankP}} 
\newcommand{\PHP}{\mathsf{PHP}} 

\newcommand{\ACZ}{\ensuremath{\mathsf{AC}^0}}

\newcommand{\NCTwo}{\ensuremath{\mathsf{NC}^2}}

\newcommand{\VAC}{\ensuremath{\mathsf{VAC}}}
\newcommand{\VP}{\ensuremath{\mathsf{VP}}}
\newcommand{\VNP}{\ensuremath{\mathsf{VNP}}}
\newcommand{\TCzero}{\ensuremath{\mathsf{TC}^0}}
\newcommand{\FNCTwo}{\ensuremath{\mathsf{FNC}^2}}
\newcommand{\FACZ}{\ensuremath{\mathsf{FAC}^0}}
\newcommand{\FTCzero}{\ensuremath{\mathsf{FTC}^0}}

\newcommand{\VTCzero}{\ensuremath{\mathsf{VTC}^0}}
\newcommand{\VNCTwo}{\ensuremath{\mathsf{VNC}^2}}

\newcommand{\Vzero}{\ensuremath{\mathsf{V}^0}}
\newcommand{\IDeltazero}{\ensuremath{\mathsf{I\Delta}^0}}

\newcommand{\COMP}{\text{-}\mathbf{COMP}} 
\newcommand{\IND}{\text{-}\mathbf{IND}} 
\newcommand{\MIN}{\text{-}\mathbf{MIN}} 
\newcommand{\MAX}{\text{-}\mathbf{MAX}} 
\newcommand{\CoeffT}{\ensuremath{\mathbf{Coeff}}} 
\newcommand{\SMCoeffT}{\ensuremath{\mathbf{SMCoeff}}} 
\newcommand{\seq}{\ensuremath{\mathbf{seq}}} 
\newcommand{\row}{\ensuremath{\mathbf{Row}}} 
\newcommand{\bit}{\ensuremath{\mathbf{bit}}} 
\newcommand{\rank}{\ensuremath{\mathbf{rank}}} 
\newcommand{\relrank}{\ensuremath{\mathbf{relrank}}} 
\newcommand{\WordP}{\ensuremath{\mathbf{WordP}}} 
\newcommand{\binaryT}{\ensuremath{\mathbf{binary}}} 
\newcommand{\Prefix}{\ensuremath{\mathbf{Prefix}}} 
\newcommand{\unbiased}{\ensuremath{\mathbf{unbiased}}} 
\newcommand{\setmultilinear}{\ensuremath{\mathbf{set\textbf{-}multilinear}}}
\newcommand{\Size}{\ensuremath{\mathbf{Size}}}
\newcommand{\Depth}{\ensuremath{\mathbf{Depth}}}
\newcommand{\ProdDepth}{\ensuremath{\mathbf{ProdDepth}}} 
\newcommand{\PermT}{\ensuremath{\mathbf{Perm}}} 
\newcommand{\IMMT}{\ensuremath{\mathbf{IMM}}} 
\newcommand{\Circuit}{\ensuremath{\mathbf{Circuit}}}
\newcommand{\Partition}{\ensuremath{\mathbf{Partition}}}
\newcommand{\SetMultCircuit}{\ensuremath{\mathbf{SetMultCircuit}}}
\newcommand{\ProjCirc}{\ensuremath{\mathbf{ProjCirc}}}
\newcommand{\SetMultOutput}{\ensuremath{\mathbf{SetMultOutput}}}

\newcommand{\lb}{\ensuremath{\mathsf{AlgCircuit}}} 
\newcommand{\Type}{\ensuremath{\mathsf{Type}}}
\newcommand{\InputType}{\ensuremath{\mathsf{InputType}}}
\newcommand{\InputVar}{\ensuremath{\mathsf{InputVar}}}
\newcommand{\InputNode}{\ensuremath{\mathsf{InputNode}}}

\newcommand{\Krajicek}{Kraj\'{\i}\v{c}ek}
\newcommand{\Jerabek}{Je\v{r}\'{a}bek}

\tcbset{tableborder/.style={
    width=\columnwidth,
    boxrule=0.6pt,
    arc=2pt,
    left=3pt,right=3pt,top=3pt,bottom=3pt,
    colback=white
},}

\newcommand{\commentout}[1]{}

\ifnum\CommentsOn=0 
\newcommand{\MakeComment}[2]{}
\else 
\newcommand{\MakeComment}[2]{\textcolor{green}{\footnotesize{#1: #2}}}
\fi

\vfuzz \hfuzz
\providecommand{\Description}[1]{}

\begin{document}

\title{\textbf{Meta-Mathematics of Algebraic Complexity}\thanks{%
A preliminary version of this work appeared in the proceedings of ACM-IEEE LICS 2026~\cite{GGLST26}.
The work of Michal Garl\'{i}k, Svyatoslav Gryaznov, Jiaqi Lu, and Iddo Tzameret was supported by the European Research Council (ERC) under the European Union's Horizon 2020 research and innovation programme (grant agreement No 101002742, \mbox{\textit{EPRICOT}} project). The work of Jiaqi Lu, Rahul Santhanam, and Iddo Tzameret was also supported by the Engineering and Physical Sciences Research Council (EPSRC) under grant EP/Z534158/1.
Emails (in author order): \href{mailto:michal.garlik@gmail.com}{\texttt{michal.garlik@gmail.com}};
\href{mailto:svyatoslav.i.gryaznov@gmail.com}{\texttt{svyatoslav.i.gryaznov@gmail.com}};
\href{mailto:jl13023@ic.ac.uk}{\texttt{jl13023@ic.ac.uk}};
\href{mailto:rahul.santhanam@cs.ox.ac.uk}{\texttt{rahul.santhanam@cs.ox.ac.uk}};
\href{mailto:i.tzameret@imperial.ac.uk}{\texttt{i.tzameret@imperial.ac.uk}}.}}

\ifnum\Anonymity=0

\author{
 Michal Garl\'{i}k
 \smallskip
 \\
 \small Imperial College London
 \and
 Svyatoslav Gryaznov
 \smallskip
 \\
 \small Imperial College London
 \and 
 Jiaqi Lu
 \smallskip
 \\
 \small Imperial College London
 \and 
 Rahul Santhanam
 \smallskip
 \\
 \small University of Oxford
 \and
 Iddo Tzameret
 \smallskip 
 \\ 
  \small Imperial College London
 }

\date{} 
\else
\author{}
\date{}
\fi

\maketitle





\begin{abstract}                        
    We initiate the study of the meta-mathematics of algebraic circuit lower bounds, aiming
    both to gain insight into the methods sufficient and necessary to prove algebraic circuit
    lower bounds, and to contribute to the study of bounded arithmetic as a logical
    foundation for complexity lower bounds. 
    We demonstrate that while algebraic circuit lower bounds are hard for somewhat weak proof systems such as Polynomial Calculus Resolution (\PCR), contemporary lower bounds are efficiently provable in proof systems and bounded arithmetic theories corresponding to \NCTwo, such as \VNCTwo\ and the corresponding class of propositional Frege proofs of quasipolynomial-size. 
    Moreover, going below $\VNCTwo$ into algebraic constant-depth reasoning is likely insufficient to efficiently prove already constant-depth algebraic circuit lower bounds.
    Under an algebraic analogue of the \Krajicek{}--Razborov conjecture, our results suggest that proving superpolynomial lower bounds for general algebraic circuits will require stronger reasoning than that used in the rank-based proofs studied here.
    Specifically, we show the following.
    
    \begin{description}[labelwidth=1pt,itemindent=0pt,labelindent=0pt,leftmargin=10pt]
        \item[$\NCTwo$-reasoning and rank method] 
        
        Algebraic circuit lower bounds are often proved via the ``rank method'', with recent prominent applications including the constant-depth lower bounds of Limaye, Srinivasan and Tavenas~\cite{LST25} and Forbes~\cite{For24}. We show that these rank-based arguments can be formalised in the bounded arithmetic theory $\VNCTwo$, which captures reasoning with \NCTwo\ concepts. This complements the work of Tzameret and Cook~\cite{TC21}, who formalised structural \emph{upper} bounds in this theory, and provides a unified framework for studying barriers to current algebraic complexity methods, complementing barriers studied by Efremenko, Garg, Makam, Oliveira, and Wigderson~\cite{EGOW18,GMOW19}.

        \item[Sparsity algebraic reasoning]
        
        We show that Polynomial Calculus Resolution (\PCR) cannot efficiently prove superpolynomial algebraic circuit lower bounds for any family of polynomials. Moreover, \PCR\ cannot efficiently prove exponential constant-depth circuit lower bounds for any family of polynomials.
        
        \item[Constant-depth algebraic reasoning]
        We introduce the Tensor Rank Principle
        and demonstrate it is hard for PCR.
        We show that if this principle is hard against constant-depth Ideal Proof System (\IPS), then constant-depth  \IPS\ cannot efficiently prove constant-depth algebraic circuit lower bounds.  
    \end{description}

\end{abstract}                          
\clearpage 

\section{Introduction}

The meta-mathematics of computational complexity asks questions about the nature and limits of efficient computation, with a central focus on \emph{provability}: which proof systems or axiom systems are necessary or sufficient to establish hardness results, such as circuit-size lower bounds and complexity-class separations.

Beyond the formal provability of complexity lower bounds, one can also take a somewhat less formal approach by asking which ``proof methods'' can yield lower bounds, rather than in which formal proof system a given lower bound is provable. We call this direction \emph{the barriers approach}, in contrast to the \emph{meta-mathematical approach}.

In the \emph{Boolean} setting, there has been extensive work on the barriers approach: developing tools for analysing which methods can yield strong Boolean circuit-size lower bounds, and which are ruled out by general obstacles. Classic examples include the \textit{relativization} barrier of Baker, Gill and Solovay~\cite{BGS75} and the \textit{natural proofs} barrier of Razborov and Rudich~\cite{RR94}, both of which rule out broad families of standard approaches to proving superpolynomial lower bounds for $\mathsf{NP}$.

These barriers have been highly instructive and influential, as they rule out large classes of lower-bound attempts. However, they can sometimes have an ad-hoc flavour, in that it is not always unambiguous whether a given lower-bound argument falls within their scope. For example, it is unclear what the relativization barrier says about low-depth circuit classes; similarly, in the natural proofs framework it can be difficult to identify the underlying ``natural property'', and to separate this identification from the task of proving that an explicit hard function satisfies the property.

In contrast, the meta-mathematical approach to complexity lower bounds goes through \textit{formal logic} and specifically proof theory and independence results (namely, results showing that a theorem, and its negation, cannot be proved in a given formal theory of mathematics). One advantage of this approach is that barrier results that are formalised as logical independence results are more robust since they are closed under basic logical operations. For example, the natural-proofs barrier implies that there are no natural proofs of certain strong lower bounds. However, this does not exclude the following possibility: there is a natural proof of a statement $A$, and from $A$ one can derive, by a non-natural argument, a strong lower bound $B$. Thus, the barrier does not rule out proving $B$ in this way. Hence, a natural argument could be helpful in proving lower bounds despite the barrier, as long as it is only part of the proof. By contrast, logical barriers are formulated inside a fixed theory or proof system, and therefore rule out not only direct proofs of $B$, but also proofs that derive $B$ through intermediate statements such as $A$, provided that the implication $A \to B$ is itself provable in the theory.

When thinking of formal theories of interest in computational complexity theory, the natural choice is \textit{Bounded Arithmetic} theories, which are weak fragments of Peano Arithmetic with their axioms and definable functions usually restricted to a prescribed complexity class.  
Razborov~\cite{Razb95,Razb95a} initiated a general
investigation of the provability of Boolean circuit lower bounds in bounded arithmetic and
propositional proof complexity, and proved strong unconditional limitations for restricted proof
systems (e.g., the absence of efficient small-width DNF Resolution proofs of \emph{any}
superpolynomial Boolean circuit lower bound~\cite{Razb15-annals}; see also~\cite{Razborov98,Raz04-JACM} for preceding $\PC$ degree 
and resolution size lower bounds, respectively). More recent developments in this direction appear in works by Chen, Li and Oliveira~\cite{LO23, CLO26, Oli25}.

A central conjectural framework in
this direction is due to \Krajicek{} and Razborov~\cite{Krajicek04c,Razb96,Razb16}:%
\footnote{Razborov conjectured in~\cite{Razb15-annals} that Frege cannot efficiently prove superpolynomial circuit lower bounds for any Boolean 
function. More specifically,~\cite[Conjecture 1]{Razb15-annals} with suitable parameters for the 
underlying combinatorial designs implies under some hardness assumptions that Frege cannot efficiently prove that $\mathsf{SAT}\notin\mathsf{P}/{\poly}$. 
Further conjectures about the impossibility of \emph{Extended} Frege to efficiently prove circuit lower bounds have been circulated in the proof complexity literature and discussions  (cf.~\cite{Razb-2016-talk,Razb-2021-talk,Kra11}).}
\begin{tcolorbox}[tableborder]
\textbf{\Krajicek{}--Razborov Conjecture:}
Extended Frege cannot prove in polynomial size \emph{any} superpolynomial Boolean circuit lower bound (for any Boolean function).
\end{tcolorbox}

\medskip 
Similar to the developments described above, the study of barriers against proving \emph{algebraic} circuit lower bounds was also considered in the past, although it remains less advanced than its Boolean counterpart. Two prominent frameworks for barriers have emerged. The \emph{algebraic natural proofs}
framework by Forbes, Shpilka, and Volk, and independently Grochow, Kumar, Saks, and Saraf~\cite{FSV18,GKSS17} develops algebraic analogues of natural properties. It is known that many\footnote{One example of an algebraic circuit lower bound technique that is not natural in the sense of algebraic natural proofs is the substitution method (see~\cite[Chapter 6]{BCS97}). Another is any lower bound technique that applies to tensor rank, but not border rank, since natural proofs necessarily prove lower bounds on border complexity. We thank the anonymous reviewer for pointing this out.} currently established algebraic lower bounds naturalise in this sense~\cite{FSV18}. However, despite interesting recent work~\cite{CKRST20, KRST22} it remains
open whether there is a compelling algebraic hardness assumption ruling out algebraic natural proofs
against general algebraic circuits,\footnote{The non-existence of algebraic natural properties against general algebraic circuits follows from the existence of \emph{succinct} hitting set generators against algebraic circuits, but this is not an overly standard algebraic hardness assumption in any sense, and indeed is easily seen to be equivalent to the non-existence of natural properties.} analogous to the cryptographic evidence that natural properties against general Boolean circuits do not exist~\cite{RR94}. 
A different line of work studies barriers for \emph{rank-based} approaches. Efremenko, Garg, Oliveira, and Wigderson~\cite{EGOW18} establish unconditional limitations on matrix-rank methods for proving tensor- and Waring-rank lower bounds. Garg, Makam, Oliveira, and Wigderson~\cite{GMOW19} extend these barriers to certain methods based on higher-order tensor rank and Waring rank.
 The limitations shown in this direction are \emph{unconditional}, and apply to a broad class of techniques. On the other hand, these barriers apply only to tensor- and Waring-rank.


While barriers in algebraic circuit complexity have been investigated, its formal meta-mathematics remains unexplored, and in this work we initiate such a general study of ``frontier'' algebraic circuit lower bounds. The guiding question is:

\begin{tcolorbox}[tableborder]
\itshape
What is the proof-theoretic (or propositional proof-complexity) strength needed to formalise and prove frontier lower bounds in
algebraic complexity?
\end{tcolorbox}

A natural test case is the family of rank-based lower bounds for constant-depth algebraic circuits
computing the Permanent over finite fields, proved by Limaye--Srinivasan--Tavenas and extended by
Forbes~\cite{LST25,For24}. These results are among the strongest unconditional lower bounds currently
known for algebraic circuits, and their proofs use a robust and conceptually clean method (relative
rank and set-multilinearization), making them particularly suitable for our analysis.

\subsection{The rank method and constant-depth algebraic lower bounds}

Following the terminology of Garg~\textit{et~al.}~\cite[Definition 1.8]{GMOW19},
the essential feature of the rank method (which~\cite{GMOW19} calls ``matrix-rank method'' to distinguish it from objects of higher-order than matrices, e.g., tensors) is the use of a linear map
$    \phi \colon V \to \operatorname{Mat}_{p,q}$
and the associated subadditive measure
$    f \mapsto \operatorname{rank}(\phi(f)).
$
In the circuit lower-bound setting, such a measure yields a lower bound once it is
shown to be small for every polynomial computed by small circuits in the model
under consideration, and large for the target polynomial.

In particular, the arguments in Limaye, Srinivasan and Tavenas, and Forbes~\cite{LST25,For24} proceed by associating to a polynomial $f$ a
``relative rank'' measure (defined with respect to an appropriate partition of variables and a
set-multilinear structure), and then proving two complementary statements:
(i) the Permanent has large relative rank, and
(ii) every polynomial computed by a small constant-depth algebraic circuit has small relative rank.
The second part requires turning a circuit into an explicit representation of the relevant coefficient
vector (after suitable multilinearization), and then bounding the rank of this vector by structural properties of low-depth circuits.


From a proof-theoretic viewpoint, this raises concrete formalisation tasks:
one must define the rank measure inside a weak theory, prove its basic properties, and formalise the
circuit-to-coefficient-vector transformation (together with the bounds needed in the rank argument).
A priori, it is not obvious how much uniform reasoning is required, because the argument mixes
algebraic manipulations, combinatorial counting bounds, and circuit-structured constructions.

\subsection{Our results}

In this work we study the efficient provability of algebraic
circuit lower bounds. We consider this question in propositional and algebraic
proof systems, as well as in bounded arithmetic. The latter are first-order theories that correspond to certain computational complexity classes and serve as a uniform counterpart for propositional proof systems (see~\cite{CN10,Kra95}). We discuss the consequences of our work in \Cref{sec:discussion}.

    
\paragraph*{Summary of results} We shall demonstrate that while algebraic circuit lower bounds are hard for somewhat weak proof systems such as Polynomial Calculus Resolution (\PCR), they are provable (efficiently) in proof systems and theories corresponding to \NCTwo, such as the bounded arithmetic theory \VNCTwo\ and the corresponding propositional proof system $\NCTwo$-Frege.
Our results suggest a complementary limitation. We show that the particular rank-based constant-depth lower-bound arguments of Limaye, Srinivasan and Tavenas, and Forbes~\cite{LST25,For24} can be formalised in the bounded arithmetic theory $\VNCTwo$.

This establishes a conditional limitation on the ability of rank-based arguments to establish general algebraic circuit lower bounds. In particular, suppose  an algebraic analogue of the \Krajicek{}--Razborov conjecture holds, asserting that $\VNCTwo$ cannot prove superpolynomial lower bounds against general algebraic circuits computing polynomials of polynomial degree. Then extensions of these rank-based lower bound arguments, whose correctness remains provable in $\VNCTwo$, cannot establish such super-polynomial algebraic circuit lower bounds.

We then show that ``algebraic constant-depth'' reasoning is plausibly insufficient to efficiently prove constant-depth algebraic circuit lower bounds. More precisely, we show that if a property of tensor rank is hard for algebraic proofs operating with constant-depth circuits (formally, constant-depth IPS), then so are constant-depth algebraic circuit lower bounds.

This picture---depicted in \Cref{tab:three-level-picture}---aligns with the working assumption that rank arguments, and rank-based lower bounds specifically, sit proof-theoretically (or proof complexity-wise) in reasoning systems that operate with concepts in \NCTwo\ (or more precisely, the complexity class $\mathsf{DET}$, that is the \ACZ-closure of the integer determinant; see~\cite{TC21} for a discussion and references).

\begin{table}[t]
\centering
\small
\begingroup
\renewcommand{\arraystretch}{1.12}
\setlength{\tabcolsep}{9pt}
\begin{tabularx}{\linewidth}{@{} >{\raggedright\arraybackslash}p{0.27\linewidth} >{\raggedright\arraybackslash}X @{}}
\noalign{\hrule height 0.7pt\vskip 5pt}
\textbf{Theory or proof system} & \textbf{Algebraic circuit lower bounds} \\
\noalign{\vskip 4pt\hrule height 0.4pt\vskip 5pt}
\textbf{\NCTwo\ reasoning}\par\smallskip \VNCTwo\ and \NCTwo-Frege &
\textbf{Efficiently provable.} Rank-based constant-depth lower bounds for the Permanent over finite fields are provable in \VNCTwo, giving polynomial-size \NCTwo-Frege and quasipolynomial-size Frege proofs. \\
\noalign{\vskip 0.8em}
\textbf{Constant-depth IPS}\par\smallskip Over $\F_2$ &
\textbf{Conditionally hard.} Constant-depth algebraic circuit lower bounds require superpolynomial-size proofs, assuming exponential-size constant-depth IPS lower bounds for the tensor rank principle. \\
\noalign{\vskip 0.8em}
\textbf{\PCR}\par\smallskip Over $\F_2$ &
\textbf{Unconditionally hard.} Superpolynomial algebraic circuit lower bounds require superpolynomial-size \PCR\ proofs. \\
\noalign{\vskip 5pt\hrule height 0.7pt}
\end{tabularx}
\endgroup
\caption{Three levels of reasoning and the status of algebraic circuit lower bounds.}\label{tab:three-level-picture}
\end{table}

\subsubsection{\texorpdfstring{$\VNCTwo$}{VNC2} and \texorpdfstring{$\NCTwo$}{NC2}-Frege suffice}

Our first contribution is to show that the entire rank-based lower bound argument against constant-depth algebraic circuits of Limaye \emph{et al.} and Forbes~\cite{LST25, For24} can be
carried out in the bounded arithmetic theory $\VNCTwo$.
Using the standard translation between bounded arithmetic and propositional proofs (see, e.g.,
Cook--Nguyen~\cite{CN10}), we obtain the corresponding propositional upper bound.

\begin{theorem}[Informal; \Cref{theorem: Provability of superpolynomial lower bound for the Permanent polynomial against constant-depth algebraic circuits} and \Cref{theorem: NC^2-Frege upper bound}]\label{thm:intro:mainthm-VNCTwo}
There are polynomial-size propositional $\NCTwo$-Frege proofs of the superpolynomial lower bound for
the Permanent polynomial against constant-depth algebraic circuits over finite fields.
Moreover, this lower bound is provable in the theory $\VNCTwo$.
\end{theorem}

To establish this, we provide $\VNCTwo$-formal definitions and proofs for the key components of the
rank-method framework in the parameter regime used by~\cite{LST25,For24}.
In particular, we formalise:
\begin{itemize}
\item the definition of the relevant relative-rank measure and its basic algebraic properties;
\item the fact that the specific ``word polynomials'' used in the argument have high rank;
\item a uniform procedure that, given a constant-depth algebraic circuit of bounded syntactic-degree, computes the corresponding coefficient vector needed for the rank analysis;
\item the set-multilinearization step in~\cite{For24}: every constant-depth circuit can be converted into a set-multilinear circuit with only a linear increase in depth and a moderate size blow-up;
\item the final rank upper bound for constant-depth set-multilinear circuits, which implies a circuit-size lower bound for the specific ``word polynomials''.
\end{itemize}

Conceptually, the theorem supports the informal correspondence between rank-based lower bounds and
$\NCTwo$-level reasoning: the proof of~\cite{LST25,For24} does not require the strength of powerful systems
such as Extended Frege.



\cref{thm:intro:mainthm-VNCTwo} gives an upper bound for the proof complexity of the circuit lower bounds shown in~\cite{LST25, For24}. We believe that the upper bound is close to being tight (with respect to the proof system), which is consistent with the constant-depth algebraic analogues of the \Krajicek{}--Razborov conjecture.

%
%


\paragraph*{Proof overview for \Cref{thm:intro:mainthm-VNCTwo}}
%
%
We shall work with \emph{first-order} theories of bounded arithmetic, which are translatable to polynomial-size propositional proofs. Specifically, we work with the theory \VNCTwo\ as developed in Cook and Nguyen~\cite{CN10} and used in Tzameret and Cook~\cite{TC21} to formalise basic linear algebra in a seemingly minimal theory. Any formula of low logical complexity (formally, $\forall\Sigma^B_0$-formulas) that has a proof in \VNCTwo\ translates naturally (by instantiating the universal variables in it by constants) to a family of propositional tautologies with polynomial-size propositional proofs in \NCTwo-Frege. 

The formalisation of the lower bound proceeds along the lines 
of~\cite{LST25} (which works over fields of large characteristic), except where we use Forbes' idea for finite fields~\cite{For24}.
Specifically, this means that we use a direct set-multilinearization step using the Binet--Minc~\cite{minc1979evaluation} identities. 
We note that for simplicity of presentation, the upper bound in \NCTwo-Frege we demonstrate is not precisely for $\lb(\perm_n,s,\ell,\Delta)$ as is used in \Cref{theorem:intro:super-poly lower bound for AUP}, but a slightly different formulation (though it can be made to match precisely our original propositional formula).

    There are three main challenges in the formalisation of the lower bounds in the theory, as follows.

    \subparagraph*{Irrational numbers:} We need to encode and reason about irrational numbers and functions on reals, such as $\exp$ and $\ln$, as in~\cite{LST25}, while the theory cannot do this directly. However, \Jerabek{}~\cite{Jev23} has formalised these functions in the theory $\VTCzero$, which is contained in $\VNCTwo$. Moreover, \Jerabek{} proved that $\VTCzero$ can prove some properties about these functions, such as $\exp(z+w) = \exp(z) \exp(w)$ and $\exp(\log z) = z$, and these properties are sufficient for our formalisation.

    \subparagraph*{Diophantine approximation:}
    We need to use the following Diophantine approximation from~\cite{LST25}: for $a,b\in\Z$ with $(a,b)\neq(0,0)$,
    \(
        |a\alpha-b| \geq \frac{1}{4|a\alpha|+2}
    \),
    where $\alpha=1/\sqrt{2}$. We show that to use this approximation, it is sufficient to prove the following statement in the theory:
    \(
        x^2=2 \ \land\ (a,b)\neq(0,0) \ \implies\ ax\neq b
    \),
    which in itself expresses the irrationality of $x$. This statement can be proved in the theory by showing that the greatest common divisor of $a$ and $b$ is definable in $\VNCTwo$ using the number-maximisation axiom~\cite{HJ18-Mathoverflow}.

    \subparagraph*{Properties of relative rank:}
    The lower bound in~\cite{LST25} is a rank argument and hence uses several properties of matrix rank.
    Mulmuley~\cite{mulmuley1986fast} showed how to compute the rank of a matrix by using the determinant. Following the work of Tzameret and Cook~\cite{TC21} on linear algebra in \VNCTwo, Ken and Kuroda~\cite{KK23} used Mulmuley's algorithm to define in the theory the rank function. Using their rank function, the required properties of the rank function can be proved in $\VNCTwo$.

\subsubsection{Polynomial Calculus Resolution lower bounds}

%

%

We study the hardness of the statements $\lb(f,s,\ell)$ stating that a polynomial $f$ has an algebraic circuit of size $s$ and syntactic-degree $\ell$. We also study analogous statements $\lb(f,s,\ell,\Delta)$, where the circuits are restricted to depth $\Delta$. Similar formulas were considered for Boolean circuits~\cite{Razborov98,Raz04-JACM,RazbPM04,Krajicek04c,Razb15-annals,GGRT25,AR23} in the case of Resolution, Polynomial Calculus, Polynomial Calculus Resolution and Sum-of-Squares. 

We define $\lb(f,s,\ell)$ in a similar fashion to the lower bound statement from Razborov~\cite{Razb15-annals}. We prove its hardness against $\PCR$ over the two-element field by reducing (the negation of) the \emph{rank principle} to $\lb(f,s,\ell)$. Informally, the (weak) rank principle states that an $m \times m$ matrix $A$ cannot be represented as a product of an $m \times n$ matrix and an $n \times m$ matrix when $m>n$ and the rank of $A$ exceeds $n$.
Garl\'{\i}k, Gryaznov, Ren and Tzameret~\cite{GGRT25} studied the complexity of different encodings of the rank principle, including its \emph{iterated} variant, which is the variant we use in the reduction. This allows us to show that $\PCR$ cannot efficiently prove lower bounds on the algebraic circuit complexity of an $n$-variable polynomial $f$ when $s \ge {(\ell n)}^{\Omega(1)}$.

Below we use $\PCR_{\F_2}$ to denote \PCR\ refutations over $\F_2$.

    \begin{theorem}[\PCR\ cannot efficiently prove algebraic circuit lower bounds; \Cref{theorem: super-poly lower bound for AUP}]\label{theorem:intro:super-poly lower bound for AUP}
        There exists a constant $c>0$ such that given any polynomial $f$ of degree at most $\ell$ over $\F_2$ in $n$ variables and $s \ge {(\ell n)}^c$, any $\PCR_{\F_2}$ refutation of $\lb(f,s,\ell)$ requires superpolynomially many monomials.
    \end{theorem}

This establishes a weaker version of an algebraic analogue of the \Krajicek{}--Razborov conjecture by showing that there are no polynomial-size \PCR\ proofs of superpolynomial  algebraic circuit lower bounds for $f$, for \emph{any} polynomial $f$. 

    As far as we are aware, \cref{theorem:intro:super-poly lower bound for AUP} is the first result showing unconditional hardness of proving an algebraic circuit lower bound.

Constant-depth algebraic circuits form a strictly weaker computational model than general algebraic circuits. Consequently, establishing lower bounds against constant-depth algebraic circuits~\cite{LST25, For24} is expected to be easier than proving lower bounds against general algebraic circuits. It follows that proving proof-complexity lower bounds for formulas expressing lower bounds against constant-depth algebraic circuits is  \emph{stronger} than proving proof-complexity lower bounds for formulas expressing lower bounds against general algebraic circuits. By reducing to the rank principle, we obtain a superpolynomial lower bound for exponential constant-depth algebraic circuit lower-bound formulas against $\PCR$:

    \begin{theorem}[\PCR\ cannot efficiently prove exponential constant-depth algebraic circuit lower bounds;\Cref{theorem: super-poly lower bound for AUP constant-depth}]
        For every constant $c>0$, given a polynomial $f$ of degree at most $\ell$ over $\F_2$ in $n$ variables and $s \ge N^{1/c}$, where $N$ is the number of monomials in $n$ variables of degree at most $\ell$, every $\PCR_{\F_2}$ refutation of $\lb(f,s,\ell,O(c))$ requires superpolynomially many monomials.
    \end{theorem}

\subsubsection{Conditional constant-depth IPS lower bounds}

While in the previous section we discussed the relatively weak \PCR\ proof system, here we shall deal with a stronger system: the Ideal Proof System (\IPS) introduced by Grochow and Pitassi~\cite{GP18}, and specifically its constant-depth fragment. Intuitively, one should think about it roughly as an algebraic proof system similar in strength to polynomial calculus in which each proof-line is written as an algebraic constant-depth circuit, instead of as a sum of monomials (as in \PCR). Alternatively, one can consider constant-depth IPS over a finite field of constant characteristic as constant-depth Frege with counting gates modulo $p$ (augmented with the ability to freely prove [arithmetised] polynomial identities).    

We consider the generalisation of the rank principle (see above), called the \emph{tensor rank principle}, denoted $\TRankP_n^m(r,A)$. It states that an $r$-tensor $A \in \F^{m \times \cdots \times m}$ does not admit a tensor-rank decomposition of rank at most $n$, which is smaller than $m$. When $r=2$, this principle coincides with the rank principle.
We show a conditional superpolynomial lower bound on $\lb(f,s,n\ell,\Delta)$
against constant-depth $\IPS$. We construct a substitution that converts a short refutation of $\lb(f,s,n\ell,\Delta)$ into a short refutation of an order-$n$ tensor rank instance over $\F_2$, using ideas from Raz~\cite{Raz04-JACM}.

\begin{theorem}[Constant-depth IPS cannot efficiently prove constant-depth algebraic circuit lower bounds conditionally; See \Cref{them: conditional IPS Lower Bounds reduction} and \Cref{corollary: conditional IPS lower bounds}]\label{thm:intro condiitonal IPS lower bounds}
    Consider a family of polynomials $f\in\F_2[u_1,\ldots,u_n]$ of degree at most $\ell$, where $n\geq1$ and $\ell\geq2$. Set
    \[
        D=n\ell,\qquad
        N=\binom{n+\ell}{\ell},\qquad
        R=\binom{n+D}{D},\qquad
        K=\lfloor\sqrt{s}\rfloor.
    \]
    Assume $K\to\infty$, $K<N$, $\Delta>4$, and $s\geq4n^2(\ell+1)^4$. Let $A$ be the diagonal order-$n$ tensor of dimension $N$ whose $(M,\ldots,M)$th entry is $\coeff_M(f)$, with all other entries zero. Suppose that, for some fixed $\varepsilon>0$ and constant $d>5$, the corresponding systems $\TRankP^N_K(n,A)$ require $d$-$\IPS$ refutations over $\F_2$ of size at least $2^{K^\varepsilon}$, and that
    \[
        \log(sn)+n\log N+D\log R=o(K^\varepsilon).
    \]
    Then $\lb(f,s,D,\Delta)$ requires $(d-5)$-$\IPS$ refutations over $\F_2$ of size $|\lb(f,s,D,\Delta)|^{\omega(1)}$.
\end{theorem}

Tensor rank lower bounds are notoriously difficult to prove, and even basic mathematical questions about tensor rank remain poorly understood. Since tensors have been studied in mathematics for much longer than algebraic circuits, there is arguably stronger empirical evidence for the difficulty of proving tensor rank lower bounds than for the difficulty of proving algebraic circuit lower bounds. Thus, grounding the formal difficulty of reasoning about algebraic circuits in the difficulty of reasoning about tensors seems to provide meaningful evidence for the hardness of proving algebraic circuit lower bounds (this was pointed out by the anonymous reviewer).

To complement our lower bounds, we show an unconditional lower bound on $\TRankP_n^m(r,A)$ against $\PCR_{\F_2}$. This is achieved using a random restriction argument combined with a degree-reduction to the pigeonhole principle.

    \begin{theorem}[\Cref{theorem: exp lower bound for TRankP}]
        Let $m,n,r$ be positive integers with $m>n$.
        Let $J \subseteq [m]$ with $|J| > n$ and let $A$ be the $r$-tensor defined as
        \[
            A_{i_1,\ldots,i_r} = \begin{cases}
                1 & \text{if } i_1 = \cdots = i_r \in J, \\
                0 & \text{otherwise}.
            \end{cases}
        \]
        Then any $\PCR_{\F_2}$ refutation of $\TRankP^m_n(r, A)$ requires $2^{\Omega(n)}$ monomials.
    \end{theorem}





\subsection{Relation to other work}

Santhanam and Tzameret~\cite{ST25}, and later Lu~\textit{et~al}.~\cite{LST26}, investigated the proof complexity of statements asserting algebraic circuit lower bounds. However, these works did not directly study the hardness of the algebraic circuit lower-bound statements themselves. Instead, they considered a more meta-level question: whether one can prove lower bounds for the \emph{statement} that there are no short proofs of algebraic circuit lower bounds. This meta-level lower bound was then used as an ingredient in a diagonalisation argument.

Garl\'{i}k, Gryaznov, Ren, and Tzameret~\cite{GGRT25} studied two related questions. In the first, they established a $\PCR_{\F_2}$ lower bound for noncommutative algebraic branching program (ncABP) lower bound statements. The model of ncABP is quite a weak model of computation, and further the result is specialised in the sense that the rank principle is intimately connected to the way ncABP are encoded. The second related question was the introduction of bounded arithmetic theories augmented with the rank principle as an axiom. They showed that \emph{Boolean} complexity lower bounds such as Smolensky--Razborov~\cite{Smolensky87} lower bounds against constant-depth Boolean circuits with mod $p$ gates can be formalised in theories corresponding to $\ACZ[p]$-reasoning augmented with the weak rank principle. These theories seem close to the theory \VNCTwo\ we use in this work.  

M\"{u}ller and Pich~\cite{MP20} and others showed that Extended Frege (which corresponds to a proof system operating with polynomial-size circuits) can prove known Boolean circuit lower bounds. Our results are analogous to theirs, as we show that $\NCTwo$-Frege, i.e., proofs operating with \NCTwo\ circuits can prove known algebraic circuit lower bounds (where the class $\NCTwo$ corresponds to polynomial-size algebraic circuits, due to known depth reductions; i.e., a polynomial-size algebraic circuit computing a polynomial of polynomial degree, can be converted to a depth-$O(\log^2 n)$ circuit of $\poly(n)$ size, for $n$ the number of variables).

\subsection{Discussion and conclusions}\label{sec:discussion}

We conclude the introduction by explaining what our results suggest from the perspectives of algebraic and proof complexity.
\medskip

First, from the perspective of algebraic complexity, our results show that known techniques,
including frontier constant-depth algebraic circuit lower bounds such as those of Limaye, Srinivasan
and Tavenas~\cite{LST25}, can be formalised in a relatively weak theory corresponding to \NCTwo\ reasoning,
namely \VNCTwo. This should be contrasted with stronger bounded arithmetic theories, such as $\mathsf{PV}$
and $S^1_2$, which correspond to polynomial-time reasoning.

Thus, under an algebraic analogue of the \Krajicek{}--Razborov conjecture, our results suggest
a possible limitation of current methods. Suppose the algebraic analogue of \Krajicek{}--Razborov
conjecture holds, namely, \VNCTwo\ reasoning cannot prove algebraic circuit lower bounds for \NCTwo\
circuits, and hence, by depth reduction, for \VP. Since we show that the present rank-based
constant-depth lower bounds are provable already in \VNCTwo, any proof of substantially stronger
algebraic circuit lower bounds would have to use ideas that are stronger, at least in their logical
strength, than the concepts currently used in these arguments.
\medskip

Second, from the perspective of proof complexity, our results contribute to three objectives: 
\begin{itemize}
    \item 
 We provide new hard instances for proof systems of interest, in the tradition of proof complexity lower bounds. In our case these are instances that are themselves meta-mathematical statements (cf.~Razborov and Raz~\cite{Razborov98,Raz04-JACM}). This can help to establish further new proof complexity lower bounds.

    \item Our instances are connected to the diagonalisation-based approach to proof complexity lower bounds. In the works of Santhanam and
    Tzameret~\cite{ST25} and Lu, Santhanam and Tzameret~\cite{LST26}, strong
    proof complexity lower bounds were obtained for CNF formulas expressing the
hardness of proving algebraic circuit lower bounds. A caveat in this approach is
that it is not known, unconditionally, whether these CNF formulas are
unsatisfiable. The present work identifies a condition under which such formulas
are unsatisfiable: namely, the hardness of proving algebraic circuit lower
bounds in constant-depth \IPS. Combined with~\cite{LST26}, this condition would
imply lower bounds for $\ACZ[p]$-Frege, a longstanding open problem in proof
complexity. 

    \item Studying the unprovability of lower bounds is a major effort in proof complexity. It helps relate the strength of efficient computation to that of efficient provability, which is motivated by the goal to understand different proof systems and bounded arithmetic theories.
\end{itemize}

    \section{Preliminaries}

        To fix notation, we define Boolean circuits with fan-in two gates $\land,\lor$ and fan-in one gates $\neg$: for every $n \in \mathbb{N}$, an $n$-input, single-output Boolean circuit is a directed acyclic graph with $n$ sources (vertices with no incoming edges) and one sink (a vertex with no outgoing edges). All non-source vertices are called gates or nodes and are labelled with one of $\land$, $\lor$ or $\neg$ (i.e., the logical operations OR, AND, and NOT). The vertices labelled with $\land$ and $\lor$ have fan-in (i.e., number of incoming edges) equal to 2, and the vertices labelled with $\neg$ have fan-in 1. The size of $C$, denoted by $|C|$, is the number of vertices in it.
        If $C$ is a Boolean circuit and $x \in {\{0,1\}}^n$ is some input, then the output of $C$ on $x$ is denoted by $C(x)$.

        We use $[n]$ to denote the set $\{1,\ldots, n\}$.

    \subsection{Algebraic complexity}
    We give a brief overview of standard notions in algebraic complexity. For a more detailed treatment, we refer the reader to~\cite{SY10}.

    Let $\mathbb{F}$ be a field. Denote by $\mathbb{F}[\overline{x}]$ the ring of polynomials with coefficients from $\mathbb{F}$ and variables $\overline{x}=(x_1,\ldots,x_n)$. A polynomial is a formal linear combination of monomials, where a monomial is a product of variables. Two polynomials are identical if all their monomials have the same coefficients. The degree of a polynomial is the maximum total degree of a monomial in it.

    \begin{definition}[Algebraic circuits and algebraic formulas]\label{def: algebraic circuits and formulas}
          An \emph{algebraic circuit} over a field $\mathbb{F}$ is a finite directed acyclic graph. The leaves (nodes with in-degree zero) are called \emph{input nodes}. Each input node is labelled either with a variable or a field element in $\mathbb{F}$. All the other nodes have in-degree two and are labelled by either $+$ (addition) or $\times$ (multiplication). The output of a $+$-node is the sum of the polynomials computed by its incoming nodes, and the output of a $\times$-node is their product. An algebraic circuit is called an \emph{algebraic formula} if the underlying directed acyclic graph is a tree (every node has at most one outgoing edge). The \emph{size} of an algebraic circuit $C$ is the number of nodes in it, denoted by $|C|$. The \emph{depth} of $C$ is the length of the longest directed path in it, denoted by $\depth(C)$. The \emph{product-depth} of $C$ is the maximum number of $\times$-nodes on a directed path in $C$. The polynomial computed by a node $u$ in an algebraic circuit is denoted by $\widehat{C}_u$. The polynomial computed by an algebraic circuit $C$ is denoted by $\widehat{C}$.
    \end{definition}

    \begin{definition}[Depth-$\Delta$ algebraic circuits and algebraic formulas]\label{def: bounded-depth algebraic circuits and formulas}
        A \emph{depth-$\Delta$ algebraic circuit (or formula)} over a field $\mathbb{F}$ is an algebraic circuit (or formula) with unbounded fan-in of depth $\Delta$. Except for the input nodes, all other nodes may have unbounded fan-in and are labelled by $+$ or $\times$. All the rest of the notions are defined similarly to \cref{def: algebraic circuits and formulas}.
    \end{definition}

     \begin{definition}[Syntactic-degree $\sdeg(\cdot)$]\label{def: syntactic-degree}
        Let $C$ be an algebraic circuit (with possibly unbounded fan-in) and $v$ a node in $C$. The \emph{syntactic-degree} $\sdeg(v)$ of $v$ is defined as follows:
        \begin{enumerate}[leftmargin=*]
            \item If $v$ is a field element or a variable, then $\sdeg(v) \coloneqq 0$ and $\sdeg(v) \coloneqq 1$, respectively;
            \item If $v=\sum_{i=0}^t u_i$ then $\sdeg(v) \coloneqq \max\{\sdeg(u_0),\ldots, \sdeg(u_t)\}$;
            \item If $v=\prod_{i=0}^t u_i$ then $\sdeg(v) \coloneqq \sum_{i=0}^t \sdeg(u_i)$.
        \end{enumerate}
    \end{definition}

    Negation is not a primitive arithmetic gate: $-g$ is computed as $(-1)g$. Since constants have syntactic degree zero, multiplication by $-1$ does not increase syntactic degree.

    \begin{definition}[Algebraic Branching Program]
            \label{def: ABPs}
            An \emph{algebraic branching program} (ABP) over a field $\F$ and variables $\mathbf{x} = (x_1,\ldots,x_n)$ is a directed acyclic graph $G = (V,E)$ with the following properties:
            \begin{itemize}[leftmargin=*]
                \item The vertices are partitioned into layers
                \[
                    V = L_0 \cup L_1 \cup \cdots \cup L_d,
                \]
                where $L_0 = \{s\}$ contains a unique source and $L_d = \{t\}$ contains a unique sink.
                \item Edges go only between consecutive layers: $E \subseteq \bigcup_{i=0}^{d-1} L_i \times L_{i+1}$.
                \item Each edge $e \in E$ is labelled by an affine linear form $\ell_e(\mathbf{x}) \in \F[x_1,\ldots,x_n]$.
            \end{itemize}
            The ABP computes the polynomial
            \[
                f(\mathbf{x}) \;=\; \sum_{P: s \to t} \prod_{e \in P} \ell_e(\mathbf{x}),
            \]
            where the sum is taken over all directed paths from $s$ to $t$. The \emph{size} of the ABP is $|V| + |E|$, and the \emph{width} is $\max_i |L_i|$.
    \end{definition}

    \begin{definition}[$\VP$~\cite{Val79:ComplClass}]
         Over a field $\mathbb{F}$, $\VP_\mathbb{F}$ is the class of families $f={(f_n)}_{n=1}^{\infty}$ of polynomials $f_n$ such that $f_n$ has $\poly(n)$ input variables, has degree $\poly(n)$, and can be computed by algebraic circuits over $\mathbb{F}$ of $\poly(n)$ size.
    \end{definition}

    \begin{definition}[$\VNP$~\cite{Val79:ComplClass}]
         Over a field $\mathbb{F}$, $\VNP_\mathbb{F}$ is the class of families $g={(g_n)}_{n=1}^{\infty}$ of polynomials $g_n$ such that $g_n$ has $\poly(n)$ input variables and has degree $\poly(n)$, and can be written as
        \[
            g_n(x_1, \ldots , x_{\poly(n)})=\sum_{\overline{e} \in {\{0,1\}}^{\poly(n)}} f_n(\overline{e},\overline{x})
        \]
        for some family $(f_n) \in \VP_\mathbb{F}$.
    \end{definition}

    \begin{definition}[$\VAC^0$]\label{def: VAC0}
        Over a field $\F$, $\VAC^0_\F$ is the class of families $f=(f_n)_{n=1}^{\infty}$ of polynomials $f_n$ such that $f_n$ has $\poly(n)$ input variables, has degree $\poly(n)$, and can be computed by algebraic circuits over $\mathbb{F}$ of $\poly(n)$ size and depth $O(1)$.
    \end{definition}
    Notice that $\VP$, $\VNP$ and $\VAC^0$ are all non-uniform complexity classes.

    \begin{definition}[Projection reduction~\cite{Val79:ComplClass}]
         A polynomial $f(x_1,\ldots, x_n)$ is a projection of a polynomial $g(y_1,\ldots, y_m)$ if there is a mapping $\sigma$ from $\{y_1,\ldots,y_m\}$ to $
         \{0,1, x_1, \ldots, x_n\}$ such that $f(x_1,\ldots,x_n) = g(\sigma(y_1),\ldots,\sigma(y_m))$.
        A family of polynomials $(f_n)$ is a polynomial projection or p-projection of another family $(g_n)$ if there is a function $t(n)=n^{\Theta(1)}$ such that $f_n$ is a projection of $g_{t(n)}$ for all (sufficiently large) $n$.
        We say that $f$ is projection-reducible to $g$ if $f$ is a projection of $g$.
    \end{definition}
    
    The symmetric group over $n$ elements $\{1,\ldots,n\}$, denoted by $S_n$, is the group consisting of all the bijective functions from $[n]$ to $[n]$ whose group operation is that of function composition. The sign $\sgn(\sigma)$ of a permutation $\sigma\in S_n$ is $1$ if $\sigma$ can be expressed as a product of an even number of transpositions (exchanges of two, not necessarily consecutive, entries), and $-1$ otherwise.
    
    \begin{definition}[Determinant]
        The Determinant of an $n\times n$ matrix $A = (a_{i,j})$ is defined as
        \[
            \det(A) \coloneqq \sum_{\sigma \in S_n} (\sgn(\sigma) \prod_{i=1}^n a_{i,\sigma(i)}).
        \]
    \end{definition}
    
    \begin{definition}[Permanent]
        The Permanent of an $n \times n$ matrix $A = (a_{i,j})$ is defined as
        \[
            \perm(A) \coloneqq \sum_{\sigma \in S_n} \prod_{i=1}^n a_{i,\sigma(i)}.
        \]
    \end{definition}

    \begin{theorem}[\cite{Val79:ComplClass}]
        For every field $\mathbb{F}$, every polynomial family on $n$ variables that is computable by an algebraic formula of size $u$ is projection-reducible to the Determinant polynomial (over the same field) on $u+2$ variables.
        For every field $\mathbb{F}$ of characteristic different from $2$, every polynomial family in $\VNP_\mathbb{F}$ is projection-reducible to the Permanent polynomial (over the same field) with polynomially more variables.
    \end{theorem}
    It is known that the Permanent polynomial is complete under p-projections for $\VNP$ when the field $\mathbb{F}$ has characteristic different from $2$. The Determinant polynomial is not known to be complete for $\VP$ under p-projections.

    \begin{definition}[Iterated Matrix Multiplication]
        Let $n$ and $d$ be positive integers. The \emph{Iterated Matrix Multiplication $\IMM_{n,d}$} on $N=dn^2$ variables is defined as follows. The underlying variables are partitioned into $d$ sets $X_1,\ldots,X_d$ of size $n^2$, each of which is represented as an $n \times n$ matrix with distinct variable entries. Then $\IMM_{n,d}$ is defined to be the $(1,1)$th entry of the matrix product $X_1\cdot X_2 \cdots X_d$.
    \end{definition}

    The Iterated Matrix Multiplication is in $\VP$.
    \begin{theorem}[Superpolynomial lower bounds against constant-depth circuits over a field of large characteristic~\cite{LST25}]\label{theorem: first-LST21}
        Assume $d \leq \frac{\log n}{100}$ and the characteristic of $\mathbb{F}$ is $0$ or greater than $d$. For any product-depth $\Delta \geq 1$, any algebraic circuit $C$ computing $\IMM_{n,d}$ of product-depth at most $\Delta$ must have size at least $n^{d^{\varepsilon}}$ where $\varepsilon=\exp(-O(\Delta))$.
    \end{theorem}

    A recent result by Forbes~\cite{For24} extended this result to \emph{any} field, including finite fields. The lower bounds in~\cite{LST25} were improved in~\cite{BDS24} and were later extended by Forbes~\cite{For24} to arbitrary fields. However, in this paper, we will focus on formalising proofs in~\cite{LST25} and~\cite{For24}.

    \begin{corollary}[Superpolynomial lower bounds on constant-depth circuits over any field~\cite{For24}]
    The lower bound of \cref{theorem: first-LST21} holds over any field.
    \end{corollary}

    Since $\IMM_{n,d}$ is a p-projection of the Permanent polynomial with $\poly(n,d)$ many variables, it follows that the Permanent does not have constant-depth polynomial-size circuits over any field, in the following sense.
    Recall that a \emph{symbolic $n\times n$ matrix} is the $n\times n$ matrix with $n^2$ distinct indeterminates $x_{ij}$, for $i,j\in[n]$.
    \begin{theorem}[No polynomial-size constant-depth circuits for the Permanent]\label{theorem: LST21}
        Let $\Delta \geq 1$ and let $\mathbb{F}$ be a field.
        There are no constants $c_1,c_2$ such that the Permanent polynomial $\perm(A)$ of the $n \times n$ symbolic matrix $A$ over the field $\mathbb{F}$ is computable by an algebraic circuit of size $c_1 n^{c_2}$ and depth $\Delta$, where $\Delta$ is independent of $n$, for all sufficiently large $n$.
    \end{theorem}

    \subsubsection{Universal algebraic circuits}
        For general algebraic circuits, we have the following.
    \begin{theorem}[Existence of universal circuits for homogeneous polynomials~\cite{Raz10}]
        \label{theorem: Existence of universal circuits for homogeneous polynomials}
        Let $\mathbb{F}$ be a field and let $\overline{x}$ consist of $n$ variables, and let $C_{t,d}^{\hom}$ denote the class of all homogeneous polynomials of total degree exactly $d$ in $\mathbb{F}[\overline{x}]$ that have algebraic circuits of size at most $t$. Then there is a circuit $U(\overline{x},\overline{w}) \in \mathbb{F}[\overline{x},\overline{w}]$ of size $O(d^2 t^8)$ and syntactic-degree $O(d)$ with $K_{t,d}=O(d^2 t^8)$ auxiliary variables $\overline{w}$ disjoint from $\overline{x}$, that is universal for $C_{t,d}^{\hom}$ in the following sense: if $f(\overline{x}) \in C_{t,d}^{\hom}$, then there exists $\overline{a} \in \mathbb{F}^{K_{t,d}}$ such that $U(\overline{x},\overline{a})=f(\overline{x})$. Moreover, given $t$ and $d$ as input, the number $K_{t,d}$ can be computed in time $\poly(t,d)$. We use $K_{t,d}$ to denote this number of auxiliary variables.
    \end{theorem}

    We use the following simple adaptation from~\cite{ST21}:
    \begin{definition}[Universal circuits for polynomials~\cite{ST21}]
        The universal circuit for polynomials of degree at most $d$ and circuit size at most $t$ is defined as:
        \begin{equation*}
            U(\overline{x},\overline{w})=\sum_{i=0}^d U_i(\overline{x},\overline{w}),
        \end{equation*}
        where $U_i(\overline{x},\overline{w})$ is the universal circuit for the class $C_{t,i}^{\hom}$ of homogeneous degree-$i$ polynomials in $\mathbb{F}[\overline{x}]$ and where the $\overline{w}$-variables in each distinct $U_i(\overline{x},\overline{w})$ are pairwise disjoint.
    \end{definition}
    The size of $U(\overline{x},\overline{w})$ is $\sum_{i=0}^d O(i^2t^8)=O(d^3t^8)$.\par

     \begin{theorem}[Existence of constant-depth universal circuits for polynomials computed by constant-depth circuits~\cite{LST26}]
        \label{theorem: Existence of universal circuits for constant-depth polynomials}
         Let $\mathbb{F}$ be a field and let $\overline{x}$ consist of $n$ variables, and let $\VAC^0_{t,d}$ denote the class of all polynomials in $\mathbb{F}[\overline{x}]$ that have algebraic circuits of size at most $t$ and depth at most $d$ where $t$ is polynomially bounded and $d$ is constant. Then there is a circuit $U(\overline{x},\overline{w}) \in \mathbb{F}[\overline{x},\overline{w}]$ of size $\poly(t)$ and constant-depth with $K$ auxiliary variables $\overline{w}$ disjoint from $\overline{x}$, that is universal for $\VAC^0_{t,d}$ in the following sense: if $f(\overline{x}) \in \VAC^0_{t,d}$, then there exists $\overline{a} \in \mathbb{F}^K$ such that $U(\overline{x},\overline{a})=f(\overline{x})$. Observe that given $t$ and $d$, $K$ can be computed efficiently and is bounded by $\poly(t)$. We use $K_{t,d}$ to denote the number $K$ computed from $t$ and $d$.
     \end{theorem}

\subsubsection{Extracting coefficients and \texorpdfstring{$\IPS$}{IPS} refutation formula}
        Let $f(\overline{x},\overline{w}) \in \mathbb{F}[\overline{x}, \overline{w}]$ be a polynomial, and let $M=\prod_{i \in I} x_i^{\alpha_i} \cdot \prod_{j \in J}w_j^{\beta_j}$ be a monomial in $f(\overline{x},\overline{w})$, for some $\alpha_i,\beta_i \in \mathbb{N}$. We call $\sum_{i \in I}\alpha_i$ the $\overline{x}$-degree of $M$.
        \begin{definition}[$\coeff_M(\cdot)$~\cite{ST21}]
            \label{def: coefficient circuit}
            Let $f(\overline{x},\overline{w})$ be a polynomial over $\mathbb{F}$ in the disjoint sets of variables $\overline{x}, \overline{w}$. Let $M$ be an $\overline{x}$-monomial of degree $j$. Viewing $f$ as a polynomial in $\F[\overline{w}][\overline{x}]$, $\coeff_M(f(\overline{x},\overline{w}))$ is defined as the (polynomial) coefficient from $\mathbb{F}[\overline{w}]$ (that is, in the $\overline{w}$-variables only) of $M$ in $f(\overline{x},\overline{w})$.
        \end{definition}
        Observe that $f(\overline{x},\overline{w})=\sum_{M_i} M_i \cdot \coeff_{M_i}(f(\overline{x},\overline{w}))$, where the $M_i$ are all possible $\overline{x}$-monomials of degree at most $d$, for $d$ the maximum $\overline{x}$-degree of a monomial in $f(\overline{x},\overline{w})$.

        \begin{proposition}[Computation of coefficients in general circuits~\cite{ST21}]
            \label{prop: computation of coefficients in general circuits}
            Let $f(\overline{x},\overline{w})$ be a polynomial over $\mathbb{F}$ in the disjoint sets of variables $\overline{x},\overline{w}$. Suppose that $M$ is an $\overline{x}$-monomial of degree $d$, and assume that there is an algebraic circuit computing $f(\overline{x},\overline{w})$ of size $s$ and syntactic-degree $\ell$. Then, there is a circuit of size $O(7^d\cdot s)$ computing $\coeff_M(f(\overline{x},\overline{w}))$ of syntactic-degree $\ell^{O(1)}$.
        \end{proposition}

        \begin{proposition}[Computation of coefficients in constant-depth circuits]
            \label{prop: computation of coefficients in constant-depth circuits}
            Let $f(\overline{x},\overline{w})$ be a polynomial over $\mathbb{F}$ in the disjoint sets of variables $\overline{x},\overline{w}$. Suppose that $M$ is an $\overline{x}$-monomial of degree $d$, and assume that there is an algebraic circuit $C(\overline{x},\overline{w})$ computing $f(\overline{x},\overline{w})$ of maximum multiplication fan-in $t$, size $s$, syntactic-degree $\ell$ and depth $\Delta$, where $\Delta$ is a constant such that
            \begin{enumerate}[leftmargin=*]
                \item all edges from the leaves are to $+$ nodes;
                \item all output-nodes are $+$ nodes;
                \item the nodes of $C$ are alternating: that is, if $v$ is a $+$ node and $(u,v)$ is an edge, then $u$ is a $\times$ node, and if $v$ is a $\times$ node and $(u,v)$ is an edge, then $u$ is either a leaf or a $+$ node.
            \end{enumerate}
            Then there is a constant-depth algebraic circuit of depth $\Delta$ and size $O(2^{(t+d)d} \cdot s)$ computing $\coeff_M(f(\overline{x},\overline{w}))$ of syntactic-degree $\ell^{O(1)}$.
        \end{proposition}

    \subsubsection{Encoding of algebraic circuit upper bound formulas}

    \begin{definition}[Algebraic circuit upper bound formula]
        \label{def: Algebraic circuit upper bound formula}
        Let $f(\overline{x})\in \mathbb{F}[\overline{x}]$ be a polynomial in $n$ variables of degree at most $\ell$. The following set of polynomial equations $\AUB(f,s,\ell)$ (in the $\overline{w}$-variables only) expresses that the polynomial $f(\overline{x})$ can be computed by an algebraic circuit of size $s$:
        \[
            \left\{\coeff_{M_i} (\widehat{U}(\overline{x},\overline{w})) = b_i : 1\leq i \leq N\right\},
        \]
        where $\overline{b}=\coeff(f(\overline{x})) \in \mathbb{F}^N$ is the $N$-dimensional coefficient vector of the polynomial $f$, $U(\overline{x},\overline{w})$ is the universal circuit for polynomials of degree at most $\ell$ and circuit size at most $s$, $\overline{w}$ are the auxiliary edge variables of $U$, ${\{M_i\}}_{i=1}^N$ is the set of all possible $\overline{x}$-monomials of degree at most $\ell$, and $N=\sum_{j=0}^\ell \binom{n+j-1}{j}=2^{O(n+\ell)}$ is the number of monomials of total degree at most $\ell$ over $n$ variables. The size of this set of polynomial equations is $O(7^\ell\ell^3s^8N)$.

    \end{definition}

    \begin{definition}[Bounded-depth algebraic circuit upper bound formula]
        \label{def: Constant-depth algebraic circuit upper bound formula}
        Let $f(\overline{x})$ be a polynomial over $\mathbb{F}$ in $n$ variables $\overline{x}$ of degree at most $\ell$. The following set of polynomial equations $\AUB(f,s,\ell,\Delta)$ (in the $\overline{w}$-variables only) expresses that the polynomial $f(\overline{x})$ can be computed by an algebraic circuit of size $s$ and depth $\Delta$:
        \[
            \left\{\coeff_{M_i} (\widehat{U}(\overline{x},\overline{w})) = b_i : 1\leq i \leq N\right\},
        \]
        where $\overline{b}=\coeff(f(\overline{x})) \in \mathbb{F}^N$ is the $N$-dimensional coefficient vector of the polynomial $f$, $U(\overline{x},\overline{w})$ is the bounded-depth universal circuit for polynomials computable by circuits of depth at most $\Delta$ and size at most $s$, $\overline{w}$ are the $K_{s,\Delta}$ edge variables, ${\{M_i\}}_{i=1}^N$ is the set of all possible $\overline{x}$-monomials of degree at most $\ell$, and $N=\sum_{j=0}^\ell \binom{n+j-1}{j}=2^{O(n+\ell)}$ is the number of monomials of total degree at most $\ell$ over $n$ variables. The size of this set of polynomial equations is $O(2^{(t+\ell)\ell}\cdot \poly(s) \cdot N)$, where $t$ is the maximum multiplication fan-in in $U(\overline{x},\overline{w})$.\par
    \end{definition}

    \begin{definition}[Straight Line Program ($\SLP$)]\label{def: Straight line program SLP}
        An $\SLP$ of a circuit $C(\overline{x})$, denoted by $\SLP(C(\overline{x}))$, is a sequence of equations between variables such that, assuming all the equations hold, the extension variable corresponding to the output node equals the value computed by the circuit. Formally, we choose any topological order $g_1,g_2, \ldots,g_i,\ldots,g_{|C|}$ on the nodes of the circuit $C$ (that is, if there is a directed path from $g_j$ to $g_k$ in $C$, then $j<k$) and define the following set of equations to be the $\SLP$ of $C(\overline{x})$:
        \begin{center}
            $g_i = g_{j1} \circ g_{j2} \circ \cdots \circ g_{jt}$ for $\circ \in \{+, \times \}$ iff $g_i$ is a $\circ$ node in $C$ with $t$ incoming edges from $g_{j1},\ldots,g_{jt}$.
        \end{center}
        An $\SLP$ representation of a circuit equation $C(\overline{x})=0$ means that we add to the $\SLP$ above the equation $g_{|C|}=0$, where $g_{|C|}$ is the output node of the circuit.
    \end{definition}

    \subsection{Proof complexity}\label{sec:propositional proofs}
 
        \begin{definition}[Polynomial Calculus Resolution~\cite{MR1919962}]\label{def: PCR}
            Let $\F$ be a field. \emph{Polynomial Calculus Resolution over $\F$ ($\PCR_\F$)} is a proof system whose lines are polynomials in $\F[x_1, \ldots, x_n, \overline{x_1}, \ldots, \overline{x_n}]$, where $\overline{x_1},\ldots,\overline{x_n}$ are treated as new formal variables. $\PCR$ has Boolean axioms for all variables, including the variables $\overline{x_i}$, together with the additional axioms $x_i + \overline{x_i}=1 : i \in [n]$.
            Each line is either an axiom or is derived from the previous lines using the following inference rules:
            \[
                \frac{f \qquad g}{\alpha f + \beta g}
                \quad\text{and}\quad
                \frac{f}{x\cdot f},
            \]
            where $\alpha, \beta \in \mathbb{F}$ are any scalars and $x$ is a variable. 
    
            For a clause $C$, denote by $\Gamma_C$ the monomial
            \begin{equation*}
                \Gamma_C \coloneqq \prod_{\overline{x} \in C} x \cdot \prod_{x \in C} \overline{x}
            \end{equation*}
            and for a CNF $\tau$, let $\Gamma_\tau \coloneqq \{\Gamma_C : C \in \tau \}$. Note that $\tau$ is unsatisfiable if and only if the polynomials in $\Gamma_\tau$, together with the Boolean and complementarity axioms of $\PCR$, have no common root. A \emph{$\PCR$ refutation} of a CNF $\tau$ is a $\PCR$ proof of the contradiction $1=0$ from $\Gamma_\tau$.

            The \emph{degree} of a $\PCR$ proof is defined as the maximum degree of a polynomial appearing in it, and its \emph{size} is the number of distinct \emph{monomials} appearing in the proof.
        \end{definition}


        It is convenient to reduce all polynomials appearing in $\PCR$ proofs modulo the ideal generated by the Boolean axioms. This makes the proof lines multilinear, and can only decrease their size and degree.

        We shall also use the system \NCTwo-Frege, which is an inference-based propositional proof system whose proof lines are represented by \NCTwo-circuits. These are fan-in-two De Morgan circuits of depth $O(\log^2 n)$, where $n$ is the number of input variables.
        We refer the reader to~\cite{HT15} for a formal definition of this system. 

        \paragraph*{Ideal Proof Systems}
            Given $f_1, \ldots, f_m \in \mathbb{F}[x_1, \ldots,x_n]$ over some field $\mathbb{F}$, Hilbert's Nullstellensatz shows that the polynomial system $f_1(\overline{x})=\cdots = f_m(\overline{x})=0$ is unsatisfiable (over the algebraic closure of $\mathbb{F}$) if and only if there are polynomials $g_1,\ldots,g_m \in \mathbb{F}[\overline{x}]$ such that $\sum_{j=1}^m g_j(\overline{x}) f_j(\overline{x})=1$ (as a formal identity), or equivalently, that $1$ is in the ideal generated by ${\{f_j\}}_{j=1}^m$ (we use the notation from~\cite{FSTW21}).
             \begin{definition}[(Boolean) Ideal Proof System ($\IPS$)~\cite{GP18}] \label{def: Boolean ips}
                Let $f_1(\overline{x}), \ldots, f_m(\overline{x}), p(\overline{x})\in \mathbb{F}[x_1,\ldots,x_n]$ be a collection of polynomials. An $\IPS$ proof of $p(\overline{x})=0$ from ${\{f_j(\overline{x})\}}_{j=1}^m$, showing that $p(\overline{x})=0$ is semantically implied from the assumptions ${\{f_j(\overline{x})=0\}}_{j=1}^m$ over $0$-$1$ assignments, is an algebraic circuit
                \[
                    C(\overline{x},\overline{y},\overline{z}) \in \mathbb{F}[\overline{x}, y_1,\ldots,y_m,z_1,\ldots,z_n],
                \]
                such that the following equalities are formal polynomial identities.\footnote{That is, $C(\overline{x},\overline{0},\overline{0})$ computes the zero polynomial and $C(\overline{x},f_1(\overline{x}),\ldots,f_m(\overline{x}),x_1^2-x_1, \ldots, x_n^2 -x_n)$ computes the polynomial $p(\overline{x})$}
                \begin{enumerate}[leftmargin=*]
                    \item $C(\overline{x},\overline{0},\overline{0})=0$.
                    \item $C(\overline{x},f_1(\overline{x}), \ldots, f_m(\overline{x}),x_1^2-x_1 , \ldots, x_n^2 -x_n)=p(\overline{x})$.
                \end{enumerate}

                The size of the $\IPS$ proof is the size of the circuit $C$. The variables $\overline{y},\overline{z}$ are sometimes called the \emph{placeholder variables} since they are used as placeholders for the axioms. An $\IPS$ proof $C(\overline{x},\overline{y},\overline{z})$ of $1=0$ from ${\{f_j(\overline{x})=0\}}_{j=1}^m$ is called an $\IPS$ refutation of ${\{f_j(\overline{x})=0\}}_{j=1}^m$.
                If $C$ comes from a restricted class of algebraic circuits $\mathcal{C}$, then such a refutation is called a $\mathcal{C}$-$\IPS$ refutation. When $\calC$ is the class of all algebraic circuits of depth at most $\Delta$, we write $\Delta$-$\IPS$ to denote $\calC$-$\IPS$.
            \end{definition}

\section{Hardness of Algebraic Circuit Lower Bounds from Hardness of Weak Rank Principle}\label{sec: pcr lower bound}

    We define the \emph{polynomial system $\lb(f,s,\ell)$} as a formalisation stating that a polynomial $f$ can be computed by an algebraic circuit of size $s$ and syntactic degree $\ell$.  
    Informally, this involves encoding the type of each gate (whether it is a variable, constant, $+$ gate, or $\times$ gate), the inputs of the gate, and the polynomial it computes by specifying the coefficient vector over all possible monomials of degree at most $\ell$. Additionally, it is convenient to introduce auxiliary variables to represent the polynomials computed by the inputs of $+$ and $\times$ gates. The axioms of $\lb(f,s,\ell)$ ensure the correctness of the encoded circuit.  
    This definition closely resembles the lower bound formula from~\cite{Razb15-annals}.
    \begin{definition}[$\lb(f,s,\ell)$]
        Let $f(\overline{u})$ be a polynomial in $n$ variables. Let $N$ be the number of $\overline{u}$-monomials of degree at most $\ell$ and ${\{M_i\}}_{i \in [N]}$ be all such monomials. The variables of $\lb(f,s,\ell)$ together with their intended meaning are listed in \cref{fig: variables of lb formula}.

        \begin{table}[t]
            \centering
            \begin{tcolorbox}[tableborder]
            \begingroup
            \renewcommand{\arraystretch}{1.25}
            \setlength{\tabcolsep}{5pt}
            \begin{tabularx}{\textwidth}{@{} p{0.35 \textwidth} X @{}}
                Variable & Meaning \\
                \hline
                $\alpha_{M_i,v} : i \in [N]$ & the coefficient of $M_i$ in the polynomial computed by the gate $v$. \\
                \hline
                $\alpha_{M_i,v,b} : i \in [N], b \in [s]$ & the coefficient of $M_i$ in the polynomial computed by the $b$th input of $v$. \\
                \hline
                $\Type(v)$ & $0$ if $v$ is a $+$ gate and $1$ if $v$ is a $\times$ gate. \\
                \hline
                $\InputType_b(v) : b \in [s]$ & $0$ if the $b$th input of $v$ is a variable or a constant and $1$ if it is a gate. \\
                \hline
                $\InputType'_b(v) : b \in [s]$ & when $\InputType_b(v)=0$, it is $0$ if the $b$th input of $v$ is a constant and $1$ if it is a variable. \\
                \hline
                $\InputType''_b(v) : b \in [s]$ & when $\InputType_b(v)=\InputType'_b(v)=0$, it is the constant value of the $b$th input of $v$. \\
                \hline
                $\InputVar_b(v, j) : b \in [s], j \in [n]$ & when $\InputType_b(v)=0$ and $\InputType'_b(v)=1$, it is $1$ iff the $b$th input of $v$ is the $x_j$. \\
                \hline
                $\InputNode_b(v, v') : b \in [s], v' < v$ & when $\InputType_b(v)=1$, it is $1$ iff the $b$th input of $v$ is the gate $v'$. 
            \end{tabularx}
            \endgroup
            \end{tcolorbox}
            \caption{Variables of $\lb(f,s,\ell)$, for every gate $v \in [s]$.}\label{fig: variables of lb formula}
        \end{table}

        For every $v \in [s]$, the axioms of $\lb(f,s,\ell)$ are the following polynomials:
        \begin{align*}
            & (1-\InputType_b(v)) (1-\InputType'_b(v)) (\alpha_{1,v,b}-\InputType''_b(v)),
            \shortintertext{for all $b \in [s]$, making $\InputType''_b(v)$ consistent with $\alpha_{1,v,b}$ when the $b$th input of $v$ is a constant,}
            & (1-\InputType_b(v)) (1-\InputType'_b(v)) \alpha_{M_i,v,b},
            \shortintertext{for all $b \in [s], i \in [N], M_i \neq 1$, setting $\alpha_{M_i,v,b}$ to $0$ when the $b$th input of $v$ is a constant,}
            & (1-\InputType_b(v)) \InputType'_b(v) \InputVar_b(v, j) \InputVar_b(v, j'),
            \shortintertext{for all $b \in [s], j\neq j' \in [n]$, ensuring that at most one $\InputVar_b(v,j)$ is set to $1$ when the $b$th input of $v$ is a variable,}
            & (1-\InputType_b(v)) \InputType'_b(v) \left(\sum_{j \in [n]} \InputVar_b(v, j)-1\right),
            \shortintertext{for all $b \in [s]$, ensuring that at least one $\InputVar_b(v,j)$ is set to $1$ when the $b$th input of $v$ is a variable,}
            & (1-\InputType_b(v)) \InputType'_b(v) \InputVar_b(v, j) (\alpha_{x_j,v,b}-1),
            \shortintertext{for all $b \in [s], j \in [n]$, making $\InputVar_b(v, j)$ consistent with $\alpha_{x_j,v,b}$ when the $b$th input of $v$ is a variable,}
            & (1-\InputType_b(v)) \InputType'_b(v) \InputVar_b(v, j) \alpha_{M_i,v,b},
            \shortintertext{for all $b \in [s], j \in [n], i \in [N], M_i \neq x_j$, setting $\alpha_{M_i,v,b}$ to $0$ when the $b$th input of $v$ is a variable,}
            & \InputType_b(v) \InputNode_b(v, v') \InputNode_b(v, v''),
            \shortintertext{for all $b \in [s], v'\neq v'' < v$, ensuring that at most one $\InputNode_b(v, v')$ is set to $1$ when the $b$th input of $v$ is a gate,} 
            & \InputType_b(v) \left(\sum_{v' < v}\InputNode_b(v, v')-1\right),
            \shortintertext{for all $b \in [s]$, ensuring that at least one $\InputNode_b(v, v')$ is set to $1$ when the $b$th input of $v$ is a gate,}
            & \InputType_b(v) \InputNode_b(v, v') (\alpha_{M_i,v,b}-\alpha_{M_i,v'}),
            \shortintertext{for all $b \in [s],v'<v,i\in[N]$, making $\alpha_{M_i,v,b}$ and $\alpha_{M_i,v'}$ consistent when the $b$th input of $v$ is a gate,}
            & (1-\Type(v)) \left(\alpha_{M_i,v} - \sum_{b \in [s]} \alpha_{M_i,v,b}\right),
            \shortintertext{for all $i \in [N]$, making $\alpha_{M_i,v}$ equal to the sum of its children when $v$ is a $+$ node,}
            & \Type(v) \left(\alpha_{M_i,v} - \left( \sum_{\substack{i_1, \ldots, i_\ell \in [N] \\ \prod_{b \in [\ell]} M_{i_b} = M_i}} \prod_{b \in [\ell]} \alpha_{M_{i_b},v,b} \right) \prod_{b \in [s] \setminus [\ell]} \alpha_{1,v,b} \right),
            \shortintertext{for all $i \in [N]$, making $\alpha_{M_i,v}$ equal to the product of its children when $v$ is a $\times$ node,}
            & \Type(v) \prod_{b \in [\ell]} \alpha_{M_{i_b},v,b},
            \shortintertext{for all $i_1, \ldots, i_\ell \in [N]$ with $\deg(\prod_{b \in [\ell]} M_{i_b}) > \ell$, setting all high-degree monomials to $0$,}
            & \Type(v) \alpha_{M_i,v,b},
            \shortintertext{for all $b \in [s] \setminus [\ell], i \in [N], M_i \neq 1$, ensuring that only the first $\ell$ inputs of $v$ are non-constant polynomials,}
            & \alpha_{M_i,s} - \coeff_{M_i}(f),
        \end{align*}
        for all $i \in [N]$, making $\alpha_{M_i,s}$ (the coefficient vector of the output gate) consistent with the coefficient vector of $f$.

        The size of $\lb(f,s,\ell)$ is $\poly(s, n, N^\ell)$.
    \end{definition}

    \begin{remark}
        We limit the syntactic degree of the circuit to $\ell$, thus every $\times$ node can have at most $\ell$ inputs computing non-constant polynomials. This allows us to assume that these input gates are among the first $\ell$ inputs, thus making the size of the encoding significantly smaller: it suffices to consider only $N^\ell$ tuples of monomials instead of $\binom{N}{s}$.
    \end{remark}

    We demonstrate the hardness of $\lb(f,s,\ell)$ by reducing an iterated encoding of the negation of the \emph{weak rank principle}~\cite{GGRT25} to it. The weak rank principle states that an $m\times m$ matrix of rank greater than $n$ cannot be expressed as the product of an $m\times n$ matrix and an $n\times m$ matrix. Under our substitution, every axiom of $\lb(f,s,\ell)$ has a short $\PCR$ derivation from the iterated-rank axioms; hence a short refutation of $\lb(f,s,\ell)$ would yield a short refutation of the iterated rank system. We first define the non-iterated polynomial system $\RankP_n^m(A)$ encoding the negation of the rank principle.

    \begin{definition}[The (Weak) Rank Principle~\cite{GGRT25}]\label{def: rank principle}
        Let $\F$ be a field and $m$ and $n$ two positive integers with $m > n$. We denote by $\RankP^m_n(A)$ the system of degree-$2$ polynomial equations stating that the rank of an $m \times m$ matrix $A$ is at most $n$ over $\F$. More precisely, $\RankP_n^m(A)$ is defined over $2mn$ variables arranged into two matrices $X \in \F^{m \times n}$ and $Y \in \F^{n \times m}$. For every $i,j \in [m]$, there is an equation in $\RankP_n^m(A)$ stating that the $(i,j)$th entry of the product $XY$ is equal to $A_{i,j}$. That is, $\RankP_n^m(A)$ consists of the polynomials
        \begin{align*}\label{equa: rank principle}
            &\sum_{k \in [n]} x_{i,k} y_{k,j} - A_{i,j}, & i,j \in [m].
        \end{align*}
    \end{definition}

    By linear algebra, when the rank of $A$ exceeds $n$, the system $\RankP_n^m(A)$ is unsatisfiable.

    To perform the above reduction, we define an iterated version of $\RankP^m_n(A)$. Let $L$ and $n$ be positive integers and $\Sigma = \{0, \ldots, L\}$ an alphabet of size $L+1$. We denote by $\Sigma^{\le n}$ the set of strings over $\Sigma$ of length at most $n$. Similarly, $\Sigma^{< n}$ denotes the set of all strings over $\Sigma$ of length strictly less than $n$. Let $K$ be a positive integer. For every string $\pi \in \Sigma^{\leq n}$, let $X^{(\pi)} = {(x_{i,k}^{(\pi)} )}_{i \in [(L+1)K], k \in [K]}$ be an $(L+1)K \times K$ variable matrix. Let $Y = {(y_{k,j})}_{k \in [K], j \in [(L+1)K]}$ be a $K \times (L+1)K$ variable matrix. For every string $\pi \in \Sigma^n$, let $A^{(\pi)} $ be an $(L+1)K \times (L+1)K$ constant matrix.

    \begin{definition}[Iterated Rank Principle~\cite{GGRT25}]\label{def: Iterations of Rank Principle}
        Let $L,n,K$ be positive integers, $\Sigma=\{0,\ldots,L\}$, and let $\{A^{(\pi)} : \pi \in \Sigma^n\}$ be a set of $(L+1)K \times (L+1)K$ constant matrices. We define $\IRankP_{L,n,K}(\{A^{(\pi)}\})$ as the system of the following polynomial axioms:
        \begin{align*}
            &\sum_{k \in [K]} x_{i,k}^{(\pi)} y_{k,j} - x_{i, j - \lfloor \frac{j-1}{K} \rfloor K}^{(\pi \ast \lfloor \frac{j-1}{K} \rfloor)}, & \forall \pi \in \Sigma^{< n}, i,j \in [(L+1)K],\\
            &\sum_{k \in [K]} x_{i,k}^{(\pi)} y_{k,j} - A_{i,j}^{(\pi)}, & \forall \pi \in \Sigma^n, i,j \in [(L+1)K].
        \end{align*}
        Namely, $\IRankP_{L,n,K}(\{A^{(\pi)}\})$ contains all the degree-$2$ polynomial equations in the following matrix multiplications:
        \begingroup\setlength{\arraycolsep}{2pt}
        \begin{align*}
             & X^{(\pi)} Y = \begin{bmatrix} X^{(\pi \ast 0)} & X^{(\pi\ast 1)} & \cdots & X^{(\pi \ast L)} \end{bmatrix}, & \forall \pi \in \Sigma^{< n}, \\
             & X^{(\pi)} Y = A^{(\pi)}, & \forall \pi \in \Sigma^n.
        \end{align*}\endgroup
    \end{definition}

    \begin{definition}[Iterated Rank Principle with extension variables~\cite{GGRT25}]\label{def: iterations of rank principle with extension variables}
        Let $L,n,K$ be positive integers, $\Sigma=\{0,\ldots,L\}$ and let $\{A^{(\pi)} : \pi \in \Sigma^n\}$ be a set of $(L+1)K \times (L+1)K$ constant matrices. We define $\IRankPE_{L,n,K}(\{A^{(\pi)}\})$ as the polynomial system consisting of the axioms of $\IRankP_{L,n,K}(\{A^{(\pi)}\})$ together with the following axioms:
        \begin{align*}
            &z_{i,j,k}^{(\pi)} - x_{i,k}^{(\pi)} y_{k,j}, & \forall \pi \in \Sigma^{\le n}, i,j \in [(L+1)K], k \in [K].
        \end{align*}
    \end{definition}

    We need these extension variables because the axioms of $\IRankP_{L,n,K}(\{A^{(\pi)}\})$ are inner products and thus require additional gates to be computed by an algebraic circuit. Namely, $z_{i,j,k}^{(\pi)}$ correspond to the $\times$ gates computing the products $x_{i,k}^{(\pi)} y_{k,j}$, which we later add by using a $+$ gate. The details of this construction are given in the proof of \cref{theorem: super-poly lower bound for AUP} below.

    The monomial size lower bound on $\IRankPE_{L,n,K}(\{A^{(\pi)}\})$ against $\PCR_{\F_2}$ is given in~\cite{GGRT25}.
    \begin{theorem}[\cite{GGRT25}]\label{theorem: exp lower bound for IRankP}
        Every $\PCR_{\F_2}$ refutation of $\IRankPE_{L, n ,K}(\{A^{(\pi)}\})$ requires $2^{K^{\Omega(1)}}$ monomials.
    \end{theorem}

    \begin{theorem}\label{theorem: super-poly lower bound for AUP}
        There exists a constant $c>0$ such that given a polynomial $f$ of degree at most $\ell$ over $\F_2$ in $n$ variables and $s \ge n^c$, every $\PCR_{\F_2}$ refutation of $\lb(f,s,\ell)$ requires superpolynomially many monomials.
    \end{theorem}

    \begin{proof}
    We prove this theorem by construction a reduction from $\IRankPE_{L, n ,K}(\{A^{(\pi)}\})$ to $\lb(f,s,\ell)$ for some $L$ and $K$ that are polynomial in $s$ and $n$. The reduction is done by substituting the variables of $\lb(f,s,\ell)$ with \emph{constants and  variables} of $\IRankPE_{L, n ,K}(\{A^{(\pi)}\})$, which does not increase the refutation size. Under this reduction, each axiom of $\lb(f,s,\ell)$ is either satisfied or becomes a simple consequence of the axioms of $\IRankPE_{L, n ,K}(\{A^{(\pi)}\})$.

    The high level idea of the construction is as follows: we explicitly define a small algebraic circuit $C$ that computes the given polynomial $f$. We proceed by embedding $C$ into $\lb(f,s,\ell)$ by fixing all the variables that are responsible for the topology of $C$ (i.e., the variables that encode the type and inputs of each gate in $C$). This embedding satisfies all the axioms of $\lb(f,s,\ell)$ except for those that describe the coefficient vectors of the gates in $C$. We show that they have simple derivations from the axioms of $\IRankPE_{L, n ,K}(\{A^{(\pi)}\})$.

    We construct a small circuit $C(\overline{u}) = C(u_1, \ldots, u_n)$ which purportedly computes $f$ and uses variables of $\IRankPE_{L, n ,K}(\{A^{(\pi)}\})$ as constants. Let $L\leq \ell$ be the degree of $f$ and let $\Sigma = \{0,\ldots,L\}$.
    The circuit uses gates $g_x[t,i,k]$ and $g_z[t,i,j,k,d]$, where $t \in \{0,\ldots,n\}$, $i,j \in [(L+1)K]$, $k \in [K]$, and $d \in \Sigma$. We also have a special output gate $\out$.
    For each $\pi \in \Sigma^{\leq n}$, we write $\overline{u}^{\pi}$ to denote the monomial $u_1^{\pi_1} u_2^{\pi_2} \cdots u_{|\pi|}^{\pi_{|\pi|}}$.

    Then for every $\pi \in \Sigma^{n}$, we define the matrix $A^{(\pi)}$ such that $A^{(\pi)}_{1,1} = \coeff_{\overline{u}^\pi(f(\overline{u}))}$ and is zero everywhere else.
    Assuming that $\IRankPE_{L, n ,K}(\{A^{(\pi)}\})$ is satisfied, we require that for each $\pi \in \Sigma^{n}$,
    \begin{equation}\label{equa: requirement for substitution}
        \coeff_{\overline{u}^{\pi}}(\widehat{C}_{\out}) = (X^{(\pi)} Y)_{1,1}.
    \end{equation}

    We define $g_z[t,i,j,k,d] = g_x[t, i, k] y_{k,j} u_{t+1}^d$ for every $i,j \in [(L+1)K], k \in [K]$, $d \in \Sigma$ and $t \in \{0,\ldots,n-1\}$ and $g_z[n,1,1,k,0] = g_x[n, 1, k] y_{k,1}$ (the other gates of the form $g_z[n,i,j,k,d]$ with $(i,j,d) \neq (1,1,0)$ are not used in the construction).
    The gates $g_x[t,i,k]$ are thus defined recursively on $t$, using $g_z$. For every $i \in [(L+1)K]$ and $k \in [K]$:
    \begin{align*}
        g_x[0,i,k] &= x^{()}_{i,k}, & \\
        g_x[t,i,k] &= \sum_{d \in \Sigma} \sum_{k' \in [K]} g_z[t-1, i, k+dK, k', d], & \forall t \in [n]. 
    \end{align*}
    Here, $()$ denotes the empty string.
    The output gate of $C(\overline{u})$ is
    \[
        \out = \sum_{k' \in [K]} g_z[n, 1, 1, k', 0]
    \].
    It is easy to see that $C(\overline{u})$ is an algebraic circuit of size at most $2 n {(L+1)}^3 K^3$.

    Now we verify that \eqref{equa: requirement for substitution} holds for this construction. Before doing this, we start by verifying that for each $t \in \{0,\ldots,n\}$, $i \in [(L+1)K]$ and $k \in [K]$,
    \begin{equation}\label{equa: requirement for substitution helper}
        \widehat{C}_{g_x[t,i,k]} = \sum_{\pi \in \Sigma^{\le t}} x^{(\pi)}_{i,k} \overline{u}^\pi,
    \end{equation}
    can be derived using the axioms of $\IRankPE_{L, n ,K}(\{A^\pi\})$.
    We prove the statement by induction on $t$. The base case holds by the definition of $g[0,i,k]$. The inductive step is as follows:
    \begingroup\allowdisplaybreaks
    \begin{flalign*}
        \widehat{C}_{g_x[t,i,k]} &= \sum_{d \in \Sigma} \sum_{k' \in [K]} g_z[t-1,i,k+dK,k',d]
        & \mathllap{\text{(by definition)}} \\
        &= \sum_{d \in \Sigma} \left( \sum_{k' \in [K]} g_x[t-1,i,k'] y_{k',k+dK} \right) u_t^{d}
        & \mathllap{\text{(expanding $g_z$)}} \\
        &= \sum_{d \in \Sigma} \left( \sum_{k' \in [K]} \left( \sum_{\pi \in \Sigma^{\le t-1}} x_{i, k'}^{(\pi)} \overline{u}^\pi \right) y_{k', k+dK}  \right) u_t^d
        & \mathllap{\text{(by the induction hypothesis)}} \\
        &= \sum_{d \in \Sigma} \left( \sum_{\pi \in \Sigma^{\le t-1}} \left( \sum_{k' \in [K]}  x_{i, k'}^{(\pi)}  y_{k', k+dK} \right) \overline{u}^\pi  \right) u_t^d
        & \mathllap{\text{(switching the summation order)}} \\
        &= \sum_{d \in \Sigma} \left( \sum_{\pi \in \Sigma^{\le t-1}} x_{i, k}^{(\pi \ast d)} \overline{u}^\pi  \right) u_t^d
        & \mathllap{\text{(by the axioms of $\IRankPE_{L, n ,K}(\{A^\pi\})$)}} \\
        &= \sum_{\pi \in \Sigma^{\le t}} x_{i, k}^{(\pi)} \overline{u}^\pi. & \mathllap{\text{(combining $\overline{u}^\pi$ and $u_t^d$)}}
    \end{flalign*}\endgroup

    Similarly, we show that
    \[
        \widehat{C}_{\out} = \sum_{\pi \in \Sigma^{\le n}} A^{(\pi)}_{1,1} \overline{u}^\pi
    \],
    which implies \eqref{equa: requirement for substitution} due to the definition of $A^{(\pi)}$.
    \begin{flalign*}
        \widehat{C}_{\out} &= \sum_{k' \in [K]} g_z[n,1,1,k',0]
        & \mathllap{\text{(by definition)}} \\
        &= \sum_{k' \in [K]} g_x[n,1,k'] y_{k',1}
        & \mathllap{\text{(expanding $g_z$)}} \\
        &= \sum_{k' \in [K]} \left( \sum_{\pi \in \Sigma^{\le n}} x_{1,k'}^{(\pi)} \overline{u}^\pi \right) y_{k', 1}
        & \mathllap{\text{(by \eqref{equa: requirement for substitution helper})}} \\
        &= \sum_{\pi \in \Sigma^{\le n}} \left( \sum_{k' \in [K]} x_{1,k'}^{(\pi)} y_{k', 1} \right) \overline{u}^\pi
        & \mathllap{\text{(switching the summation order)}} \\
        &= \sum_{\pi \in \Sigma^{\le n}} A^{(\pi)}_{1,1} \overline{u}^\pi. & \mathllap{\text{(by the axioms of $\IRankPE_{L, n ,K}(\{A^\pi\})$)}}
    \end{flalign*}

    Now, we present a substitution $\rho$, which embeds the circuit $C$ into $\lb(f,s,\ell)$. We begin by encoding the topology of the circuit by assigning all the $\Type$, $\InputType$, $\InputType'$, $\InputType''$, $\InputVar$ and $\InputNode$ variables to the appropriate constants.
    Slightly abusing the notation, we proceed by setting for each $t \in \{0,\ldots,n\}, i \in [(L+1)K], k \in [K], \pi \in \Sigma^{\le t}$, $d \in \Sigma$,
    \[
        \rho(\alpha_{\overline{u}^\pi, g_x[t,i,k]})=x^{(\pi)}_{i,k}
    \]
    and $\rho(\alpha_{M, g_x[t,i,k]})=0$ for the remaining monomials $M$. The values of $g_z[t,i,j,k]$ are defined similarly.
    Then the inputs of $g_x[t,i,k]$ are defined as the appropriate gates $g_z[t-1,i,k+dK,k',d]$ to satisfy the definition of $C$. The inputs of $g_z[t,i,j,k,d]$ are similarly defined as the gate $g_x[t,i,k]$, followed by $d$ variables $u_{t+1}$ and a constant $y_{k,j}$.
    
    Observe that such a substitution replaces each variable with either a constant or another variable (possibly a negation). This preserves the $\PCR$ size of the refutation.
    It follows from the definition that the axioms of $\lb(f,s,\ell) \restriction \rho$ can be easily derived from the axioms of $\IRankPE_{L,n,K}(\{A^{(\pi)}\})$.
    
    Now suppose that $\lb(f, s, \ell))$ has a $\PCR$ refutation of polynomial size, specifically,
    \[
        \poly(s, n, N^\ell) \le \exp({(\ell n)}^{O(1)})
    \].
    Then, $\IRankPE_{L,n,K}(\{A_\pi\})$ has a $\PCR$ refutation of size $\exp({(\ell n)}^{O(1)})$. By taking $K$ to be sufficiently large (say, ${(\ell n)}^{\Omega(1)}$), this contradicts \cref{theorem: exp lower bound for IRankP} since such a refutation must have size $\exp(K)$.
    \end{proof}

    Similarly to $\lb(f,s,\ell)$, we can define $\lb(f,s,\ell,\Delta)$ stating that an $n$-variable polynomial $f$ has a depth-$\Delta$ algebraic circuit of size $s$ and syntactic-degree $\ell$.
    We can apply a similar idea to get a lower bound on the size of $\PCR$ refutations of $\lb(f,s,\ell,\Delta)$ provided the rank principle does not hold.

    \begin{theorem}\label{theorem: super-poly lower bound for AUP constant-depth}
        For every constant $c>0$, given a polynomial $f$ of degree at most $\ell$ over $\F_2$ in $n$ variables and $s \ge N^{1/c}$, every $\PCR_{\F_2}$ refutation of $\lb(f,s,\ell,O(c))$ requires superpolynomially many monomials.
    \end{theorem}
    \begin{proof}[Proof sketch]
        We construct the circuit $C$ similarly to the one from the proof of \cref{theorem: super-poly lower bound for AUP} but instead of multiplying the $\overline{u}$-variables one-by-one, we multiply them by all monomials of degree at most $\ell$ in the subsequent $n/c$ variables. This increases the size of the circuit since we are forced to choose $L$ to be the number of such monomials, which is approximately $N^{1/c}$.
        However, this way, $C$ only requires depth $O(c)$.
    \end{proof}

    \begin{remark}
        Observe that the above circuit constructions do not use the whole matrices $A^{(\pi)}$, but only their  $(1,1)$th entry. We can construct some slightly more efficient circuits by using more entries of $A^{(\pi)}$ to store the coefficients by computing the polynomials
        \[
            \sum_{\pi \in \Sigma^{\le n-2}, d, d' \in \Sigma} A^{(\pi)}_{d,d'} u^{(\pi)} u_{n-1}^d u_n^{d'}.
        \]
        This can be done by introducing gates $\alpha[k]$ computing the polynomials
        \[
            \sum_{d \in \Sigma} \sum_{\pi \in \Sigma^{\le n-2}} x^{(\pi)}_{d,k} u^{(\pi)} u_{n-1}^d
        \]
        and gates $\beta[d'+dK]$ computing $y_{d'+dK} u_n^{d'}$.
        Then, the output gate can be defined as the inner product of $\alpha$ and $\beta$.
        
        A similar idea can be applied to the proof for $\lb(f,s,\ell,O(c))$.
    \end{remark}

\begin{remark}
        Note that we can also view $\lb(f,s,\ell)$ as a \textbf{proof complexity generator} in the algebraic setting: treating the coefficient vector of $f$ as an output vector $\overline{b} \in \mathbb{F}^N$ (where $N$ is the number of monomials of degree at most $\ell$ over $n$ variables) of a function $\tau(s,\ell)$, we can view $\lb(f,s,\ell)$ as an encoding of the statement $\tau(s,\ell)=\overline{b}$.  Thus, \cref{theorem: super-poly lower bound for AUP} shows that $\lb(f,s,\ell)$ is a proof complexity generator against $\PCR$.

        Similarly, $\lb(f,s,\ell,\Delta)$ can be viewed as $\tau(s,\ell,\Delta)=\overline{b}$, where $\overline{b} \in \mathbb{F}^N$ and $N$ is the number of monomials of degree at most $\ell$ over $n$ variables, and is a proof complexity generator against $\PCR$.
\end{remark}

%

\section{Hardness of Bounded-Depth Algebraic Circuit Lower Bounds from Hardness of Weak Tensor Rank Principles}\label{sec: tensor rank}

In this section, we need a generalisation of the rank principle, which we call the \emph{tensor rank principle}.

\begin{definition}[Tensor Rank Principle]\label{def: tensor rank principle}
    Let $\F$ be a field and $m, n$ two positive integers with $m > n$. We denote $\TRankP^m_n(r, A)$ the system of degree-$r$ polynomial equations stating that the tensor rank of an $r$-tensor $A \in \F^{\overbrace{m \times \cdots \times m}^{r \text{ times}}}$ is at most $n$ over $\F$. More precisely, $\TRankP^m_n(r, A)$ is defined over $rmn$ variables arranged into
    $rn$ vectors $\x_{j,k} \in \F^m : j \in [r], k \in [n]$.
    
    For every $i_1, \ldots, i_r \in [m]$ (which are not necessarily distinct), there is an equation in $\TRankP^m_n(r,A)$ stating that the $(i_1, \ldots, i_r)$th entry in the summation of the tensor product $\bigotimes_{j=1}^r \x_{j,k}$ is equal to $A_{i_1, \ldots, i_r}$. That is, $\TRankP^m_n(r, A)$ consists of the polynomials
    \begin{align}\label{equa: tensor rank principle}
        &\sum_{k \in [n]} \prod_{j \in [r]} x_{j,k,i_j} - A_{i_1, \ldots, i_r}, & i_1, \ldots, i_r \in [m],
    \end{align}
    where $x_{j,k,i_j}$ denotes the $(i_j)$th entry of the vector $\x_{j,k}$.
    
    Namely, $\TRankP^m_n(r, A)$ is the set of polynomial equations stating that
    \[
        \sum_{k \in [n]} \bigotimes_{j=1}^r \x_{j,k} = A.
    \]
\end{definition}

By definition, $\sum_{k \in [n]} \bigotimes_{j=1}^r \x_{j,k}$ is the tensor rank decomposition of rank at most $n$. Thus, when the tensor rank of $A$ exceeds $n$, $\TRankP^m_n(r, A)$ is unsatisfiable. This principle can be viewed as a generalisation of the rank principle, and coincides with it when $r=2$.

\begin{theorem}\label{theorem: exp lower bound for TRankP}
    Let $m,n,r$ be positive integers with $m>n$.
    Let $J \subseteq [m]$ with $|J| > n$ and let $A$ be the $r$-tensor defined as
    \[
        A_{i_1,\ldots,i_r} = \begin{cases}
            1 & \text{if } i_1 = \cdots = i_r \in J, \\
            0 & \text{otherwise}.
        \end{cases}
    \]
    Then every $\PCR_{\F_2}$ refutation of $\TRankP^m_n(r, A)$ requires $2^{\Omega(n)}$ monomials.
\end{theorem}
\begin{proof}
    Without loss of generality assume that $A = I_m$ (and thus $J = [m]$): we can set all $x_{j,k,i}$ to $0$ for all $j \in [r]$, $k \in [n]$ and $i \notin J$. This reduces $\TRankP^m_n(r, A)$ to $\TRankP^{|J|}_n(r, I_{|J|})$.

    The idea of the proof is analogous to the proof of the lower bound for $\RankP$ in~\cite{GGRT25}.

    We say that a term $t$ \emph{mentions} an index $k \in [n]$ if there exist $j \in [r]$ and $i \in [m]$ such that $x_{j,k,i}$ or $\overline{x}_{j,k,i}$ appears in $t$.
    We start by proving a lower bound on the number of different indices $k \in [n]$ each term mentions by reducing $\PHP^m_n$ to $\TRankP^m_n(r, A)$.
    The axioms of $\PHP^m_n$ are the \emph{pigeon axioms}
    \begin{align*}
        &\sum_{k \in [n]} x_{i,k} - 1, & \forall i \in [m],
    \end{align*}
    and the \emph{hole axioms}
    \begin{align*}
        &x_{i,k}x_{j,k}, & \forall i \neq j \in [m], k \in [n].
    \end{align*}
    It follows from~\cite{Razborov98,IPS99} that any $\PCR_{\F_2}$ refutation of $\PHP^m_n$ must have a term that mentions at least $n/2+1$ different holes $k \in [n]$ (in fact, it must contain a term that correspond to a \emph{matching} of this size).
    
    Let $X = {(x_{i,k})}_{i \in [m], k \in [n]}$ be fresh variables and define the restriction $\rho$ as
    \begin{align*}
        &\rho(x_{j,k,i}) = x_{i,k}, & \forall j \in [d], k\in[n], i \in [m].
    \end{align*}

    Under $\rho$, the axioms of $\TRankP^m_n(r, A)$ become\footnote{Recall that proof lines are \emph{multilinear} polynomials.}
    \begin{equation}\label{equa: trank to php diagonal}
        \sum_{k\in[n]} x_{i,k} - 1,
    \end{equation}
    when $i_1 = \cdots = i_r = i$ and
    \begin{equation}\label{equa: trank to php non-diagonal}
        \sum_{k\in[n]} \prod_{j \in [r]} x_{i_j,k},
    \end{equation}
    where at least two $i_j$ and $i_{j'}$ are different.
    Observe that \eqref{equa: trank to php diagonal} is precisely the pigeon axiom of $\PHP^m_n$ and \eqref{equa: trank to php non-diagonal} can be derived easily from the hole axioms $x_{i_j,k}x_{i_{j'},k}$.
    Thus, every $\PCR_{\F_2}$ refutation of $\TRankP^m_n(r, A)$ must contain a term involving at least $n/2+1$ variables $x_{j,k,i}$ (or their negations) with different $k$.

    Now let $\sigma$ be a partial restriction defined as follows. Choose a random subset $K \subseteq [n-1]$ of size $n/2-1$ uniformly and for each $k \in K$ independently choose $c_k \in \{0,1\}$ at random. Define $c_n = 1$ if $\sum_{k\in K}c_k$ is odd and $c_n=0$ otherwise. This guarantees that $\sum_{k\in K \cup \{n\}}c_k$ is even. Then assign
    \[
        \sigma(x_{j,k,i}) = \begin{cases}
            c_k & \text{if } k \in K \cup \{n\}, \\
            * & \text{if } k \notin K \cup \{n\}.
        \end{cases}
    \]

    Under $\sigma$, $\TRankP^m_n(r, A) \restriction \sigma$ becomes $\TRankP^m_{n/2}(r, A)$ (up to the variable renaming). Moreover, each term that mentions at least $n/2+1$ different $k \in [n]$, becomes zero with probability $1-\exp(-\Omega(n))$ by the Chernoff bound. Applying the union bound, we can eliminate all such terms in a proof of size at most $\exp(\epsilon n)$ for a sufficiently small constant $\epsilon>0$, contradicting the above degree bound.
\end{proof}

\begin{theorem}\label{them: conditional IPS Lower Bounds reduction}
    Let $n\geq 1$ and let $f \in \F_2[u_1,\ldots,u_n]$ have degree at most $\ell\geq 2$. Set
    \[
        D=n\ell,\qquad
        N=\binom{n+\ell}{\ell},\qquad
        R=\binom{n+D}{D},\qquad
        K=\lfloor\sqrt{s}\rfloor.
    \]
    Assume $K<N$, $\Delta>4$, and $s\geq 4n^2(\ell+1)^4$.

    Let $A$ be the diagonal order-$n$ tensor of dimension $N$ whose $(M,\ldots,M)$th entry is
    $\coeff_M(f)$, with all other entries zero.

    If $\lb(f,s,D,\Delta)$ admits a $\Delta'$-$\IPS$ refutation of size $S$ over $\F_2$, then $\TRankP^N_K(n,A)$ admits a $(\Delta'+5)$-$\IPS$ refutation over $\F_2$ of size
    \[
        S+\poly(s,n,\ell,K,N^n,R^D).
    \]
\end{theorem}

The proof of this theorem is analogous to the proof of \cref{theorem: super-poly lower bound for AUP}. We construct a circuit $C$ parametrised by the variables of $\TRankP^N_K(n,A)$, with depth at most $4$, size at most $s$, and syntactic degree at most $D=n\ell$. We then embed this circuit into $\lb(f,s,D,\Delta)$ and construct bounded-depth $\IPS$ derivations of the substituted circuit axioms from the axioms of $\TRankP^N_K(n,A)$. Consequently, a short $\IPS$ refutation of $\lb(f,s,D,\Delta)$ yields a short $\IPS$ refutation of $\TRankP^N_K(n,A)$.

\begin{proof}
    We show how to reduce the bounded-depth algebraic circuit upper bound formulas to $\TRankP^{N}_{K}(n, A)$, where $T_f \coloneqq \{ M : M \text{ is a monomial in $f$ with a non-zero coefficient} \}$ and $K$ is bounded by $\sqrt{s}$.

    We use the notation $\deg_{u_j}(M)$ to denote the individual degree of $u_j$ in the monomial $M$.

    Consider the depth-$4$ algebraic circuit
    \[
        C(u_1, \ldots, u_n) \coloneqq \sum_{k \in [K]} \prod_{j \in [n]} \sum_{d \in \{0,1,\ldots,\ell\}} \alpha_{j,k,d} u_j^d
    \]
    where $\alpha_{j,k,d}$ are placeholder variables defined as
    \[
        \alpha_{j,k,d} \coloneqq \sum_{\substack{M \in T_f\\ \deg_{u_j}(M)=d}} x_{j,k,M}.
    \]
    The syntactic degree of $C$ is at most $D=n\ell$.

    For every monomial $M=u_1^{d_1}\cdots u_n^{d_n}$ of degree at most $D$ with $d_j\leq\ell$ for all $j$, let $X_{k,M}$ denote the sum of all off-diagonal terms in the expansion of $\prod_{j\in[n]}\alpha_{j,k,d_j}$. More precisely,
    \[
        X_{k,M} \coloneqq \sum_{\substack{M_1,\ldots,M_n\in T_f\\
                                \deg_{u_j}(M_j)=d_j\ \forall j\\
                                (M_1,\ldots,M_n)\text{ not diagonal}}}
                        \prod_{j\in[n]}x_{j,k,M_j}.
    \]
    The off-diagonal axioms of $\TRankP^N_K(n,A)$ imply
    \[
        \sum_{k\in[K]}X_{k,M}\equiv 0,
    \]
    where $\equiv$ denotes congruence modulo the ideal generated by the tensor-rank axioms.

    For every $M=u_1^{d_1}\cdots u_n^{d_n}\in T_f$, we have
    \begingroup\allowdisplaybreaks
    \begin{flalign*}
        \coeff_M(\widehat{C}) &= \coeff_M\left(\sum_{k\in[K]}\prod_{j\in[n]} \sum_{d=0}^{\ell}\alpha_{j,k,d}u_j^d \right)
        &\mathllap{\text{(by definition)}}\\
        &= \sum_{k\in[K]}\coeff_M\left(\prod_{j\in[n]}\sum_{d=0}^{\ell}\alpha_{j,k,d}u_j^d\right)
        &\mathllap{\text{(since coefficients add termwise)}}\\
        &= \sum_{k\in[K]}\prod_{j\in[n]}\alpha_{j,k,d_j}
        &\mathllap{\text{(computing the coefficient for $M$)}}\\
        &= \sum_{k\in[K]} \left(\prod_{j\in[n]}x_{j,k,M}+X_{k,M}\right)
        &\mathllap{\text{(expanding the products)}}\\
        &\equiv \sum_{k\in[K]}\prod_{j\in[n]}x_{j,k,M}
        &\mathllap{\text{(by the off-diagonal tensor axioms)}}\\
        &\equiv A_{M,\ldots,M}
        &\mathllap{\text{(by the diagonal tensor axiom)}}\\
        &= \coeff_M(f).
        &\mathllap{\text{(by the definition of $A$)}}
    \end{flalign*}
    \endgroup

    For a monomial $M\notin T_f$ of degree at most $D$ with $d_j\leq\ell$ for all $j$, we have
    \begin{flalign*}
        \coeff_M(\widehat{C}) &= \sum_{k\in[K]}\coeff_M \left(\prod_{j\in[n]} \sum_{d=0}^{\ell}\alpha_{j,k,d}u_j^d \right)
        &\mathllap{\text{(since coefficients add termwise)}}\\
        &= \sum_{k\in[K]}X_{k,M}
        &\mathllap{\text{(computing the coefficient for $M$)}}\\
        &\equiv 0
        &\mathllap{\text{(by the off-diagonal tensor axioms)}}\\
        &= \coeff_M(f).
    \end{flalign*}
    If $d_j>\ell$ for some $j$, then $\coeff_M(\widehat{C})=\coeff_M(f)=0$ identically.

    Consequently, for every monomial $M$ of degree at most $D=n\ell$, the output-coefficient equation
    \[
        \coeff_M(\widehat{C})-\coeff_M(f)=0
    \]
    has an IPS derivation from the axioms of $\TRankP^N_K(n,A)$.

    We construct a substitution $\rho$ that embeds $C$ into $\lb(f,s,n\ell,\Delta)$. Note that $C$ is a depth-$4$ algebraic circuit having size at most $2Kn{(\ell+1)}^2$ and syntactic degree at most $n\ell$. The formal variables $\alpha_{j,k,d}$ are replaced with $\sum_{\substack{M \in T_f, \\ \deg_{u_j}(M) = d}} x_{j,k,M}$ of size at most $N$.

    Under this substitution, every axiom in $\lb(f,s,n\ell,\Delta)$, where $s \ge 2K n(\ell+1)^2$, is transformed into a polynomial derivable from $\TRankP^{N}_{K}(n,A)$ with total size at most
    \[
        \poly(s,n,\ell,K,N^n,R^{n\ell}), \qquad R=\binom{n+n\ell}{n\ell},
    \]
    and constant depth. In particular, the high-degree multiplication axioms vanish identically, and the output-coefficient axioms follow from the tensor-rank axioms. Consequently, if $\lb(f,s,n\ell,\Delta)$ admits a $\Delta'$-$\IPS$ refutation of size $S$, then $\TRankP^{N}_{K}(n,A)$ admits a $(\Delta'+5)$-$\IPS$ refutation of size
    \[
        S+\poly(s,n,\ell,K,N^n,R^{n\ell}).
    \]
    Finally, since $s\geq 4n^2(\ell+1)^4$, we may set $K=\lfloor\sqrt{s}\rfloor$, which guarantees $2Kn(\ell+1)^2\leq s$.
\end{proof}

\begin{corollary}\label{corollary: conditional IPS lower bounds}
    Consider a family of instances satisfying the hypotheses of \cref{them: conditional IPS Lower Bounds reduction}, with $K\to\infty$. Let
    \[
        D=n\ell,\qquad R=\binom{n+D}{D}.
    \]
    Suppose that, for some fixed $\varepsilon>0$ and constant $d>5$, the corresponding systems $\TRankP^N_K(n,A)$ require $d$-$\IPS$ refutations over $\F_2$ of size at least $2^{K^\varepsilon}$, and that
    \[
        \log(sn)+n\log N+D\log R=o(K^\varepsilon).
    \]
    Then $\lb(f,s,D,\Delta)$ requires $(d-5)$-$\IPS$ refutations over $\F_2$ of size $|\lb(f,s,D,\Delta)|^{\omega(1)}$.
\end{corollary}

\begin{proof}
    The encoding has size polynomial in $s,n,R^D$. Under the stated growth condition, both every fixed polynomial in this encoding size and the additive overhead in \cref{them: conditional IPS Lower Bounds reduction} are $2^{o(K^\varepsilon)}$. A polynomial-size $(d-5)$-$\IPS$ refutation would therefore give a $d$-$\IPS$ refutation of the tensor system of size less than $2^{K^\varepsilon}$, a contradiction.
\end{proof}

\begin{remark}
    The tensor $A$ in \cref{them: conditional IPS Lower Bounds reduction} is diagonal. Over $\F_2$, its restriction to the coordinates in the support of $f$ is a unit tensor. The assumption in \cref{corollary: conditional IPS lower bounds} concerns these diagonal tensor rank instances. For order $2$, the tensor rank principle coincides with the matrix rank principle. Establishing constant-depth $\IPS$ lower bounds for the diagonal tensor rank instances remains the hypothesis of our conditional result.
\end{remark}

\section{Short \texorpdfstring{$\NCTwo$}{NC2}-Frege Proofs for Constant-Depth Algebraic Circuit Lower Bounds}\label{sec: VNCtwo}

In this section, we show that \NCTwo-Frege (see~\cite{HT15} for the definition) admits short proofs of the constant-depth algebraic circuit lower bounds for the Permanent and Determinant established in~\cite{LST25,For24}.

For this purpose, we use a formalisation in $\VNCTwo$, which is a first-order theory whose (low logical complexity) theorems translate to polynomial-size proofs in \NCTwo-Frege. We shall demonstrate a formalisation of the proof of the superpolynomial lower bound of the Permanent polynomial against constant-depth algebraic circuits over any field. This is an adaptation of the proof of~\cite{LST25} and~\cite{For24}. The work in~\cite{LST25} demonstrated a superpolynomial lower bound against constant-depth algebraic circuits over fields of large characteristic. Forbes extended this result to any field by replacing the original set-multilinearization argument in~\cite{LST25}, which works only over fields of large characteristic, with a new set-multilinearization argument over any field.

In the previous sections, we presented lower-bound (or unprovability) results for constant-depth algebraic circuit lower bound statements. In this section, we show an \emph{upper bound} on proving constant-depth algebraic circuit lower bounds.


\subsection{Preliminaries for Bounded Arithmetic}

For the convenience of the reader, we provide nearly full details of the theory \VNCTwo\ in which we work. We refer the reader to~\cite{CN10} for more information.

\subsubsection{The theory \texorpdfstring{$\Vzero$}{V0}}

\begin{definition}[Language of two-sorted arithmetic $\calL_A^2$]
    The language of two-sorted arithmetic, denoted $\calL_A^2$, consists of the following relation, function and constant symbols:
    \[
        \{+, \cdot, \leq ,0 ,1, |\cdot|, =_1, =_2, \in\}.
    \]
\end{definition}

A number quantifier is said to be \emph{bounded} if it is of the form $\exists x (x \leq t \land \cdots)$ or $\forall x (x \leq t \to \cdots)$ for some \emph{number term} $t$ that does not contain $x$. We abbreviate $\exists x(x \leq t \land \cdots)$ and $\forall x (x \leq t \to \cdots)$ by $\exists x \leq t$ and $\forall x \leq t$, respectively. A string quantifier is said to be \emph{bounded} if it is of the form $\exists X (|X| \leq t \land \cdots)$ or $\forall X (|X| \leq t \to \cdots)$ for some \emph{number term} $t$ that does not contain $X$. We abbreviate $\exists X (|X| \leq t \land \cdots)$ and $\forall X (|X| \leq t \to \cdots)$ by $\exists X \leq t$ and $\forall X \leq t$, respectively.

A formula is in the class of formulas $\Sigma_0^B$ or $\Pi_0^B$ if it uses \emph{no string quantifiers} and all number quantifiers are bounded. A formula is in $\Sigma_{i+1}^B$ or $\Pi_{i+1}^B$ if it is of the form $\exists X_1 \leq t_1 \cdots \exists X_m \leq t_m \psi$ or $\forall X_1 \leq t_1 \cdots \forall X_m \leq t_m \psi$, where $\psi \in \Pi_i^B$ and $\psi \in \Sigma_i^B$, respectively, and none of $X_1, \dots, X_m$ occurs in any of $t_1, \dots, t_m$.

A $\Sigma_1^1(\calL)$-formula is one of the form $\exists \Vec{X} \varphi$, where $\Vec{X}$ is a vector of zero or more string variables, and $\varphi$ is a $\Sigma_0^B(\calL)$-formula.

We write $T(t)$ to abbreviate $t \in T$, for a number term $t$ and a string term $T$. For a formula $\psi$, we write $\psi(a/ x)$ to denote the \emph{substitution instance} of $\psi$ in which every free occurrence of the variable $x$ is replaced by the term $a$.

A finite set $N$ of natural numbers represents the finite binary string $S_N=S_N^0\ldots S_N^{|N|-1}$ such that $S_N^i =1$ if and only if $i \in N$.

\begin{definition}[Comprehension Axiom]
    If $\Phi$ is a set of formulas, then the \emph{comprehension axiom scheme} for $\Phi$, denoted by $\Phi\COMP$, is the set of all formulas
    \[
        \exists X \leq y \ \forall z <y \ (z \in X \leftrightarrow \varphi(z)),
    \]
    where $\varphi$ is a formula in $\Phi$ and $X$ does not occur free in $\varphi$.
\end{definition}

\begin{definition}[$\Vzero$]
    The theory $\Vzero$ has the vocabulary $\calL_A^2$ and is axiomatized by $2\text{-}\textsc{Basic}$ and $\Sigma_0^B\COMP$. The axioms of $2\text{-}\textsc{Basic}$ are listed in \cref{tbl: two-basic axioms}.
\end{definition}

\begin{table}[ht]
\centering
\begin{tcolorbox}[tableborder]
    \begingroup\renewcommand{\arraystretch}{1.15}\setlength{\tabcolsep}{0pt}
    \begin{tabularx}{\textwidth}{@{} X X @{}}
      $\mathbf{B1}. \; x+1\neq 0$ & $\mathbf{B7}. \; (x\le y \wedge y\le x)\supset x=y$\\
      $\mathbf{B2}. \; x+1=y+1 \supset x=y$ & $\mathbf{B8}. \; x\le x+y$\\
      $\mathbf{B3}. \; x+0=x$ & $\mathbf{B9}. \; 0\le x$\\
      $\mathbf{B4}. \; x+(y+1)=(x+y)+1$ & $\mathbf{B10}. \; x\le y \lor y\le x$\\
      $\mathbf{B5}. \; x\cdot 0=0$ & $\mathbf{B11}. \; x\le y \leftrightarrow x<y+1$\\
      $\mathbf{B6}. \; x\cdot (y+1)=(x\cdot y)+x$ & $\mathbf{B12}. \; x\neq 0 \supset \exists y\le x\, (y+1=x)$\\[0.25em]
      $\mathbf{L1}. \; X(y)\supset y<|X|$ & $\mathbf{L2}. \; y+1=|X|\supset X(y)$\\[0.25em]
      \multicolumn{2}{l}{$\mathbf{SE}. \; (\,|X|=|Y|\ \wedge\ \forall i<|X|\, (\,X(i)\leftrightarrow Y(i)\,)\,)\ \supset\ X=Y$}
    \end{tabularx}\endgroup%
\end{tcolorbox}
\caption{Axioms of $2\text{-}\textsc{Basic}$.}\label{tbl: two-basic axioms}
\Description{The axiomatization of 2-Basic.}
\end{table}

\begin{definition}[Number Induction Axiom]
    \label{def: number induction axiom}
    If $\Phi$ is a set of two-sorted formulas, then $\Phi\IND$ axioms are the formulas
    \[
        \varphi(0) \land \forall x \ (\varphi(x) \to \varphi(x+1)) \to \forall z \varphi(z)
    \]
    where $\varphi$ is a formula in $\Phi$.
\end{definition}

\begin{definition}[Number Minimization and Maximization Axioms]
    \label{def: number minimiazation and maximization axioms}
    The number minimization axioms (or \emph{least number principle} axioms) for a set $\Phi$ of two-sorted formulas are denoted $\Phi\MIN$ and consist of the formulas
    \[
        \varphi(y) \to \exists x  \leq y (\varphi(x) \land \neg \exists z < x \varphi(z))
    \]
    where $\varphi$ is a formula in $\Phi$. Similarly, the number maximization axioms for $\Phi$ are denoted $\Phi\MAX$ and consist of the formulas
    \[
        \varphi(0) \to \exists x \leq y (\varphi(x) \land \neg \exists z \leq y(x < z \land \varphi(z)))
    \]
    where $\varphi$ is a formula in $\Phi$.
\end{definition}

\begin{lemma}
    \label{lemma: comp = ind = min = max}
    Let $\calT$ be an extension of \,$\Vzero$ and $\Phi$ be a set of formulas in $\calT$. Suppose that $\calT$ proves the $\Phi\COMP$ axiom scheme. Then $\calT$ also proves the $\Phi\IND$ axiom scheme, the $\Phi\MIN$ axiom scheme, and the $\Phi\MAX$ axiom scheme.
\end{lemma}

\begin{proposition}
    The theory $\Vzero$ proves $\Sigma_0^B\IND$, $\Sigma_0^B\MIN$, and $\Sigma_0^B\MAX$.
\end{proposition}

\begin{definition}[Two-Sorted Polynomially Bounded Functions]
    \label{def: two-sorted polynomially bounded functions}
    A \emph{number function} is a two-sorted function whose output belongs to the number sort. A \emph{string function} is a two-sorted function whose output belongs to the string sort.
    
    A number function $f$ is \emph{polynomially bounded} (or \emph{p-bounded}) if there is a polynomial $p$ such that $f(\Vec{x},\Vec{Y})\leq p(\Vec{x},|\Vec{Y}|)$. A string function $F$ is p-bounded if $|F(\Vec{x},\Vec{Y})|\leq p(\Vec{x},|\Vec{Y}|)$ for some polynomial $p$.
\end{definition}

\begin{definition}[Graph, Bit Graph]
    \label{def: graph, bit graph}
    The \emph{graph} $G_f$ of a number function $f(\Vec{x}, \Vec{Y})$ is defined by
    \[
        G_f(z, \Vec{x}, \Vec{Y})  \leftrightarrow z = f(\Vec{x}, \Vec{Y}).
    \]
    The \emph{bit graph} $B_F$ of a string function $F(\Vec{x}, \Vec{Y})$ is defined by
    \[
        B_F(i, \Vec{x}, \Vec{Y}) \leftrightarrow F(\Vec{x}, \Vec{Y})(i).
    \]
\end{definition}

Later, we shall introduce bit-definability of new string function symbols based on the notion of bit graph.

\begin{definition}[Function class]
    If $C$ is a two-sorted complexity class of relations, then the corresponding function class $FC$ consists of all p-bounded number functions whose graphs are in $C$, together with all p-bounded string functions whose bit graphs are in $C$.
\end{definition}

Saying that a graph $G_f$ (respectively, a bit graph $B_F$) is in $C$ means that there exists a family of circuits in $C$ which, given $z$ (respectively, $i$), the tuple $\Vec{x}$ in unary, and the strings $\Vec{Y}$, decides whether $(z,\Vec{x},\Vec{Y})$ is in $G_f$ (respectively, $(i,\Vec{x},\Vec{Y})$ is in $B_F$).

Note that $C$ is a uniform complexity class.

\begin{definition}[$\FACZ$]
    A string (respectively, number) function is in $\FACZ$ if it is polynomially bounded and its bit graph (respectively, graph) is definable by a $\Sigma_0^B$-formula in the language $\calL_A^2$.
\end{definition}

\begin{definition}[Two-sorted definability in a theory $\calT$]
    Let $\calT$ be a theory in a language $\calL \supseteq \calL_A^2$ and let $\Phi$ be a set of formulas in the language $\calL$. A \emph{number function} $f$ is $\Phi$-definable in the theory $\calT$ iff there is a formula $\varphi(\Vec{x}, y ,\Vec{X})$ in $\Phi$ such that
    \begin{equation}
        \calT \vdash \forall \Vec{x} \forall \Vec{X} \exists! y \varphi(\Vec{x}, y, \Vec{X}),
    \end{equation}
    and it holds in the standard two-sorted model $\calN_2$ (whose number sort is $\mathbb{N}$ and the set sort is the set of finite subsets of $\mathbb{N}$) that
    \begin{equation}
        \label{equa: defining axioms for number functions in two-sorted theory}
        y = f(\Vec{x}, \Vec{X}) \leftrightarrow \varphi(\Vec{x}, y , \Vec{X}).
    \end{equation}
    A \emph{string function} $F$ is $\Phi$-definable in the theory $\calT$ iff there is a formula $\varphi(\Vec{x}, \Vec{X}, Y)$ in $\Phi$ such that
    \begin{equation}
        \calT \vdash \forall \Vec{x} \forall \Vec{X} \exists! Y \varphi(\Vec{x}, \Vec{X}, Y)
    \end{equation}
    and it holds in the standard two-sorted model $\calN_2$ that
    \begin{equation}
        \label{equa: defining axioms for string functions in two-sorted theory}
        Y = F(\Vec{x}, \Vec{X}) \leftrightarrow \varphi(\Vec{x},\Vec{X}, Y).
    \end{equation}
    The formulas \eqref{equa: defining axioms for number functions in two-sorted theory} and \eqref{equa: defining axioms for string functions in two-sorted theory} are the \emph{defining axioms} for $f$ and $F$, respectively.
\end{definition}

For string functions, $\varphi(\Vec{x},\Vec{X},Y)$ defines the entire output string: it holds exactly when $Y=F(\Vec{x},\Vec{X})$. A bit definition instead specifies the individual output bits.

\begin{definition}[Conservative extension of a theory]
    Let $\calT$ be a theory in a language $\calL$. We say that a theory $\calT^\prime \supseteq \calT$ in a language $\calL^\prime \supseteq \calL$ is \emph{conservative} over $\calT$ if every $\calL$ formula provable in $\calT^\prime$ is also provable in $\calT$.
\end{definition}

We may expand the language $\calL$ and a theory $\calT$ over $\calL$. If $f$ is definable in $\calT$, in particular if $\calT$ proves existence and uniqueness of the value satisfying the defining formula, then adjoining a symbol for $f$ together with its defining axiom gives a conservative extension. The resulting extension, denoted $\calT(f)$ for a function $f$ (and similarly for relations), is a conservative extension of $\calT$. However, such an extension \emph{does not guarantee} that one can use the new function symbols in the \emph{comprehension} (and \emph{induction}) axiom schemes.

\begin{definition}[Provably Total Function]
    A function (either a number function or a string function) is said to be \emph{provably total} in a theory $\calT$ iff it is $\Sigma_1^1$-definable in $\calT$.
\end{definition}

We now state an important corollary, which follows from the two-sorted version of Parikh's theorem. 

\begin{definition}[Polynomially Bounded Theory]
    Let $\calT$ be a two-sorted theory over a vocabulary $\calL \supseteq \calL_A^2$. A number function symbol $f(\Vec{x},\Vec{X})$ in $\calL$ is said to be \emph{polynomially bounded in $\calT$} if there is an $\calL_A^2$-number term $t(\Vec{x},|\Vec{X}|)$ such that
    \[
        \calT \vdash f(\Vec{x},\Vec{X}) \leq t(\Vec{x},|\Vec{X}|).
    \]
    Similarly, a string function symbol $F(\Vec{x},\Vec{X})$ in $\calL$ is said to be \emph{polynomially bounded in $\calT$} if there is an $\calL_A^2$-number term $t(\Vec{x},|\Vec{X}|)$ such that
    \[
        \calT \vdash |F(\Vec{x},\Vec{X})| \leq  t(\Vec{x},|\Vec{X}|).
    \]
    We say that $\calT$ is a \emph{polynomially bounded theory} if:
    \begin{enumerate}
        \item $\calT$ extends $\Vzero$;
        \item $\calT$ can be axiomatized by a set of bounded formulas; and
        \item every number and string function symbol in the vocabulary $\calL$ is polynomially bounded in $\calT$.
    \end{enumerate}
\end{definition}

The following is a consequence of the two-sorted version of Parikh's theorem.

\begin{corollary}
    Let $\calT$ be a polynomially bounded theory. Then every function provably total in $\calT$ is polynomially bounded. Moreover, a function is provably total in $\calT$ if and only if it is $\Sigma_1^B$-definable in $\calT$.
\end{corollary}

\begin{definition}[Representable/Definable Relations]
    \label{def: represnetable/definable relations}
    Let $\calL \supseteq \calL_A^2$ be a two-sorted vocabulary, and let $\varphi$ be an $\calL$-formula. Then we say that $\varphi(\Vec{x}, \Vec{X})$ \emph{represents} (or \emph{defines}) a relation $R(\Vec{x}, \Vec{X})$ if
    \begin{equation}
        R(\Vec{x}, \Vec{X}) \leftrightarrow \varphi(\Vec{x}, \Vec{X})
    \end{equation}
    holds in the standard two-sorted model $\calN_2$. If $\Phi$ is a set of $\calL$-formulas, then we say that $R(\Vec{x},\Vec{X})$ is $\Phi$-\emph{representable} (or $\Phi$-\emph{definable}) if it is represented by some $\varphi \in \Phi$.
\end{definition}

\begin{definition}[$\Sigma_0^B(\Phi)$]
    Let $\Phi$ be a set of formulas in a two-sorted language $\calL \supseteq \calL_A^2$. We denote by $\Sigma_0^B(\Phi)$ the smallest class of formulas containing $\Phi \cup \Sigma_0^B(\calL)$ and closed under Boolean connectives and bounded number quantification. That is, $\Sigma_0^B(\Phi)$ is closed under $\neg$, $\land$, $\lor$, and quantifiers of the form
    \[
        \exists x \leq t \qquad\text{and} \qquad \forall x \leq t,
    \]
    where $t$ is a number term not containing $x$.
\end{definition}

\begin{definition}[Bit-definable Function]
    \label{def: Bit-definable function}
    Let $\Phi$ be a set of $\calL$-formulas, where $\calL \supseteq \calL_A^2$. We say that a string function symbol $F(\Vec{x},\Vec{Y})$ not in $\calL$ is \emph{$\Phi$-bit-definable from $\calL$} if there is a formula $\varphi(i,\Vec{x},\Vec{Y}) \in \Phi$ and an $\calL_A^2$-number term $t(\Vec{x},|\Vec{Y}|)$ such that the bit graph of $F$ satisfies
    \begin{equation}
        \label{equa: bit-defining axiom}
        F(\Vec{x},\Vec{Y})(i) \leftrightarrow \bigl(i<t(\Vec{x},|\Vec{Y}|)\land \varphi(i,\Vec{x},\Vec{Y})\bigr).
    \end{equation}
    We call the formula on the right-hand side of \eqref{equa: bit-defining axiom} a \emph{bit-defining axiom}, or a \emph{bit definition}, for $F$.
\end{definition}

One crucial point is that adding a $\Phi$-bit-definable string function $F$ together with its bit-defining axiom need not yield a conservative extension of a theory $\calT$. The bit definition describes the individual bits of the output, but the theory must also be able to prove the existence of a string collecting all these bits.

However, if $\calT$ proves the $\Phi\COMP$ axiom scheme, then comprehension can be used to construct the output string. Consequently, a polynomially bounded string function whose bit graph is $\Phi$-representable is $\Sigma_0^B(\Phi)$-definable in $\calT$. More generally, we have the following lemma.

\begin{lemma}[Extension by Bit Definition]
    \label{lemma: extension by bit definition}
    Let $\calT$ be a theory over $\calL$ that contains $\Vzero$, and let $\Phi$ be a set of $\calL$-formulas such that $\Phi \supseteq \Sigma_0^B(\calL)$. Suppose that $\calT$ proves the $\Phi\COMP$ axiom scheme. Then:
    \begin{enumerate}
        \item every polynomially bounded number function whose graph is $\Phi$-representable is $\Sigma_0^B(\Phi)$-definable in $\calT$; and
        \item every polynomially bounded string function whose bit graph is $\Phi$-representable is $\Sigma_0^B(\Phi)$-definable in $\calT$.
    \end{enumerate}
\end{lemma}

In particular, taking $\calT=\Vzero$ and $\Phi=\Sigma_0^B(\calL_A^2)$, the lemma has the following consequence (thanks to  $\Sigma_0^B(\Phi)=\Sigma_0^B(\calL_A^2)$).

\begin{corollary}
    Every function in $\FACZ$ is $\Sigma_0^B$-definable in $\Vzero$.
\end{corollary}

We now specialise the preceding notion of bit-definability to formulas in a given language. If a p-bounded string function $F$ is bit-definable by a $\Sigma_0^B(\calL)$ formula, we say that $F$ is \emph{$\Sigma_0^B$-definable from $\calL$}. Similarly, a p-bounded number function $f$ is \emph{$\Sigma_0^B$-definable from $\calL$} if its graph is represented by a $\Sigma_0^B(\calL)$ formula.

\begin{corollary}
    Let $\calT$ be a theory over $\calL$ that contains $\Vzero$, and suppose that $\calT$ proves the $\Sigma_0^B(\calL)\COMP$ axiom scheme. Then a function which is $\Sigma_0^B$-definable from $\calL$ is $\Sigma_0^B(\calL)$-definable in $\calT$.
\end{corollary}

\begin{corollary}
    Suppose that the theory $\calT$ proves $\Sigma_0^B(\calL)\COMP$ where $\calL$ is the vocabulary of $\calT$. Then the theory resulting from $\calT$ by adding the $\Sigma_0^B(\calL)$-defining axioms or the $\Sigma_0^B(\calL)$-bit-defining axioms for a collection of number functions and string functions is a conservative extension of $\calT$.
\end{corollary}

The following theorem shows that extending a theory with the $\Sigma_0^B(\calL)$-defining axioms or the $\Sigma_0^B(\calL)$-bit-defining axioms is not only a conservative extension over $\Vzero$, it also proves the $\Sigma_0^B(\calL)\COMP$ and $\Sigma_0^B(\calL)$ induction axioms schemes,
\begin{lemma}[$\Sigma_0^B$-Transformation]
    Let $\calT$ be a polynomially bounded theory which extends $\Vzero$, and assume that the vocabulary $\calL$ of $\calT$ has the same predicate symbols as $\calL_A^2$. Suppose that for every number function $f$ in $\calL$, $\calT$ proves a $\Sigma_0^B(\calL_A^2)$ defining axiom for $f$, and for every string function $F$ in $\calL$, $\calT$ proves a $\Sigma_0^B(\calL_A^2)$ bit-defining axiom for $F$. Then for every $i \geq 0$ and every $\Sigma_i^B(\calL)$ formula $\varphi^+$ there is a $\Sigma_i^B(\calL_A^2)$ formula $\varphi$ such that
    \[
        \calT \vdash \varphi^+ \leftrightarrow \varphi.
    \]
\end{lemma}

\begin{definition}[$\Delta_1^1(\calL)$ and $\Delta_1^B(\calL)$ Definable Relation]
    Let $\calT$ be a theory over the vocabulary $\calL$ and $R$ a relation symbol not in $\calL$. We say that $R$ is $\Delta_1^1(\calL)$-definable (or simply $\Delta_1^1$-definable) in $\calT$ if there are $\Sigma_1^1(\calL)$ formulas $\varphi(\Vec{x}, \Vec{X})$ and $\psi(\Vec{x}, \Vec{X})$ such that
    \[
        R(\Vec{x}, \Vec{X}) \leftrightarrow \varphi(\Vec{x}, \Vec{X})
    \]
    holds and
    \[
        \calT \vdash \varphi(\Vec{x}, \Vec{X}) \leftrightarrow \neg \psi(\Vec{x}, \Vec{X}).
    \]
    We say that $R$ is $\Delta_1^B(\calL)$-definable (or simply $\Delta_1^B$-definable) in $\calT$ if the formulas $\varphi$ and $\psi$ above are $\Sigma_1^B$ formulas.
\end{definition}

So for now, what we can conclude is that the number functions in $FAC^0$ are $\Sigma_0^B$-definable in $\Vzero$, and string functions in $FAC^0$ are not only $\Sigma_0^B$-definable, but also $\Sigma_0^B$-bit-definable in $\Vzero$. As a consequence of the witnessing theorem for $\Vzero$ (see Theorem~V.5.1 in~\cite{CN10}), one has the following definability theorem.

\begin{theorem}[$\Sigma_1^1$-Definability Theorem for $\Vzero$]
    A function is in $\FACZ$ iff it is $\Sigma_1^1$-definable in $\Vzero$ iff it is $\Sigma_1^B$-definable in $\Vzero$ iff it is $\Sigma_0^B$-definable in $\Vzero$.
\end{theorem}

\begin{corollary}
    A relation is in $\ACZ$ iff it is $\Delta_1^1$ definable in $\Vzero$ iff it is $\Delta_1^B$ definable in $\Vzero$.
\end{corollary}

Intuitively, the universal theory $\overline{V}^0$ is a conservative extension over the language $\calL_{FAC}$ which contains a function symbol for every $\FACZ$ function. Also, all those function symbols can be applied in the comprehension axioms scheme, induction scheme, etc.

\begin{lemma}
    $\overline{V}^0$ proves the $\Sigma_0^B(\calL_{FAC^0})\COMP$, $\Sigma_0^B(\calL_{FAC^0})\IND$, and $\Sigma_0^B(\calL_{FAC^0})\MIN$ axiom schemes.
\end{lemma}

\subsubsection{The Theory \texorpdfstring{$\VNCTwo$}{VNC2}}

A layered monotone Boolean circuit $C$ with $d+1$ layers is encoded as follows:
\begin{itemize}[leftmargin=*]
    \item $I$ (for ``inputs'') is a string variable, where $|I| \leq n$, that defines the (Boolean) input gates to the circuit.
    \item $G$ (for ``gates'') is a string variable such that $G(x,y)$ holds for $x \in [d]$ if the $y$th gate in layer $x$ is $\land$, and is $\lor$ otherwise.
    \item The wires of $C$ are encoded by a three-dimensional array, namely a string variable $E$ (for ``edges'') such that $E(z, x, y)$ holds if and only if the output of gate $x$ on layer $z$ is connected to the input of gate $y$ on layer $z+1$.
\end{itemize}

To evaluate the gates of $C$ on input $I$, we compute the values of the gates in each layer, starting from the input layer, in $d+1$ stages, using the values of the previous layer. The formula $\delta_{LMCV}(n,d,E,G,I,Y)$ below formalises this evaluation procedure (where LMCV stands for ``layered monotone circuit value''). The two-dimensional array $Y$ stores the computation results, namely the evaluation string: for $0 \leq z \leq d$, row $Y^{[z]}$ contains the gates on layer $z$ that output $1$.
\begin{align*}
    \delta_{LMCV}(n, d, E, G, I, Y) \equiv \forall x < n \forall z < d ( (Y(0, x) \leftrightarrow I(x)) \\ \land (Y(z+1, x) \leftrightarrow 
    ( (G(z+1, x) \land \forall u < n , E(z, u, x) \to Y(z,u)) \\ \lor (\neg G(z+1, x) \land \exists u <n, E(z, u, x) \land Y(z,u)) ))  )
\end{align*}

The following formula states that the circuit with underlying graph $(n, d, E)$ has fan-in at most two:
\begin{align*}
    & \mathit{Fanin2}(n, d, E) \equiv \forall z < d \forall x < n \exists u_1 <n \exists u_2 < n \forall v < n \\
    &\quad (E(z, v, x) \to (v = u_1 \lor v = u_2)).
\end{align*}
\begin{definition}[$\VNCTwo$]
    The theory $\VNCTwo$ has vocabulary $\calL_A^2$ and is axiomatized by the axioms of $\Vzero$ and the axiom:
    \begin{align*}
        \mathit{Fanin2}(n ,|n|^2 , E) \to \exists Y \leq \langle |n|^2 +1, n \rangle \delta_{LMCV}(n , |n|^2, E, G, I, Y).
    \end{align*}
\end{definition}

\begin{theorem}
    \begin{enumerate}
    \item A function is $\Sigma_1^B$-definable in $\VNCTwo$ if and only if it is in $\FNCTwo$.
    \item A relation is $\Delta_1^1$-definable in $\VNCTwo$ if and only if it is $\Delta_1^B$-definable in $\VNCTwo$ if and only if it is in $\NCTwo$.
\end{enumerate}
\end{theorem}

\begin{definition}[The universal extension $\overline{\VNCTwo}$]
    Let $\calL_{\FNCTwo}$ denote the universal language associated with the complexity class $\NCTwo$. It extends $\calL_A^2$ by string-function symbols representing all string functions in $\FNCTwo$; moreover, every number function in $\FNCTwo$ is represented by a term of $\calL_{\FNCTwo}$.

    Applying the general construction of~\cite[Chapter~IX]{CN10} to the theory $\VNCTwo$, we obtain a universally axiomatized theory $\overline{\VNCTwo}$ over the language $\calL_{\FNCTwo}$. The theory $\overline{\VNCTwo}$ contains defining axioms for the new function symbols given by the Cook--Nguyen construction and is a conservative extension of $\VNCTwo$.
\end{definition}

\begin{theorem}
    The theory $\overline{\VNCTwo}$ has the following properties.
    \begin{enumerate}[leftmargin=*]
        \item $\overline{\VNCTwo}$ is conservative over $\VNCTwo$. In particular, every sentence in the original language $\calL_A^2$ that is provable in $\overline{\VNCTwo}$ is already provable in $\VNCTwo$.
        \item The string functions represented by function symbols of $\calL_{\FNCTwo}$ are precisely the string functions in $\FNCTwo$, and the number functions represented by terms of $\calL_{\FNCTwo}$ are precisely the number functions in $\FNCTwo$.
        \item $\overline{\VNCTwo}$ proves the
        \[
            \Sigma_0^B(\calL_{\FNCTwo})\COMP, \qquad \Sigma_0^B(\calL_{\FNCTwo})\IND, \qquad \Sigma_0^B(\calL_{\FNCTwo})\MIN    
        \]
        axiom schemes.
    \end{enumerate}
\end{theorem}

\begin{theorem}
    \begin{itemize}[leftmargin=*]
        \item A string function is in $\FNCTwo$ iff it is represented by a string function symbol in $\calL_{\FNCTwo}$.
        \item A relation is in $\NCTwo$ iff it is represented by an open formula of $\calL_{\FNCTwo}$ iff it is represented by a $\Sigma_0^B(\calL_{\FNCTwo})$ formula.
    \end{itemize}
\end{theorem}

\begin{theorem}
    \begin{itemize}[leftmargin=*]
        \item A function is in $\FNCTwo$ iff it is $\Sigma_1^B(\calL_{\FNCTwo})$-definable in $\overline{\VNCTwo}$ iff it is $\Sigma_1^B(\calL_A^2)$-definable in $\overline{\VNCTwo}$.
        \item A relation is in $\NCTwo$ iff it is $\Delta_1^B(\calL_{\FNCTwo})$-definable in $\overline{\VNCTwo}$ iff it is $\Delta_1^B(\calL_A^2)$-definable in $\overline{\VNCTwo}$.
    \end{itemize}
\end{theorem}

\subsubsection{Some Basic Formalisation in \texorpdfstring{$\Vzero$}{V0}}\label{subsection: basic formalisation in V0}
In this section, we fix some basic notation for primitive objects coded in $\Vzero$. Most formalisations here are routine and can be found in~\cite{CN10} (cf.~\cite{MT14}). Recall that when speaking about \emph{numbers} in the theory, we mean natural numbers (that is, first-sort objects), and when discussing \emph{strings} we refer to the second-sort objects.

\paragraph{Natural number sequences of constant length.} For two numbers $x,y$ let $\langle x, y \rangle \coloneq (x+y)(x+y+1)+2y$ be the \emph{pairing function}. We also $\Sigma_0^B$-define inductively $\langle v_1, \ldots, v_k \rangle \coloneqq \langle \langle v_1, \ldots, v_{k-1}\rangle , v_k \rangle$, for every $k>2$. Then $\Vzero$ proves the injectivity of the pairing function and enables us to handle such pairs in a standard way.

\paragraph{Sequences of numbers of variable length.} In our two-sorted setting, we represent sequences of natural numbers whose \emph{length is non-constant} using string variables. Using the pairing function we can encode sequences as sets of numbers: a sequence is encoded as a string $Z$ such that the $x$th number in the sequence is $y$ if the number $\langle x, y \rangle $ is in $Z$. Formally, we have the following $\Sigma_0^B$-defining formula for the number function $\seq(x, Z)$ returning the $x$th element in the sequence $Z$:
\begin{equation}\label{equa: seq}
    \begin{aligned}
        y = \seq(x, Z) \leftrightarrow&\ (y < |Z| \land Z(\langle x , y\rangle) \land \forall z< y \neg Z(\langle x, z \rangle)   ) \\
        &\ \lor (\forall z < |Z| \neg Z(\langle x, z \rangle) \land y = |Z| ).
    \end{aligned}
\end{equation}
Formula \eqref{equa: seq} states that the $x$th element in the sequence coded by $Z$ is $y$ iff $\langle x , y \rangle$ is in $Z$ and no other number smaller than $y$ also ``occupies the $x$th position in the sequence'', and that if no number occupies position $x$ then the function returns the length of the string variable $Z$.

\paragraph{Array of strings.} We wish to encode a \emph{sequence} of binary strings as an array, that is, as a two-dimensional array. We use the function $\row(x, Z)$ to denote the $x$th string in $Z$ as follows (we follow the treatment in~\cite[Definition~V.4.26, p.~114]{CN10}). The string function $\row(x, Z)$, abbreviated $Z^{[x]}$, is $\Sigma_0^B$-definable in $\Vzero$ (and can be added to the theory $\Vzero$ to get a conservative extension and used in the induction axiom) using the following bit-definition:
\[
    \row(x, Z)(i) \leftrightarrow (i < |Z| \land Z(\langle x,i\rangle)).
\]

The theory $\Vzero$ is a conservative extension of the theory $\IDeltazero$~\cite[Thm, V.1.9]{CN10} (the single-sorted theory of arithmetic axiomatized by induction for $\Delta_0$ formulas over a base theory such as the axioms $\mathbf{B1}, \ldots, \mathbf{B8}$). In particular, recall that $\Vzero$ proves the $\Sigma_0^B\IND$ axioms, which include the $I\Delta_0\IND$ axioms. We can therefore use in $\Vzero$ some well-known definitions developed originally in $I\Delta_0$, such as the following. 

\paragraph{Graph of $2^n$.} There is a $\Delta_0^B$-definition of the graph of $2^n$ such that $\Vzero$ proves that it is a partial function that satisfies the defining recurrence $2^0 =1$, $2^{n+1} = 2 \cdot 2^n$. Also, $\Vzero$ proves that if $y =2^x$, then for any $n < x$, there exists $m$ such that $m = 2^n$. Moreover, if $y=2^x$ and $z = 2^h$, then $ y \cdot z = 2^{x+h}$ (see, e.g.,~\cite[III.3, p.~61]{CN10}). Thus, there is also a well-behaved $\Delta_0^B$-definition of the function $\bit(x, i) = \lfloor x 2^{-i} \rfloor \mod 2$, and $|x| = \min\{n : x < 2^n\}$.

\subsection{Some basic formalisation of algebra in \texorpdfstring{$\VNCTwo$}{VNC2}}

In this section we fix some basic notation for algebraic objects coded in $\VNCTwo$. Throughout the remainder of the paper, we use a first-sort variable $N$ as the main size parameter; in our applications, $N$ denotes the number of coordinates in the coefficient vector under consideration. We only consider objects whose encodings have size polynomially bounded in $N$: all first-sort parameters, such as indices, degrees, matrix dimensions, and circuit sizes, are bounded by a fixed polynomial in $N$, and all second-sort objects are represented by strings of length bounded by a fixed polynomial in $N$. Large integers and rational numbers are represented in binary, so polynomial boundedness refers to the length of their binary encodings rather than to their numerical values. Under these bounds, the instances of binary integer and rational arithmetic used in the proof, including addition, multiplication, iterated addition and multiplication, division, factorials, and binomial coefficients, are computable by functions in $\FNCTwo$ and hence are available in the universal conservative extension $\overline{\VNCTwo}$. We use $|\cdot|_1$ for the ordinary absolute value and for its formal counterpart on encoded rational numbers. Rational numbers are encoded as tuples of binary integers. See~\cite{Jev22,Jev23} for the relevant bounded-arithmetic formalisation of binary arithmetic.

Throughout each formalised algebraic statement, the underlying coefficient field $\F$ is fixed. We use either $\Qbb$ or a fixed finite field $\F_q$. Elements of $\Qbb$ are encoded as above, while elements of a fixed finite field $\F_q$ have constant-length encodings. In either case, field addition, multiplication, and equality are computable in $\FNCTwo$. Under the polynomial bounds imposed above, all algebraic objects appearing below have encodings of length polynomially bounded in $N$.

\paragraph{Matrices.}
We encode an $n\times m$ matrix as an array of $n$ strings, each of which encodes one row of the matrix, and each row is itself an array of $m$ field elements. We only consider matrices whose dimensions and whose total encoding lengths are polynomially bounded in $N$. The basic matrix operations that we use, including matrix addition, matrix multiplication, and the Kronecker product of a constant number of matrices, are computable by functions in $\FNCTwo$. Hence, by the definability theorem for $\VNCTwo$, these operations are $\Sigma_1^B$-definable in $\VNCTwo$ and may be represented by function symbols in the universal conservative extension $\overline{\VNCTwo}$. Consequently, when working in $\overline{\VNCTwo}$, we may use these matrix operations inside $\Sigma_0^B(\calL_{\FNCTwo})$ formulas and apply the corresponding comprehension and induction schemes.

\paragraph{Polynomial.}
Fix a canonical ordering of monomials. For a polynomial $f$, we write $\coeff_N(f)$ for the vector consisting of the coefficients of the first $N$ monomials in this ordering. Whenever we use this notation, $N$ is chosen large enough that every monomial which may occur in $f$ is among these $N$ monomials, so that $\coeff_N(f)$ determines $f$. We encode the coefficient vector as a string; throughout the proof, the bit-length of every coefficient and hence the total length of this encoding are polynomially bounded in $N$.

We say that a polynomial family $f$ is $\Phi$-definable if its $N$-coordinate coefficient-vector function is $\Phi$-bit-definable with output length polynomially bounded in $N$. In particular, the Permanent and Iterated Matrix Multiplication polynomials have $\Sigma_0^B$-definable coefficient-vector functions. For the Permanent polynomial
\[
    \perm_n(X)=\sum_{\sigma\in S_n}\prod_{i=1}^n x_{i,\sigma(i)},
\]
the coefficient of a monomial is non-zero if and only if the monomial contains exactly one variable from every row and every column. For
\[
    \IMM_{n,d}(X) = \bigl(X^{(1)}X^{(2)}\cdots X^{(d)}\bigr)_{1,1},
\]
the coefficient of a monomial is non-zero if and only if the monomial contains exactly one variable from each matrix $X^{(t)}$, the selected variable in $X^{(1)}$ has row index $1$, the selected variable in $X^{(d)}$ has column index $1$, and, for every $t<d$, the column index selected in $X^{(t)}$ equals the row index selected in $X^{(t+1)}$. These conditions are decidable in $\ACZ$ and therefore in $\NCTwo$. We use
\[
    \PermT(n,N)\coloneqq \coeff_N(\perm_n)  \qquad\text{and}\qquad \IMMT(n,d,N)\coloneqq \coeff_N(\IMM_{n,d})
\]
for the corresponding string functions. Whenever these functions are used, the parameters $n,d$ and the output encoding lengths are polynomially bounded in $N$.

\paragraph{Algebraic circuit.}
    Algebraic circuits are encoded using strings in the theory as follows~\cite{TC21}:
    \begin{enumerate}[leftmargin=*]
    \item a sequence $V$ of node labels $v_1,v_2,\ldots,v_k$, where each label is a natural number;
    \item a string of gates $G$, where each gate is encoded by a pair $(v,t)$ of natural numbers: $v$ identifies a node in $V$, while $t$ specifies whether the node is a $+$-gate, a $\times$-gate, or an input gate;
    \item a two-dimensional string $I$ encoding the input gates, where an input gate encodes either a variable $x_j$ or an element of the fixed coefficient field $\F$; and
    \item a string $E$ encoding directed edges between nodes, where $(u,v)\in E$ means that there is an edge from $u$ to $v$.
    \end{enumerate}
    Throughout the formalisation, the lengths of $V,G,I,E$, the number of circuit nodes, the indices of the variables occurring in the circuit, and the bit-lengths of all scalar constants are polynomially bounded in $N$. Thus a circuit of the type considered in the proof is represented by a tuple of strings of length polynomial in $N$.

\paragraph{Rank function.}
We use the standard $\NCTwo$ algorithms for matrix rank. For rational matrices, Ken and Kuroda~\cite{KK23}, building on Tzameret and Cook~\cite{TC21}, formalised the relevant rank algorithm in $\VNCTwo$ and proved the usual linear-algebraic characterisations of rank. We write $\rank_{\F}(M)$ for the rank over the fixed underlying field. In particular, the rank is bounded by the number of rows and by the number of columns, and it satisfies the standard rank identities and inequalities used below. More precisely, Ken and Kuroda showed that the rank $\rank(M)$ of a matrix $M$ with $\F$-coefficients coincides with~\cite[Lemmas~60, 62 and Theorem~65]{KK23}:
\begin{itemize}[leftmargin=*]
    \item the number of linearly independent row vectors;
    \item the number of linearly independent column vectors;
    \item the largest order of a non-zero minor.
\end{itemize}

\subsection{Formalisation of~\texorpdfstring{\cite{LST25}}{[LST25]} and~\texorpdfstring{\cite{For24}}{[For24]} in \texorpdfstring{$\VNCTwo$}{VNC2}}

In this section, the notations and concepts are taken from~\cite{LST25} and~\cite{For24}.

We present a proof outline here:
\begin{enumerate}[leftmargin=*]
    \item We define a rank measure, the relative rank, in the theory.
    \item By a basic rank argument, the coefficient matrix of the word polynomial $P_w$ induced by a $b$-unbiased word $w$ has high relative rank.
    \item We then show that, whenever the full coefficient vector has polynomially bounded encoding length in the parameter regime considered below, polynomial identities involving the relevant constant-depth circuits can be proved in $\VNCTwo$ by reasoning about their full coefficient vectors.
    \item With the ability to prove polynomial identities of constant-depth circuits, we show that for any constant-depth circuit $C(\overline{X})$ that computes a polynomial $f(\overline{X})$, the set-multilinear projection $\pi_{sm}(f)$ can be computed by a constant-depth set-multilinear circuit $C'$ of size $\poly(s,d^d)$ over the sets of variables $(X_1,\ldots,X_d)$. In particular, for a set-multilinear polynomial $f$, a sufficiently strong lower bound against constant-depth set-multilinear circuits implies a lower bound against general constant-depth circuits computing $f$.
    \item We prove an upper bound on the relative rank of any constant-depth set-multilinear circuit $C'$ in terms of the size of $C'$. Roughly speaking, the relative rank of small-size constant-depth set-multilinear circuits cannot be too big.
    \item Combining the above lower bound for the relative rank of $P_w$ with the upper bound for the relative rank of small-size constant-depth set-multilinear circuits, we obtain a size lower bound against constant-depth set-multilinear circuits that compute $P_w$. Moreover, this lower bound is sufficient to imply a superpolynomial-size lower bound against constant-depth circuits.
    \item Finally, we transfer the lower bound for $P_w$ to the Permanent polynomial using the explicit polynomial projection constructed later in the proof from a suitable Permanent polynomial to $P_w$.
\end{enumerate}

Throughout the formalisation, we use a first-sort variable $N$ as the main size parameter. The parameter $N$ denotes the number of coordinates in the coefficient vector under consideration. For polynomials in $n$ variables of total degree at most $d$, this means $N \geq \binom{n+d}{d}$. 

All auxiliary first-sort quantities occurring in the formalisation are bounded by a fixed polynomial in $N$, and every second-sort object is encoded by a string of length bounded by a fixed polynomial in $N$. This convention guarantees that all objects manipulated in the formal proof have polynomial-size encodings with respect to $N$. It does not change the size parameter in the algebraic circuit lower bound itself.

We retain the parameter choices from the algebraic lower-bound argument. In particular,
\[
    n(N) = \lfloor \log_2 N \rfloor, \quad d(N) = \lfloor c\log_2 n(N) \rfloor, \quad k(N)=10d(N).
\]
where $c$ is a fixed rational constant. These relations can be represented by relations in $\NCTwo$ and all the resulting first-sort parameters are polynomially bounded in $N$.

We will study word polynomials indexed by words over a finite alphabet depending on the parameter $k$. Let $ a(k)=\left\lfloor\frac{k}{\sqrt{2}}\right\rfloor$ and define the two-element alphabet $A(k)\coloneqq\{a(k),-k\}$. We then consider words $w=(w_1,\ldots,w_d)\in A(k)^d$ and the associated word polynomials $P_w$, which will be defined below. The integer $a(k)$ can be defined without introducing the irrational number $1/\sqrt{2}$ into the language: it is the unique nonnegative integer $a$ satisfying $2a^2\leq k^2<2(a+1)^2$. Thus every entry of $w$ has absolute value at most $k$. The word polynomial $P_w$ is defined for every $w\in A(k)^d$; when proving the lower bound, we will additionally require $w$ to be $k$-unbiased.

We will use $\Delta$ to denote the product-depth of algebraic circuits, which is constant. We will consider algebraic circuits of size polynomially bounded in the number variables $n$. Namely, if $s$ is a fixed standard polynomial, we consider circuits of size at most $s(n)$. The strings $V,G,I,E$ encoding such a circuit therefore have length at most $\poly(s(n))$.

\subsubsection{Exponentiation in \texorpdfstring{$\VTCzero$}{VTC0}}

The proof of~\cite{LST25} applies one elementary analytic function, the exponential function, to rational inputs, producing values that need not be rational. \Jerabek{}~\cite{Jev23} formalises the required elementary analytic functions in $\VTCzero$ and gives uniform $\TCzero$ rational approximations of these functions. We give a short exposition of the results necessary for our work. We restrict attention to real inputs, although the results in~\cite{Jev23} are proved more generally for complex inputs.

Let $D\subseteq\Rbb$ and let $f:D\to\Rbb$. An \emph{additive $\TCzero$ approximation} of $f$ is a $\TCzero$ function
\[
    f_+:\Qbb\times\mathbb{N}\to\Qbb
\]
such that
\[
    |f_+(z,t)-f(z)|_1\leq 2^{-t}
\]
for every $z\in D\cap\Qbb$ and every first-sort precision parameter $t$. A \emph{multiplicative $\TCzero$ approximation} of $f$ is a $\TCzero$ function
\[
    f_\times:\Qbb\times\mathbb{N}\to\Qbb
\]
such that
\[
    |f_\times(z,t)-f(z)|_1 \leq 2^{-t}|f(z)|_1
\]
for every $z\in D\cap\Qbb$ and every first-sort precision parameter $t$. Equivalently, whenever $f(z)\neq0$,
\[
    \left| \frac{f_\times(z,t)}{f(z)}-1 \right|_1 \leq 2^{-t}.
\]
The approximation algorithms for the exponential function require an additional magnitude parameter, since the binary length of the output may grow with the magnitude of the input. Let $P(z,r)$ be a property of a rational input $z$ and a first-sort parameter $r$. We say that a $\TCzero$ function $f_+(z,r,t)$ is an \emph{additive approximation of $f(z)$ parametrised by $r$ subject to $P(z,r)$} if
\[
    P(z,r) \implies |f_+(z,r,t)-f(z)|_1\leq2^{-t}
\]
for every $z\in D\cap\Qbb$ and all first-sort parameters $r,t$. We additionally require that for every $z \in D \cap \Qbb$, there exists a first-sort parameter $r$ such that $P(z,r)$. Multiplicative approximation parametrised by $r$ is defined analogously.

In the applications below, the binary encoding length of $z$ and the numerical values of the first-sort parameters $r,t$ are all bounded by fixed polynomials in the global size parameter $N$. Therefore every $\TCzero$ approximation function used below has output encoding length polynomially bounded in $N$.

\begin{definition}[The exponential function]
    \label{def: exp function}
    For $z\in\Qbb$ and a first-sort number $m$, define the truncated exponential series
    \[
        e(z,m) \coloneqq \sum_{j<m}\frac{z^j}{j!}.
    \]      
    For semantic purposes, we use $\exp:\Rbb\to\Rbb_{>0}$ for the ordinary real exponential function. For every rational $z$,
    \[
        \exp z = \lim_{m\to\infty}e(z,m).
    \]       
    The exact real number $\exp z$ is used only in the semantic correctness statements for the rational approximation functions below; the bounded-arithmetic theory itself manipulates their rational outputs.
    We use the standard identity
    \[
         \exp(z+w)=\exp z\,\exp w
    \]
    for all $z,w\in\Rbb$.
\end{definition}

\begin{lemma}[Existence of $\TCzero$ approximations of $\exp$~\cite{Jev23}]
    \label{lemma: existence of TCzero approximation of exp}
    There are $\TCzero$ functions $E_\times(z,r,t)$ and $E_+(z,r,t)$ such that, for all $z\in\Qbb$ and all first-sort parameters $r,t$:
    \begin{enumerate}[leftmargin=*]
        \item if $|z|_1\leq r$, then
        \[          
            |E_\times(z,r,t)-\exp z|_1 \leq 2^{-t}\exp z;
        \]
        hence $E_\times(z,r,t)$ is a multiplicative approximation of $\exp z$ parametrised by $r$ subject to $|z|_1\leq r$;
        \item if $z\leq r$, then
        \[
            |E_+(z,r,t)-\exp z|_1 \leq 2^{-t};
        \]         
        hence $E_+(z,r,t)$ is an additive approximation of $\exp z$ parametrised by $r$ subject to $z\leq r$.
    \end{enumerate}
\end{lemma}

In the proof, we use a rational approximation of $2^z$ rather than directly approximating $e^z$. Since $2^z=\exp(z\ln2)$, we also require a rational approximation of $\ln2$. The construction of the logarithm function in~\cite{Jev23} is more involved, so we state only the facts that we need.

\begin{lemma}[\cite{Jev23}]
    The real logarithm function $ \ln:\Rbb_{>0}\to\Rbb $ is the inverse of the real exponential function. Thus $\exp(\ln z)=z$ for every $z>0$, and $\ln(\exp z)=z$ for every $z\in\Rbb$.
    
    Moreover, $\ln z$ admits uniform $\TCzero$ additive and multiplicative rational approximations on positive rational inputs. In particular, there are $\TCzero$ functions $L_+(z,t)$ and $L_\times(z,t)$ such that for every $z\in\Qbb_{>0}$ and every first-sort precision parameter $t$,
    \[
        |L_+(z,t)-\ln z|_1\leq2^{-t}
    \]
    and
    \[
        |L_\times(z,t)-\ln z|_1 \leq 2^{-t}|\ln z|_1.
    \]
\end{lemma}

\begin{remark}
    In~\cite{Jev23}, the functions $\exp$ and $\log$ are developed more generally over complex inputs. We need only their restriction to real inputs and, for the logarithm, only the value $\ln2$. The exact real-valued functions are used only to state the semantic guarantees of the approximation procedures, while the theory itself works with their rational-valued $\TCzero$ approximations.
\end{remark}

We now construct the approximation of the powering function
\[
    \pow(z)\coloneqq 2^z=\exp(z\ln2).
\]

\begin{lemma}
    \label{lemma: existence of powering function}
    Define
    \[
         P_\times(z,r,t) \coloneqq  E_\times\!\left(zL_\times(2,t+r+4), 2r, t+2 \right).
    \]  
    Then $P_\times$ is a $\TCzero$ multiplicative approximation of $\pow(z)=2^z$, parametrised by $r$ subject to $|z|_1\leq r$. More precisely,
    \[
        |P_\times(z,r,t)-2^z|_1 \leq 2^{-t}2^z
    \]        
    whenever $z\in\Qbb$ and $|z|_1\leq r$.
\end{lemma}

\begin{proof}
    Let $s=t+r+4.$ By the multiplicative approximation property of $L_\times$, there exists a real number $\delta_1$ such that $L_\times(2,s) = (\ln2)(1+\delta_1)$, with $|\delta_1|_1\leq2^{-s}$. Since $s\geq4$ and $\ln2<1$, we have $|L_\times(2,s)|_1 \leq (\ln2)(1+2^{-4}) <2$. Consequently, if $|z|_1\leq r$, then $|zL_\times(2,s)|_1\leq 2r$, so the magnitude parameter $2r$ is sufficient for the call to $E_\times$.
    
    By the multiplicative approximation property of $E_\times$, there exists a real number $\delta_2$ such that
    \[
        P_\times(z,r,t) = \exp\!\bigl(zL_\times(2,s)\bigr)(1+\delta_2), \qquad |\delta_2|_1 \leq 2^{-(t+2)}.
    \]
    Substituting the expression for $L_\times(2,s)$ gives
    \[
        zL_\times(2,s) = z\ln2+z\ln2\,\delta_1.
    \]
    Hence, using $\exp(a+b)=\exp a\,\exp b$,
    \[
        P_\times(z,r,t) = 2^z \exp(z\ln2\,\delta_1) (1+\delta_2).
    \]
    Therefore
    \[
        \frac{P_\times(z,r,t)}{2^z} = (1+\delta_2)e^\tau,  
    \]
    where $\tau\coloneqq z\ln2\,\delta_1.$ Since $|z|_1\leq r$, $\ln2<1$, and $|\delta_1|_1\leq2^{-(t+r+4)}$, we obtain
    \[
        |\tau|_1 \leq r\,2^{-(t+r+4)}.
    \]
    Using $r\leq2^r$ for every natural number $r$, we get
    \[
        |\tau|_1 \leq 2^{-(t+4)}.
    \]
    In particular $|\tau|_1\leq1$. For every real $x$ with $|x|_1\leq1$,
    \[
        |e^x-1|_1\leq2|x|_1,
    \]
    and therefore
    \[
        |e^\tau-1|_1 \leq 2^{-(t+3)}.
    \]
    We now estimate the relative error:
    \begin{align*}
        \frac{P_\times(z,r,t)}{2^z}-1 &= (1+\delta_2)e^\tau-1\\
        &= \delta_2+(e^\tau-1)+\delta_2(e^\tau-1).
    \end{align*}
    Hence,
    \begin{align*}
        \left| \frac{P_\times(z,r,t)}{2^z}-1 \right|_1 &\leq |\delta_2|_1 + |e^\tau-1|_1 + |\delta_2|_1|e^\tau-1|_1\\
            &\leq 2^{-(t+2)} + 2^{-(t+3)} + 2^{-(t+2)}2^{-(t+3)}\\
            &\leq 2^{-t}.
    \end{align*}
    Multiplying by $2^z>0$ yields
    \[
        |P_\times(z,r,t)-2^z|_1 \leq 2^{-t}2^z.
    \]
    Thus $P_\times$ is a multiplicative $\TCzero$ approximation of $\pow$.
\end{proof}

Since $P_\times$ is obtained by composition of $\TCzero$ functions and rational arithmetic, it is itself a $\TCzero$ function. Hence $P_\times\in \FTCzero \subseteq\FNCTwo$ and is available in the universal conservative extension $\overline{\VNCTwo}$. In all applications below, the encoding length of $z$ and the first-sort parameters $r,t$ are polynomially bounded in the global size parameter $N$. Therefore the binary encoding of $P_\times(z,r,t)$ is also polynomially bounded in $N$.

\subsubsection{Word polynomials}

Consider the set of words on an alphabet $A\subseteq\Z\setminus{0}$. Let $w=(w_1,\ldots,w_d)\in A^d$. For a subset $S\subseteq[d]$, let $ w_S\coloneqq\sum_{i\in S}w_i.$ We define $\calP_w=\{i\in[d]:w_i\geq0\}$ and $\calN_w=\{i\in[d]:w_i<0\}$, the positive and negative index sets of $w$, respectively. We say that $w\in A^d$ is \emph{$b$-unbiased} if
$|w_{[t]}|_1\leq b$ for every $t\leq d$. Note that $|\cdot|_1$ denotes the ordinary absolute value.

Given $w$, we denote by $\overline{X}(w) = (X(w_1),\ldots,X(w_d))$ a tuple of $d$ sets of disjoint variables, where $\#X(w_i)=2^{|w_i|_1}.$ We denote by $\F_{sm}[T]$ the space of set-multilinear polynomials over a tuple $T$ of sets of variables. Hence $\F_{sm}[\overline{X}(w)]$ consists of polynomials in which every monomial contains exactly one variable from each set $X(w_i)$.

Since $X(w_i)$ contains $2^{|w_i|_1}$ variables, we label its variables by strings in $\{0,1\}^{|w_i|_1}.$ Let $M_w^\calP$ and $M_w^\calN$ denote the sets of set-multilinear monomials over the positive and negative variable sets, respectively. Given a monomial $m\in\F_{sm}[\overline{X}(w)]$, let $m_+$ denote its positive part in $M_w^\calP$ and let $m_-$ denote its negative part in $M_w^\calN$. We use $\sigma(m_+)$ to denote the Boolean string associated with $m_+$, and $\sigma(m_{-})$ to denote the Boolean string associated with $ m_{-}$.

More precisely, if $j_1 < \cdots < j_t$ and $m^\prime = x_{\sigma_1}^{(j_1)} x_{\sigma_2}^{(j_2)} \cdots x_{\sigma_t}^{(j_t)}$ with $x_{\sigma_i}^{(j_i)} \in X_{j_i}$ and $\sigma_i \in {\{0,1\}}^{|w_{j_i}|_1}$ for each $i \in [t]$, then $\sigma(m^\prime)$ is defined to be $\sigma_1 \cdots \sigma_t$. 

The length of $\sigma(m_+)$ is $\sum_{i\in\calP_w}w_i$, whereas the length of $\sigma(m_-)$ is $-\sum_{i\in\calN_w}w_i$.

Consequently, the difference between these two lengths is $\left| \sum_{i\in[d]}w_i \right|_1 = |w_{[d]}|_1$. In particular, if $w$ is $b$-unbiased, this difference is at most $b$.

For binary strings $x$ and $y$, write $x\sim y$ if the shorter of $x,y$ is a prefix of the longer one. The word polynomial $P_w$ is
\[
    P_w = \sum_{\substack{m\in\F_{sm}[\overline{X}(w)]\\ \sigma(m_+)\sim\sigma(m_-)}} m.
\]
where $m$ is a monomial.

\subsubsection{Word polynomials are \texorpdfstring{$\Sigma_1^B$}{Sigma1B}-definable in \texorpdfstring{$\VNCTwo$}{VNC2}}

Let $L(w)\coloneqq \sum_{j=1}^d |w_j|_1$ and $N_{sm}(w)\coloneqq 2^{L(w)}$. Thus $N_{sm}(w)$ is the number of set-multilinear monomials over $\overline{X}(w)$. We choose the parameters $n,d,k$ as functions of the global parameter $N$ so that $ L(w)\leq kd\leq \lfloor\log_2 N\rfloor$, and hence $ N_{sm}(w)=2^{L(w)}\leq N$. In particular, every number and string occurring in the definition below is polynomially bounded in $N$.

We order the set-multilinear monomials over $\overline{X}(w)$ by identifying them with the binary strings in ${0,1}^{L(w)}$, as in the definition of $P_w$ above. We define the string function $\WordP(w,d,N)$ to be the $N$-coordinate coefficient vector of $P_w$ in this ordering, padded by zeros after the first $N_{sm}(w)$ coordinates.

For $i<N_{sm}(w)$, let $u=\binaryT(i,L(w))$ be the binary representation of $i$ of length exactly $L(w)$, padded with leading zeros when necessary. Thus
\[
    \binaryT(K,L)(j) = \left\lfloor \frac{K}{2^{L-1-j}} \right\rfloor \bmod 2 \qquad \text{for }j=0,\ldots,L-1,
\]
whenever $K<2^L$.

We partition $u$ into $d$ consecutive blocks such that $u=u_1u_2\cdots u_d$ and $|u_j|=|w_j|_1$.Let
\[
    \calP_w={j_1<j_2<\cdots<j_p}, \qquad \calN_w={k_1<k_2<\cdots<k_q},
\]
we define
\[
    u^\calP \coloneqq u_{j_1}u_{j_2}\cdots u_{j_p}, \qquad u^\calN \coloneqq u_{k_1}u_{k_2}\cdots u_{k_q}.
\]
where the blocks are concatenated in increasing order of their indices. Thus $u^\calP$ and $u^\calN$ are precisely the strings $\sigma(m_+)$ and $\sigma(m_-)$ associated with the set-multilinear monomial indexed by $i$.

We define the prefix relation by
\[
    \Prefix(x,y) \coloneqq |x|\leq |y| \land \forall j<|x|\, \bigl(x(j)\leftrightarrow y(j)\bigr).
\]
Then the bit graph of $\WordP$ is given by
\begin{align*}
    \WordP(w,d,N)(i) \leftrightarrow{}& i<N \land i<N_{sm}(w) \\
    &\land \Bigl( \Prefix(\binaryT(i,L(w))^\calP, \binaryT(i,L(w))^\calN) \\
    &\qquad\lor \Prefix(\binaryT(i,L(w))^\calN, \binaryT(i,L(w))^\calP)\Bigr).
\end{align*}
For $N_{sm}(w)\leq i<N$, the $i$th bit of $\WordP(w,d,N)$ is therefore $0$.

The functions used above---computing $L(w)$, extracting the blocks of $\binaryT(i,L(w))$, concatenating the positive and negative blocks, and checking the prefix relation---are computable in $\FNCTwo$. Moreover, the output length of $\WordP(w,d,N)$ is at most $N$. Therefore $\WordP\in\FNCTwo.$

All the operations appearing in the above definition, including $\binaryT$, block extraction, concatenation, and the prefix test, are computable in $\FNCTwo$. Hence they are available as function symbols in the universal language $\calL_{\FNCTwo}$. Therefore, the displayed formula gives a $\Sigma_0^B(\calL_{\FNCTwo})$ bit-definition of the string function $\WordP(w,d,N)$ in $\overline{\VNCTwo}$.

Since the output length of $\WordP(w,d,N)$ is bounded by $N$, the corresponding string exists by $\Sigma_0^B(\calL_{\FNCTwo})$-comprehension in $\overline{\VNCTwo}$. By conservativity and the $\Sigma_0^B$-transformation theorem, this definition can be translated back to the language of $\VNCTwo$. In particular, $\WordP$ is $\Sigma_1^B$-definable in $\VNCTwo$.

\subsubsection{Relative rank}

For $f\in\F_{sm}[\overline{X}(w)]$, define $M_w(f)$ to be the coefficient matrix whose rows are indexed by $M_w^\calP$, whose columns are indexed by $M_w^\calN$, and whose $(m_1,m_2)$-entry is the coefficient of the monomial $m_1m_2$ in $f$.

Let $ L(w)\coloneqq\sum_{i=1}^d|w_i|_1$ and $s(w)\coloneqq w_{[d]}=\sum_{i=1}^d w_i$. Then $|M_w^\calP| = 2^{\sum_{i\in\calP_w}w_i}$, and $|M_w^\calN| = 2^{-\sum_{i\in\calN_w}w_i}$, and consequently $ |M_w^\calP|\cdot|M_w^\calN| = 2^{L(w)}$. Moreover, $\min \{|M_w^\calP|,|M_w^\calN|\} = 2^{\frac{L(w)-|s(w)|_1}{2}}$. Notice that $L(w)-|s(w)|_1$ is even, so the exponent in the last expression is an integer.

The relative rank of $f$ is defined as
\[
    \relrk_w(f) \coloneqq \frac{\Rank(M_w(f))}{\sqrt{|M_w^\calP|\cdot|M_w^\calN|}}
    = \frac{\Rank(M_w(f))}{2^{L(w)/2}} \leq 1.
\]
As the denominator need not be rational, so inside $\VNCTwo$ we instead use the rational approximation $P_\times$ defined above.

Both dimensions of $M_w(f)$ are at most $N$, and the matrix has at most $N^2$ entries. Thus its encoding has length polynomially bounded in $N$. If $Y$ is the $N$-coordinate coefficient vector of a set-multilinear polynomial, we also write $M_w(Y)$ for the matrix obtained from $Y$ by extracting the corresponding coefficients. This matrix-extraction operation is in $\FNCTwo$.

We define the rational relative-rank function used in $\VNCTwo$ by
\[
    \relrank(w,d,N,M) \coloneqq \frac{\rank(M)} {P_\times\!\left(\frac{L(w)}2,N,N\right)}.
\]
Since $L(w)\leq N$, the magnitude parameter $N$ is valid. We also use $N$ as the precision parameter. Put $\epsilon_N\coloneqq2^{-N}.$ By the multiplicative approximation property of $P_\times$,
\[
    (1-\epsilon_N)2^{L(w)/2}
    \leq
    P_\times\!\left(\frac{L(w)}2,N,N\right)
    \leq
    (1+\epsilon_N)2^{L(w)/2}.
\]
For $N\geq1$, the denominator is therefore positive.

By the rank properties established in \cite{KK23}, $\VNCTwo$ proves the following rational analogues of the standard properties of relative rank.

\begin{claim}
    Let $w\in A^d$ satisfy $N_{sm}(w)\leq N$.
    \begin{itemize}[leftmargin=*]
        \item \textbf{Imbalance.} For every $f\in\F_{sm}[\overline{X}(w)]$,
        \[
            \relrank(w,d,N,M_w(f)) \leq \frac{2^{(L(w)-|s(w)|_1)/2}}{P_\times(L(w)/2,N,N)}.
        \]
        Moreover, 
        \[
            \relrank(w,d,N,M_w(f)) \leq \frac{1+\epsilon_N}{1-\epsilon_N}\,\frac{1}{P_\times(|s(w)|_1/2,N,N)}.
        \]
        \item \textbf{Subadditivity.} For every $f,g\in\F_{sm}[\overline{X}(w)]$,
        \[
            \relrank(w,d,N,M_w(f+g)) \leq \relrank(w,d,N,M_w(f)) + \relrank(w,d,N,M_w(g)).
        \]
        \item \textbf{Multiplicativity.} Suppose $f=f_1f_2\cdots f_t$, where $(S_1,\ldots,S_t)$ is a partition of $[d]$ and $f_i\in \F_{sm}[\overline{X}(w\restriction S_i)]$ for every $i\in[t]$. Then
        \[
            \begin{aligned}
                \relrank(w,d,N,M_w(f)) &\leq \frac{(1+\epsilon_N)^t}{1-\epsilon_N} \prod_{i=1}^t \relrank\!\left(w\restriction S_i, |S_i|, N, M_{w\restriction S_i}(f_i) \right)\\
                &\leq \left(\frac{1+\epsilon_N}{1-\epsilon_N} \right)^t \prod_{i=1}^t \relrank\!\left( w\restriction S_i, |S_i|, N, M_{w\restriction S_i}(f_i) \right).
            \end{aligned}
        \]
    \end{itemize}
\end{claim}

\begin{proof}
    \textbf{Imbalance.} Since $\rank(M_w(f)) \leq \min\{|M_w^\calP|,|M_w^\calN|\}$, we have $\rank(M_w(f)) \leq 2^{(L(w)-|s(w)|_1)/2}$. Dividing by the positive rational number
    $P_\times(L(w)/2,N,N)$ gives
    \[
        \relrank(w,d,N,M_w(f)) \leq \frac{2^{(L(w)-|s(w)|_1)/2}}{P_\times(L(w)/2,N,N)}.
    \]
    From the approximation bounds, $P_\times(L(w)/2,N,N) \geq (1-\epsilon_N)2^{L(w)/2}$ and $P_\times(|s(w)|_1/2,N,N) \leq (1+\epsilon_N)2^{|s(w)|_1/2}$.
    Therefore, we have
    \[
        \begin{aligned}
            \relrank(w,d,N,M_w(f)) &\leq \frac{2^{-|s(w)|_1/2}}{1-\epsilon_N}\\
            &\leq \frac{1+\epsilon_N}{1-\epsilon_N}\frac{1}{P_\times(|s(w)|_1/2,N,N)}.
        \end{aligned}
    \]
    
    \textbf{Subadditivity.} The matrices satisfy $M_w(f+g)=M_w(f)+M_w(g)$, and matrix rank is subadditive, $\rank(M_w(f+g)) \leq \rank(M_w(f))+\rank(M_w(g)).$ Since all three relative ranks have the same positive denominator $P_\times(L(w)/2,N,N)$, the desired inequality follows.
    
    \textbf{Multiplicativity.} Up to permutations of rows and columns,
    \[
        M_w(f_1\cdots f_t) = \bigotimes_{i=1}^t M_{w\restriction S_i}(f_i).
    \]
    Hence
    \[
        \rank(M_w(f_1\cdots f_t)) = \prod_{i=1}^t \rank(M_{w\restriction S_i}(f_i)).
    \]
    Set
    \[
        L_i \coloneqq \sum_{j\in S_i}|w_j|_1.
    \]
    Since $\sum_i L_i=L(w)$, we obtain
    \[
        \begin{aligned}
            &\relrank(w,d,N,M_w(f_1\cdots f_t))\\
            &\quad= \prod_{i=1}^t \relrank\!\left(w\restriction S_i, |S_i|, N, M_{w\restriction S_i}(f_i) \right) \cdot
            \frac{\prod_{i=1}^tP_\times(L_i/2,N,N)}{P_\times(L(w)/2,N,N)}.
        \end{aligned}
    \]
    The approximation bounds imply
    \[
        \prod_{i=1}^tP_\times(L_i/2,N,N) \leq (1+\epsilon_N)^t2^{L(w)/2}
    \]
    and
    \[
        P_\times(L(w)/2,N,N) \geq (1-\epsilon_N)2^{L(w)/2}.
    \]
    Thus
    \[
        \frac{\prod_{i=1}^tP_\times(L_i/2,N,N)}{P_\times(L(w)/2,N,N)} \leq \frac{(1+\epsilon_N)^t}{1-\epsilon_N},
    \]
    which proves the first inequality. The second follows because $(1-\epsilon_N)^{t-1}\leq1$.
\end{proof}

\subsubsection{Word polynomials have high rank}

The matrix $M_w(P_w)$ has full rank. Indeed, suppose $|M_w^\calP|\geq|M_w^\calN|$. Then $\sigma(m_-)$ is no longer than $\sigma(m_+)$. For every column indexed by $m_-\in M_w^\calN$, the non-zero entries are precisely the rows whose associated strings extend $\sigma(m_-)$. The supports of two distinct columns are disjoint, so the columns are linearly independent over every field. Hence $\Rank(M_w(P_w))=|M_w^\calN|$ The case $|M_w^\calP|<|M_w^\calN|$ is symmetric, using the rows instead. Therefore
\[
    \Rank(M_w(P_w)) = \min\{|M_w^\calP|,|M_w^\calN|\} = 2^{(L(w)-|s(w)|_1)/2}.
\]
Consequently, for the ordinary relative-rank measure, $\relrk_w(P_w)=2^{-|s(w)|_1/2}.$ In particular, if $w$ is $b$-unbiased, then $|s(w)|_1\leq b$ and $\relrk_w(P_w)\geq2^{-b/2}.$

Inside $\VNCTwo$, we use the rational approximation instead. Since
\[
    \rank\!\left(M_w(\WordP(w,d,N))\right) = 2^{(L(w)-|s(w)|_1)/2},
\]
we have
\[
    \relrank\!\left(w,d,N,M_w(\WordP(w,d,N))\right) = \frac{2^{(L(w)-|s(w)|_1)/2}}{P_\times(L(w)/2,N,N)}.
\]
If $w$ is $b$-unbiased, then
\[
    \relrank\!\left(w,d,N,M_w(\WordP(w,d,N))\right) \geq \frac{1-\epsilon_N}{1+\epsilon_N}\, \frac{1}{P_\times(b/2,N,N)}.
\]

We define the predicate
\[
    \unbiased(w,d,b) \coloneqq \forall t\leq d\, \left| \sum_{j=1}^t w_j \right|_1 \leq b,
\]
where each $w_j$ is obtained from the corresponding segment of the string $w$. This is a $\Sigma_0^B(\calL_{\FNCTwo})$ predicate.

Finally, fix the constant $c$ occurring in the lower-bound construction. Choose a standard constant $N_0$ sufficiently large that, for every $N\geq N_0$, if
\[
    n(N)\coloneqq\lfloor\log_2N\rfloor, \qquad d(N)\coloneqq \lfloor c\log_2n(N)\rfloor, \qquad
    k(N)\coloneqq10d(N),
\]
then
\[
    k(N)d(N)\leq n(N).
\]
For every $w\in A(k(N))^{d(N)}$ we then have $L(w)\leq k(N)d(N)\leq n(N)$, and hence $N_{sm}(w)=2^{L(w)} \leq2^{n(N)} \leq N$. Thus all coefficient vectors and coefficient matrices occurring above have length polynomially bounded in $N$.

With these definitions, $\VNCTwo$ proves that the coefficient matrix of the word polynomial has high rank:
\[
    \begin{aligned}
        \forall N\geq N_0\ \forall w\in A(k(N))^{d(N)},\qquad &\unbiased(w,d(N),k(N)) \ \longrightarrow\\
        & \relrank\!\left(w,d(N),N, M_w(\WordP(w,d(N),N)) \right)\\
        &\qquad \geq \frac{1-2^{-N}}{1+2^{-N}}\, \frac{1}{P_\times(k(N)/2,N,N)}.
    \end{aligned}
\]
All quantities in this displayed statement are rational or integer-valued and have encodings of length polynomially bounded in $N$.

\subsubsection{The algorithm computing the coefficient vector of a constant-depth circuit} 

Before proving the field-independent set-multilinearization theorem of Forbes~\cite{For24}, we first formalise coefficient extraction in $\VNCTwo$. 

Throughout this section, $N$ is the global size parameter. We choose $n=n(N)=\lfloor\log_2N\rfloor$ and $d=d(N)=\lfloor c \log_2 n\rfloor$ where $c$ is a fixed constant. All circuits to which the coefficient-extraction procedure is applied have encoding length polynomially bounded in $N$. Let $C$ be such a circuit, and let $q(C)$ denote the number of variables occurring in $C$. The number of monomials in these variables of total degree at most $d$ is $D(C,d)=\binom{q(C)+d}{d}.$ For all applications below, the parameters are chosen so that $D(C,d)\leq N$. Indeed, the number of variables and the circuit size are bounded by a fixed polynomial in $n=\Theta(\log N)$, while $d=O(\log n)$, and therefore
\[ 
    \binom{q(C)+d}{d} \leq (q(C)+d)^d = 2^{O((\log n)^2)} \leq N 
\] for all sufficiently large $N$. For a polynomial $f$ in the variables of $C$, let $\coeff_{\leq d,N}(f)$ denote its coefficient vector on all monomials of total degree at most $d$, in some fixed ordering, padded with zeros to $N$ coordinates. 

We define an $\FNCTwo$ string function \[ \CoeffT(V,G,I,E,d,N) \] which, given an encoding $(V,G,I,E)$ of a constant-depth algebraic circuit $C$, outputs $\coeff_{\leq d,N}(C)$. Thus $\CoeffT$ does not attempt to compute coefficients of monomials of degree greater than $d$. 

We compute the truncated coefficient vector gate by gate. For a gate $v$, write $C_v$ for the polynomial computed at $v$, and let $\coeff_{\leq d,N}(C_v)$ denote its truncated coefficient vector. 

\begin{itemize}[leftmargin=*] 
    \item If $v$ is an input gate representing a scalar $a$, then $\coeff_1(C_v)=a,$ and every other coefficient is zero. 
    \item If $v$ is an input gate representing a variable $x_j$, then $\coeff_{x_j}(C_v)=1,$ and every other coefficient is zero. 
    \item If $v$ is a $+$-gate with incoming gates $u_1,\ldots,u_t$, then for every monomial $M$ of degree at most $d$, \[ \coeff_M(C_v) = \sum_{i=1}^t \coeff_M(C_{u_i}). \] 
    \item Suppose that $v$ is a $\times$-gate with incoming gates $u_1,\ldots,u_t$. For polynomials $f$ and $g$, define their truncated product by 
    \[ 
        f*_{\leq d}g \coloneqq \operatorname{tr}_{\leq d}(fg), 
    \] 
    where $\operatorname{tr}_{\leq d}$ deletes all monomials of total degree greater than $d$. For every monomial $M$ of degree at most $d$, 
    \[ 
        \coeff_M(f*_{\leq d}g) = \sum_{M_1M_2=M} \coeff_{M_1}(f)\coeff_{M_2}(g), 
    \] 
    where the sum ranges over all pairs of monomials $(M_1,M_2)$ whose product is $M$. If $M=x_1^{a_1}\cdots x_q^{a_q},$ then the number of such pairs is 
    \[ 
        \prod_{j=1}^q(a_j+1) \leq 2^{\sum_j a_j} = 2^{\deg M} \leq 2^d. 
    \] 
    
    Hence the coefficient of every truncated product can be computed using only polynomially many arithmetic operations in $N$. For a product gate of arbitrary fan-in $t$, we compute 
    \[ 
        \operatorname{tr}_{\leq d} \bigl(C_{u_1}\cdots C_{u_t}\bigr) 
    \] by a balanced tree of truncated binary products. This requires $O(\log t)$ stages. Since $t$ is bounded by the circuit size and the circuit encoding has length polynomially bounded in $N$, the resulting computation is in $\FNCTwo$. \end{itemize} 
    
    The truncation does not affect coefficients of degree at most $d$. More precisely, for arbitrary polynomials $f$ and $g$, 
    \[ 
        \operatorname{tr}_{\leq d}(fg) = \operatorname{tr}_{\leq d} \bigl(\operatorname{tr}_{\leq d}(f)\operatorname{tr}_{\leq d}(g) \bigr), 
    \] because every monomial occurring in a product has degree equal to the sum of the degrees of its factors. Thus a term of degree greater than $d$ at an intermediate gate can never contribute to a coefficient of degree at most $d$ at a later multiplication gate. The analogous statement for addition is immediate. Since the algebraic circuit has constant depth, we apply this procedure layer by layer. Addition of polynomially many encoded field elements, computation of a binary truncated product of coefficient vectors, and balanced iteration of truncated products are all computable in $\FNCTwo$. Hence $\CoeffT\in\FNCTwo.$ Moreover, $\VNCTwo$ proves the correctness of the procedure: 
    \[ 
        \VNCTwo\vdash \CoeffT(V,G,I,E,d, N) = \coeff_{\leq d,N}(C). 
    \] 
    
    If the polynomial computed by $C$ has total degree at most $d$, this is its full coefficient vector, padded to $N$ coordinates. All intermediate coefficient vectors have at most $N$ coordinates and total encoding length polynomially bounded in $N$. All auxiliary indices and strings used by the computation have length polynomially bounded in $N$. Thus $\CoeffT$ is represented by a function symbol in $\overline{\VNCTwo}$ and may be used in $\Sigma_0^B(\calL_{\FNCTwo})$ formulas.

\subsubsection{Field-independent set-multilinearization}

We now formalise the field-independent set-multilinearization argument of~\cite{For24}. Recall that $N$ is our global size parameter and 
\[
    n=n(N)\coloneqq \lfloor\log_2N\rfloor \qquad\text{and}\qquad d=d(N)\coloneqq \lfloor c\log_2n\rfloor,
\]
where $c$ is a fixed constant.

Let $s$ be a fixed polynomial bounding the size of the original algebraic circuits. Since $s(n)=n^{O(1)}$ and $d=O(\log n)$, there is a fixed constant $C$ such that $\binom{d\,s(n)+d}{d}  \leq (d\,s(n)+d)^d \leq 2^{C(\log n)^2}.$

Since $C(\log n)^2\leq n$ for all sufficiently large $n$, there is a standard constant $N_0$ such that, for every $N\geq N_0$, $\binom{d\,s(n)+d}{d}\leq 2^n\leq N.$ Thus all the coefficient vectors with degree at most $d$ appearing are polynomially bounded by the coefficient-vector length $N$.

We will construct an $\FNCTwo$ function $\setmultilinear(V,G,I,E,X,d,N)$ which, given an algebraic circuit $C$ encoded by $(V,G,I,E)$ and a string $X$ encoding a partition of the variables of $C$ into $d$ sets, outputs a set-multilinear circuit computing the set-multilinear projection of the polynomial computed by $C$. If $C$ has product-depth $\Delta$, the resulting circuit has product-depth at most $2\Delta$ and size $\poly(s,d^d)$.

\begin{definition}
    Let $X=(x_{i,j})$ be an $a\times m$ symbolic matrix of pairwise distinct variables, where $a\leq m$. The \emph{rectangular Permanent} is
    \[
        \perm_{a\times m}(X) \coloneqq \sum_{\sigma:[a]\hookrightarrow[m]} x_{1,\sigma(1)}\cdots x_{a,\sigma(a)},
    \]
    where the sum ranges over all injective maps $\sigma:[a]\hookrightarrow[m]$.
\end{definition}

The polynomial $\perm_{a\times m}(X)$ is set-multilinear with respect to the partition of its variables into the $a$ rows of $X$.

\begin{theorem}[Binet--Minc Identity~\cite{minc1979evaluation}]
    \label{theorem: Binet--Minc}
    Let $X=(x_{i,j})$ be an $a\times m$ symbolic matrix, where $a\leq m$. Then
    \begin{equation}
        \label{equa: Binet--Minc Identity}
        \perm_{a\times m}(X) = \sum_{\mathcal{F}\in P_a} (-1)^{a-|\mathcal{F}|} \prod_{S\in\mathcal{F}}(|S|-1)! \left( \sum_{j=1}^m \prod_{i\in S}x_{i,j} \right),
    \end{equation}
    where $P_a$ denotes the set of all set partitions of $[a]$ into nonempty parts.
\end{theorem}

Every polynomial occurring in \eqref{equa: Binet--Minc Identity} has total degree at most $a$. In all applications below, we have $a\leq d$ and $m\leq s(n).$ Hence the number of variables is at most $d \cdot s(n)$, and by the choice of $N_0$, the coefficient vectors of degree at most $d$ of both sides of \eqref{equa: Binet--Minc Identity} have length at most $N$. Therefore the standard coefficient-wise proof of the Binet--Minc identity can be formalised in $\VNCTwo$ using the function $\CoeffT$ defined above.

\begin{definition}
    For $k\geq1$, define the Stirling number of the second kind by
    \[
        S(a,k) = \frac{1}{k!} \sum_{j=0}^k (-1)^{k-j}\binom{k}{j}j^a.
    \]
    We set $S(0,0)=1$ and $S(a,0)=0$ for $a > 0$.
    The $a$th Bell number is
    \[
        B_a\coloneqq\sum_{k=0}^aS(a,k).
    \]
\end{definition}

\begin{lemma}
    \label{lemma: Stirling above diagonal}
    For all integers $0\leq i<p$, $\sum_{j=0}^{p} (-1)^{p-j} \binom{p}{j}j^i = 0$. Consequently, $S(i,p)=0$ whenever $p>i$.
\end{lemma}

\begin{proof}
    We prove the first identity by strong induction on $i$. For $i=0$,
    \[
        \sum_{j=0}^{p} (-1)^{p-j}\binom{p}{j} = (1-1)^p = 0
    \]
    for every $p>0$.
    Now let $i\geq1$ and assume the statement holds for every exponent smaller than $i$. Let $p>i$. Since the $j=0$ term vanishes,
    \[
        \begin{aligned}
            \sum_{j=0}^{p}(-1)^{p-j}\binom{p}{j}j^i &= \sum_{j=1}^{p}(-1)^{p-j}\binom{p}{j}j^i\\
            &= p\sum_{j=1}^{p}(-1)^{p-j}\binom{p-1}{j-1} j^{i-1}.
        \end{aligned}
    \]
    Writing $q=j-1$, this is
    \[
        p \sum_{q=0}^{p-1} (-1)^{p-1-q} \binom{p-1}{q} (q+1)^{i-1}.
    \]
    By the Binomial Theorem,
    \[
        (q+1)^{i-1}= \sum_{r=0}^{i-1} \binom{i-1}{r}q^r.
    \]
    Therefore the preceding expression equals
    \[
        p\sum_{r=0}^{i-1} \binom{i-1}{r} \sum_{q=0}^{p-1} (-1)^{p-1-q} \binom{p-1}{q}q^r.
    \]
    Since $p>i$, we have $p-1>r$ for every $r\leq i-1$. Each inner sum is therefore zero by the induction hypothesis. This proves the first identity. The conclusion about $S(i,p)$ follows immediately from its explicit formula.
\end{proof}

\begin{lemma}
    \label{lemma: Bell coarse upper bound}
    $\VNCTwo$ proves that for every $a\geq1$, $B_a\leq a^a.$
\end{lemma}

\begin{proof}
    We use the recurrence $B_{a+1} = \sum_{i=0}^{a}\binom{a}{i}B_i$. We first derive it from the explicit formula for the Stirling numbers. For $k\geq1$,
    \begin{align*}
        S(a+1,k) &= \frac{1}{k!} \sum_{j=1}^{k}(-1)^{k-j} \binom{k}{j}j^{a+1}\\
        &= \frac{1}{(k-1)!} \sum_{q=0}^{k-1} (-1)^{k-1-q} \binom{k-1}{q}(q+1)^a\\
        &= \sum_{i=0}^{a}\binom{a}{i}S(i,k-1).
    \end{align*}
    Hence
    \begin{align*}
        B_{a+1} &= \sum_{k=1}^{a+1}S(a+1,k) = \sum_{i=0}^{a} \binom{a}{i} \sum_{p=0}^{a}S(i,p).
    \end{align*}
    By \cref{lemma: Stirling above diagonal}, $S(i,p)=0$ whenever $p>i$, and therefore $B_{a+1} = \sum_{i=0}^{a} \binom{a}{i}B_i.$

    We now prove $B_a\leq a^a$ by induction on $a$. The case $a=1$ is immediate. Assume that $B_i\leq i^i$ for every $1\leq i\leq a$. Using $B_0=1$ and the bound $i^i\leq a^i$ for $i\leq a$, we obtain
    \[
        \begin{aligned}
            B_{a+1} &=\sum_{i=0}^{a}\binom{a}{i}B_i\\
                &\leq \sum_{i=0}^{a} \binom{a}{i}a^i\\
                &= (a+1)^a\\
                &\leq (a+1)^{a+1}.
        \end{aligned}
    \]
    This completes the induction.
\end{proof}

Thus we may use the estimate $B_a=O(a^a)$. This is weaker than the sharper estimate used in~\cite{For24}, but it is sufficient for our purposes.

The Binet--Minc identity immediately gives the following circuit upper bound.

\begin{corollary}
    \label{corollary: rectangular permanent formula}
    The rectangular Permanent $\perm_{a\times m}$ has a $\sum^{B_a}\prod^a\sum^m\prod^a$ formula of size $\poly(m,a^a)$. Moreover, the formula is set-multilinear with respect to the partition of the variables into rows.
\end{corollary}

We next need a variant of the rectangular Permanent in which some specified columns are required to occur in the image of the injection.

\begin{definition}
    Let $0\leq k\leq a\leq m$. Define the \emph{$k$-surjective rectangular Permanent} $\perm_{a\times m;k}(X)$ by
    \[
        \perm_{a\times m;k}(X) = \sum_{\substack{\sigma:[a]\hookrightarrow[m]\\
            \operatorname{im}(\sigma)\supseteq[k]}}x_{1,\sigma(1)}\cdots x_{a,\sigma(a)}.
    \]
\end{definition}

Again, $\perm_{a\times m;k}$ is set-multilinear with respect to the partition into rows.

\begin{lemma}
    \label{lemma: surjective permanent decomposition}
    Let $0\leq k\leq a\leq m$. Then
    \[
        \begin{aligned}
            \perm_{a\times m;k}(X) &= \sum_{R\in\binom{[a]}{k}} \perm_{k\times k} \bigl(X\restriction(R\times[k])\bigr)\\
                &\qquad\cdot \perm_{(a-k)\times(m-k)} \bigl(X\restriction(([a]\setminus R)\times([m]\setminus[k]))\bigr).
        \end{aligned}
    \]
This expression is set-multilinear with respect to the partition into rows.
\end{lemma}

\begin{proof}
    Every injection $\sigma:[a]\hookrightarrow[m]$ whose image contains $[k]$ has a unique set $R=\sigma^{-1}([k])$ of size $k$. The restriction $\sigma|_R$ is a bijection from $R$ to $[k]$, while the restriction to $[a]\setminus R$ is an injection into $[m]\setminus[k]$. Conversely, every such pair of maps determines a unique injection $\sigma$ whose image contains $[k]$. Summing the corresponding monomials gives the identity.
\end{proof}

\begin{corollary}
    \label{corollary: surjective permanent formula}
    Let $0\leq k\leq a\leq m$. Then the $k$-surjective rectangular Permanent $\perm_{a\times m;k}$ has a depth-$4$ formula of size $\poly(m,a^a)$ which is set-multilinear with respect to its rows.
\end{corollary}

\begin{proof}
    By \cref{lemma: surjective permanent decomposition,corollary: rectangular permanent formula}, for each $R\in\binom{[a]}{k}$ the two rectangular Permanents in the corresponding summand have formulas whose top addition gates have fan-in at most $B_k$ and $B_{a-k}$, respectively. We distribute the product of these two formulas over their top addition gates. This gives at most $B_kB_{a-k}$ summands for each $R$; after merging adjacent multiplication gates, the resulting formula still has depth $4$. Since there are $\binom{a}{k}$ choices for $R$, the resulting formula has size
    \[
        \poly\left(m,\binom{a}{k}B_kB_{a-k}\right).
    \]
    Using \cref{lemma: Bell coarse upper bound},
    \[
        \binom{a}{k}B_kB_{a-k}
        \leq 2^a k^k(a-k)^{a-k}
        \leq 2^a a^a.
    \]
    For $a\geq2$, $2^aa^a=(2a)^a\leq a^{2a}$, which is polynomial in $a^a$. The cases $a=0,1$ are immediate. Hence the formula has size $\poly(m,a^a)$.
\end{proof}

In the applications below, $a\leq d$ and $m\leq s(n)$. Therefore all circuits in the preceding construction have degree at most $d$, use at most $d\cdot s(n)$ variables, and have coefficient vectors of length at most $N$ for $N\geq N_0$. Hence the preceding identities and circuit constructions can be verified in $\VNCTwo$ using $\CoeffT$.

\begin{definition}
    Let the variables $\overline{x}$ be partitioned as $\overline{x} = \overline{x}_1 \sqcup\cdots\sqcup \overline{x}_d$. For $S\subseteq[d]$, a monomial is \emph{$S$-set-multilinear} if it contains exactly one variable from $\overline{x}_i$ for each $i\in S$ and no variables from $\overline{x}_j$ for $j\notin S$. Define the \emph{$S$-set-multilinear projection} $\pi_{sm,S}$ to be the linear map which fixes every $S$-set-multilinear monomial and sends every other monomial to zero.
    
    A monomial is called \emph{at-most-set-multilinear} if it is $S$-set-multilinear for some $S\subseteq[d]$. Define $\pi_{\neg sm}$ to be the linear map which sends every at-most-set-multilinear monomial to zero and fixes every other monomial.
\end{definition}

For $S=[d]$, we write
\[
    \pi_{sm}(f) \coloneqq \pi_{sm,[d]}(f).
\]
For $S=\emptyset$, $\pi_{sm,\emptyset}(f) = f(\overline{0}).$ Every monomial belongs to exactly one of the following classes: it is $S$-set-multilinear for a unique $S\subseteq[d]$, or it is not at-most-set-multilinear. Hence
\[
    f = \pi_{\neg sm}(f) + \sum_{S\subseteq[d]}\pi_{sm,S}(f).
\]

Addition is simulated by adding the corresponding set-multilinear parts.

\begin{lemma}
    \label{lemma: set multilinear addition}
    Let $f=\alpha_1f_1+\cdots+\alpha_mf_m.$ For every $S\subseteq[d]$, $\pi_{sm,S}(f)=\sum_{j=1}^m \alpha_j\pi_{sm,S}(f_j)$. Thus, the polynomial $\pi_{sm,S}(f)$ can be computed from the inputs $\{\pi_{sm,S}(f_j) : j \in [m]\}$ by a depth-$1$ set-multilinear sum circuit of size $O(m)$. Simultaneously maintaining all subsets $S\subseteq[d]$ incurs an additional factor $2^d$.
\end{lemma}

We now simulate multiplication gates.

\begin{lemma}
    \label{lemma: set multilinear multiplication}
    Let $f=f_1f_2\cdots f_m$, and let $S\subseteq[d]$ with $|S|=\ell$. Then the polynomial $\pi_{sm,S}(f)$ can be computed from the inputs $\left\{ \pi_{sm,T}(f_j) : T\subseteq[d],\ j\in[m] \right\}$ by a depth-$4$ set-multilinear circuit of size $\poly(m,\ell^\ell,2^d)$.
\end{lemma}

\begin{proof}
    Rearrange the factors so that $ f_1(\overline0) = \cdots = f_k(\overline0) = 0$ and $f_{k+1}(\overline0), \ldots, f_m(\overline0) \neq0$ for some $0\leq k\leq m$.

    For $i>k$, define
    \[
        g_i = \frac{f_i}{f_i(\overline0)},
    \]
    and for $i\leq k$, put $g_i=f_i.$ Thus, for $i \leq k$, $ g_i(\overline0)=0$ whereas $g_i(\overline0)=1$ for $i > k$. If $g=\prod_{i=1}^m g_i$, then
    \[
        f = \left(\prod_{i>k}f_i(\overline0) \right)g,
    \]
    and therefore
    \[
        \pi_{sm,S}(f) = \left(\prod_{i>k}f_i(\overline0) \right) \pi_{sm,S}(g).
    \]
    If $k>\ell$, then $\pi_{sm,S}(g)=0$. Indeed, each of the first $k$ factors has zero constant term, so a term contributing to an at-most-set-multilinear monomial must choose a nonconstant set-multilinear part from each of these $k$ factors. These parts must involve pairwise disjoint nonempty subsets of $S$, which is impossible when $k>|S|$.

    Assume that $ k\leq\ell$ henceforth. We have
    \[
        \begin{aligned}
            g = \prod_{i=1}^{k} \left(\pi_{\neg sm}(g_i)+ \sum_{\emptyset\neq T\subseteq[d]} \pi_{sm,T}(g_i) \right) \cdot \prod_{i=k+1}^{m} \left(\pi_{\neg sm}(g_i)+1+\sum_{\emptyset\neq T\subseteq[d]}\pi_{sm,T}(g_i) \right).
        \end{aligned}
    \]

    Any product term containing a monomial from $\pi_{\neg sm}(g_i)$ cannot contribute to an at-most-set-multilinear monomial, since multiplying a non-at-most-set-multilinear monomial by another monomial cannot make it at-most-set-multilinear.

    Therefore, a term contributing to $\pi_{sm,S}(g)$ is obtained by choosing pairwise disjoint nonempty subsets of $S$ from distinct factors, whose union is $S$. Equivalently, these subsets form a set partition $\mathcal{F}\in P_S$. Since the first $k$ factors have no constant term, each index in $[k]$ must be used.

    Thus
    \[
        \pi_{sm,S}(g) = \sum_{\substack{ \mathcal{F}\in P_S\\ k\leq|\mathcal{F}|\leq m }} \;
        \sum_{\substack{\sigma:\mathcal{F}\hookrightarrow[m]\\
        \operatorname{im}(\sigma)\supseteq[k]}}
        \prod_{T\in\mathcal{F}} \pi_{sm,T} \bigl(g_{\sigma(T)}\bigr).
    \]
    For a partition $\mathcal{F}$ of $S$, define the $|\mathcal{F}|\times m$ matrix $A_{\mathcal{F}}$ by $(A_{\mathcal{F}})_{T,j} = \pi_{sm,T}(g_j)$, for $T \in \mathcal{F}$ and $j \in [m]$. Then
    \[
        \pi_{sm,S}(g) = \sum_{\substack{\mathcal{F}\in P_S\\k\leq|\mathcal{F}|\leq m}} \perm_{|\mathcal{F}|\times m;k}(A_{\mathcal{F}}).
    \]
    Since $|\mathcal{F}|\leq|S|=\ell$, \cref{corollary: surjective permanent formula} gives a depth-$4$ set-multilinear circuit of size $\poly(m,\ell^\ell)$ for each summand. There are $B_\ell\leq\ell^\ell$ partitions of $S$. Hence $\pi_{sm,S}(g)$, and therefore $\pi_{sm,S}(f)$, has a depth-$4$ set-multilinear circuit of size $\poly(m,\ell^\ell,2^d)$.
\end{proof}

For the formalisation, it is convenient to introduce the coefficient-vector version of the set-multilinear projection. Given a circuit $C$ encoded by the set $V,G,I,E$, the partition $X$, and $S\subseteq [d(N)]$, define $\SMCoeffT(V,G,I,E,X,S,N)$ to be the $N$-coordinate coefficient vector obtained from $\CoeffT(V,G,I,E,d(N),N)$ by retaining the coefficient of a monomial precisely when that monomial is $S$-set-multilinear with respect to $X$, and replacing all other coefficients by zero.

Since decoding a monomial and checking whether it is $S$-set-multilinear are in $\FNCTwo$, the function $\SMCoeffT$ is in $\FNCTwo$. Semantically,
\[
    \SMCoeffT(V,G,I,E,X,S,N) =  \coeff_{\leq d,N} \bigl( \pi_{sm,S}(C) \bigr).
\]
For $S=[d(N)]$, we have
\[
    \SMCoeffT(V,G,I,E,X,[d(N)],N) = \coeff_{\leq d,N} \bigl(\pi_{sm}(C)\bigr).
\]

Applying \cref{lemma: set multilinear addition,lemma: set multilinear multiplication} gate by gate gives the field-independent set-multilinearization theorem.

\begin{corollary}[Field-independent set-multilinearization~\cite{For24}]
    \label{corollary: formal set multilinearization}
    Let $\F$ be a field. Let the variables $\overline{x}$ be partitioned as $\overline{x} = \overline{x}_1 \sqcup\cdots\sqcup \overline{x}_d$. Suppose that a polynomial $f\in\F[\overline{x}]$ is computed by a size-$s$ algebraic circuit of product-depth $\Delta$. Then $\pi_{sm}(f)$ is computed by a set-multilinear circuit of product-depth at most $2\Delta$ and size $\poly(s,d^d).$
\end{corollary}

We now state the formal version using the global parameter $N$. Let $s$ be the fixed standard polynomial bounding the size of the original circuit. There is a fixed standard polynomial $s'$ and an $\FNCTwo$ function
\[
    \setmultilinear(V,G,I,E,X,d,N)
\]
such that, for sufficiently large $N$, if $n=\lfloor\log_2N\rfloor$ and $d=\lfloor c\log_2n\rfloor$, and if $(V,G,I,E)$ encodes a circuit $C$ of size at most $s(n)$ and product-depth at most $\Delta$, then
\[
    C' = \setmultilinear(V,G,I,E,X,d,N)
\]
encodes a set-multilinear circuit satisfying $C'$ computing $\pi_{sm}(C)$ with $\Size(C') < s'\bigl(s(n),d^d\bigr)$, and $\ProdDepth_{2\Delta}(C')$.

More precisely, working in the universal conservative extension $\overline{\VNCTwo}$, we have
\[
    \begin{aligned}
        \forall N\geq N_0\; \forall V,G,I,E,X\; \Bigl( &\Circuit(V,G,I,E) \land \Size(V,G,I,E)<s(n(N)) \\
            &\land \ProdDepth_{\Delta}(V,G,I,E) \land \Partition(X,d(N))\\
            &\longrightarrow \CoeffT(C',d(N),N) = \SMCoeffT(V,G,I,E,X,[d(N)],N)\\
            &\qquad\land \SetMultCircuit(C',X,d(N))\\
            &\qquad\land \Size(C')< s'\bigl(s(n(N)), d(N)^{d(N)} \bigr)\\
            &\qquad\land \ProdDepth_{2\Delta}(C') \Bigr),
    \end{aligned}
\]
where $C'= \setmultilinear \bigl(V,G,I,E,X,d(N),N \bigr).$ Here, $n(N)$ and $d(N)$ are the parameter functions fixed above, and $N_0$ is a fixed standard constant chosen so that all relevant encodings have length polynomially bounded in $N$. The predicate $\Circuit(V,G,I,E)$ checks that $(V,G,I,E)$ is a valid encoding of an algebraic circuit, $\Size(V,G,I,E)$ returns its size, and $\ProdDepth_{\Delta}(V,G,I,E)$ checks that its product-depth is at most $\Delta$. The predicate $\Partition(X,d)$ checks that $X$ encodes a partition of the circuit variables into $d$ sets. The function $\CoeffT(C',d(N),N)$ outputs the degree-at-most-$d(N)$ coefficient vector of $C'$, padded to length $N$, while $\SMCoeffT(V,G,I,E,X,[d(N)],N)$ outputs the corresponding coefficient vector of the set-multilinear projection of the polynomial computed by $(V,G,I,E)$. Finally, $\SetMultCircuit(C',X,d(N))$ checks that $C'$ is set-multilinear with respect to the partition encoded by $X$. The functions $s$ and $s'$ are fixed standard polynomials giving, respectively, the assumed size bound on the original circuit and the size bound resulting from the set-multilinearization construction.

Since $2^{d(N)} = n(N)^{O(1)}$ and $d(N)^{d(N)} = 2^{O(\log n(N)\log\log n(N))} \leq N$ for all sufficiently large $N$, the size of $C'$ and the lengths of all strings used in the construction are polynomially bounded in $N$. The functions which enumerate subsets and set partitions, construct the Binet--Minc circuits, form the surjective-Permanent circuits, and assemble the gate-by-gate simulation are therefore in $\FNCTwo$.

Consequently, $\setmultilinear$ is represented by a function symbol in $\overline{\VNCTwo}$, and the displayed correctness statement can be expressed by a $\Sigma_0^B(\calL_{\FNCTwo})$ formula. By conservativity, the corresponding statement in the original language is provable in $\VNCTwo$.

\subsubsection{Lower bounds for small-depth circuits}

Recall that $n=n(N)\coloneqq\lfloor\log_2N\rfloor$, $d=d(N)\coloneqq \lfloor c\log_2n\rfloor$ and $k=k(N)\coloneqq10d$. For the lower-bound construction, define $ a(k)\coloneqq\max\{h\leq k:2h^2\leq k^2\}$, and let $A(k)\coloneqq\{a(k),-k\}$. Namely, $a(k)=\lfloor k/\sqrt{2}\rfloor$.

For every fixed standard integer $\delta \geq 1$, define
\[
    D_\delta(m) \coloneqq \max\{q\leq m:q^{2^\delta-1}\leq m\}.
\]
Thus $D_\delta(m) = \left\lfloor m^{1/(2^\delta-1)}\right\rfloor = \Theta\left(m^{1/(2^\delta-1)}\right)$. Since $\delta$ is a fixed standard constant, $D_\delta$ is an $\FNCTwo$ function.

We put $\epsilon_N\coloneqq2^{-N}$ and $\eta_N\coloneqq \frac{1+\epsilon_N}{1-\epsilon_N}$. By the relative-rank properties proved in the previous section, $\VNCTwo$ proves that relative rank is imbalance, which is
\[
    \relrank(u,m,N,M_u(f)) \leq \frac{\eta_N} {P_\times(|s(u)|_1/2,N,N)},
\]
and that if $f=f_1\cdots f_t$ is a set-multilinear product with respect to a partition
$(S_1,\ldots,S_t)$, then
\[
    \relrank(u,m,N,M_u(f)) \leq \eta_N^t \prod_{i=1}^t \relrank \left( u\restriction S_i, |S_i|, N, M_{u\restriction S_i}(f_i) \right).
\]
We will also repeatedly use the following direct consequences of the approximation bounds. For $z_1,\ldots,z_t\geq0$,
\[
    \prod_{i=1}^t \frac{1}{P_\times(z_i,N,N)} \leq \eta_N^t \frac{1}{P_\times(\sum_{i=1}^t z_i,N,N)},
\]
and if $0\leq z\leq z'$, then
\[
    \frac{1}{P_\times(z',N,N)} \leq \eta_N \frac{1}{P_\times(z,N,N)}.
\]
Indeed, both inequalities follow directly from
\[
    (1-\epsilon_N)2^z \leq P_\times(z,N,N) \leq (1+\epsilon_N)2^z.
\]

\begin{lemma}
    \label{lemma: set multilinear circuit lower bound}
    Let $w\in A(k)^d$ be any $k$-unbiased word, where $k\geq10d$ and $d=\lfloor c\log n\rfloor$. Any set-multilinear circuit $C$ of product-depth $\Delta$ computing $P_w$ has size at least $n^{\Omega\left(\frac{D_\Delta(d)}{\Delta}\right)} = n^{\Omega\left(\frac{d^{1/(2^\Delta-1)}}{\Delta}\right)}$.
\end{lemma}

We first prove the relative-rank upper bound for formulas. It is important that the formal statement is uniform in the degree $m$, since the induction will be applied to factors of degree smaller than $d$.

\begin{lemma}
    \label{lemma: relative rank upper bound}
    Let $1\leq m\leq d$ and let $u\in A(k)^m$, where $k\geq10d$. For a constant $\Delta\geq1$, any set-multilinear formula $C$ of product-depth $\Delta$, degree $m$, and size at most $s$ satisfies
    \[
        \relrk_u(C) \leq s\cdot 2^{-\frac{kD_\Delta(m)}{20}}.
    \]
    Moreover, $\VNCTwo$ proves the corresponding rational inequality
    \[
        \relrank(u,m,N,M_u(C)) \leq \frac{s\eta_N^{5\Delta m}}{P_\times(kD_\Delta(m)/20,N,N)}.
    \]
\end{lemma}

\begin{proof}
    We prove both statements simultaneously by induction on product depth $\Delta$. If $\Delta=1$, then $C=C_1+\cdots+C_t$, where each $C_i$ is a product of $m$ linear forms
    \[
        C_i=L_1\cdots L_m.
    \]
    Since every entry of $u$ is either $a(k)$ or $-k$, and $a(k)\geq\frac{k}{2}$, we have
    \[
        \relrk_u(C_i) = \prod_{j=1}^m \relrk_{u_j}(L_j) \leq \prod_{j=1}^m 2^{-|u_j|_1/2} \leq 2^{-km/4}.
    \]
    By subadditivity,
    \[
        \relrk_u(C) \leq s2^{-km/4} \leq s2^{-km/20}.
    \]
    Since $D_1(m)=m$, this proves the first statement.

    In the theory, multiplicativity and imbalance give
    \begin{align*}
        \relrank(u,m,N,M_u(C_i)) &\leq \eta_N^m\prod_{j=1}^m \relrank(u_j,1,N,M_{u_j}(L_j))\\
        &\leq \eta_N^{2m} \prod_{j=1}^m \frac{1}{P_\times(|u_j|_1/2,N,N)}\\
        &\leq \frac{\eta_N^{3m}}{P_\times(\frac12\sum_{j=1}^m|u_j|_1,N,N)}.
    \end{align*}
    Since $|u_j|_1\geq k/2$, $\frac12\sum_{j=1}^m|u_j|_1 \geq  \frac{km}{4}$. Hence
    \[
        \relrank(u,m,N,M_u(C_i)) \leq \frac{\eta_N^{3m+1}}{P_\times(km/4,N,N)}
        \leq \frac{\eta_N^{3m+2}}{P_\times(km/20,N,N)}
        \leq \frac{\eta_N^{5m}}{P_\times(kD_1(m)/20,N,N)}.
    \]
    
    By subadditivity,
    \[
        \relrank(u,m,N,M_u(C)) \leq \frac{s\eta_N^{5m}}{P_\times(kD_1(m)/20,N,N)}.
    \]
    
    Now assume the result holds for product-depth at most $\Delta$, and let $C$ have product-depth $\Delta+1$. Write
    \[
        C=C_1+\cdots+C_t.
    \]
    
    Put $q\coloneqq D_{\Delta+1}(m)$ and $T_\Delta\coloneqq q^{2^\Delta-1}$. By the definition of $D_{\Delta+1}$, $ q^{2^{\Delta+1}-1}\leq m$. Therefore $qT_\Delta^2 = q^{2^{\Delta+1}-1} \leq m$, and in particular
    \[
        T_\Delta^2\leq m\leq d.
    \]

    Each $C_i$ of size $s_i$ is said to be of type $1$ if one of its factors has degree at least $T_\Delta$, and otherwise it is of type $2$.

    \begin{itemize}[leftmargin=*]

    \item \textbf{Type 1.}
    Suppose
    \[
        C_i=C_{i,1}\cdots C_{i,t_i},
    \]
    and, after reordering, let $m'\coloneqq\deg(C_{i,1})\geq T_\Delta$. Let $S_{i,1}$ be the set of variable blocks associated with $C_{i,1}$, so $|S_{i,1}|=m'$. Since $q^{2^\Delta-1} = T_\Delta \leq m'$, we have $D_\Delta(m')\geq q$. The induction hypothesis therefore gives
    \[
        \begin{aligned}
            \relrk_u(C_i) &\leq \relrk_{u\restriction S_{i,1}}(C_{i,1}) \\
                &\leq s_{i,1}2^{-kD_\Delta(m')/20}\\
                &\leq s_i2^{-kq/20}.
        \end{aligned}
    \]

    We now prove the corresponding statement in the theory. Constant factors may be absorbed into one of the nonconstant factors; if one factor is zero, then $C_i=0$. Thus we may
    assume that every factor has positive degree, and consequently $t_i\leq m$.
    
    For every set-multilinear polynomial $f$ we have
    \[
        \relrank(v,r,N,M_v(f)) \leq \frac{1}{1-\epsilon_N} \leq \eta_N,
    \]
    since $\rank(M_v(f))\leq2^{L(v)/2}$.

    Hence approximate multiplicativity gives
    \begin{align*}
        \relrank(u,m,N,M_u(C_i))
        &\leq \eta_N^{t_i} \prod_{j=1}^{t_i} \relrank\left(u\restriction S_{i,j}, |S_{i,j}|, N, M_{u\restriction S_{i,j}}(C_{i,j})\right)\\
        &\leq \eta_N^{2m} \relrank \left(u\restriction S_{i,1}, m', N, M_{u\restriction S_{i,1}}(C_{i,1})\right).
    \end{align*}
    Applying the induction hypothesis,
    \[
        \relrank(u,m,N,M_u(C_i)) \leq \frac{s_{i,1}\eta_N^{2m+5\Delta m'}}{P_\times(kD_\Delta(m')/20, N, N)}.
    \]
    Since $m'\leq m$ and $D_\Delta(m')\geq q$,
    \[
        \relrank(u,m,N,M_u(C_i))
        \leq \frac{s_i\eta_N^{2m+5\Delta m+1}}{P_\times(kq/20, N, N)}
        \leq \frac{s_i\eta_N^{5(\Delta+1)m}}{P_\times(kq/20, N, N)}.
    \]

    \item \textbf{Type 2.}
    Suppose
    \[
        C_i=C_{i,1}\cdots C_{i,t_i},
    \]
    where each nonconstant factor $C_{i,j}$ has degree $\deg(C_{i,j})<T_\Delta$. As above, we absorb all nonzero constant factors. Since $C_i$ is set-multilinear of degree $m = \sum_{j=1}^{t_i}\deg(C_{i,j})$. Therefore $t_iT_\Delta>m$, and hence $t_i>\frac{m}{T_\Delta}$.

    As the formula is set-multilinear, $(S_{i,1},\ldots,S_{i,t_i})$ forms a partition of $[m]$, where $C_{i,j}$ is set-multilinear with respect to the variable blocks indexed by $S_{i,j}$. Let $u^{ij}\coloneqq u\restriction S_{i,j}$. Since $C_{i,j}$ is set-multilinear, we have
    \[
        |u^{ij}|=|S_{i,j}|=\deg(C_{i,j}).
    \]
    Let $a_{ij}$ be the number of entries of $u^{ij}$ equal to $a(k)$, and let $b_{ij}$ be the number of entries equal to $-k$. Then
    \[
        a_{ij}+b_{ij}  = |S_{i,j}| = \deg(C_{i,j}),
    \]
    so in particular $ b_{ij} = \deg(C_{i,j})-a_{ij}$. 

    Therefore, we have
    \[
        s(u^{ij}) \coloneqq \sum_{\ell\in S_{i,j}}u_\ell = a_{ij}a(k)-b_{ij}k.
    \]

    We use the following elementary Diophantine estimate.

    \begin{claim}
        Let $T\geq1$ and let $a,b\geq0$ satisfy $1\leq a+b<T$. If $k\geq10T^2$, then
        \[
            |aa(k)-bk|_1 \geq \frac{k}{10T}.
        \]
    \end{claim}

    \begin{proof}
        Put $h=a(k)$. By the definition of $h$, we have
        \[
            2h^2\leq k^2<2(h+1)^2,
        \]
        and therefore
        \[
            0\leq k^2-2h^2<4h+2\leq5k.
        \]

        We first note that
        \[
            a^2\neq2b^2
        \]
        whenever $(a,b)\neq(0,0)$. Indeed, the equation
        $a^2=2b^2$ has no nonzero integer solution. This follows
        by the standard parity descent: $a$ must be even, and
        then $b$ must also be even, producing a strictly smaller
        solution. The same argument is formalizable in $\Vzero$
        by bounded induction.

        Suppose, for contradiction, that
        \[
            10T|ah-bk|_1<k.
        \]
        Put $\delta\coloneqq ah-bk$. Since $a^2-2b^2$ is a nonzero integer,
        \[
            k^2 \leq |(a^2-2b^2)k^2|_1.
        \]
        On the other hand,
        \begin{align*}
            |(a^2-2b^2)k^2|_1
            &= \left| a^2(k^2-2h^2) + 2(a^2h^2-b^2k^2)\right|_1\\
            &= \left| a^2(k^2-2h^2) + 2\delta(ah+bk) \right|_1\\
            &\leq a^2|k^2-2h^2|_1 + 2|\delta|_1|ah+bk|_1.
        \end{align*}
        Since $a<T$ and $a+b<T$, we have
        \[
            a^2|k^2-2h^2|_1 < 5T^2k \leq \frac{k^2}{2},
        \]
        while
        \[
            |ah+bk|_1 < Tk.
        \]
        Therefore
        \[
            2|\delta|_1|ah+bk|_1 <  2\frac{k}{10T}Tk = \frac{k^2}{5}.
        \]
        Hence
        \[
            |(a^2-2b^2)k^2|_1 < \frac{7}{10}k^2,
        \]
        which contradicts
        \[
            |(a^2-2b^2)k^2|_1\geq k^2.
        \]
    \end{proof}

    Since
    \[
        a_{ij}+b_{ij} = \deg(C_{i,j})  < T_\Delta 
    \]
    and
    \[
        k\geq10d\geq10T_\Delta^2,
    \]
    the claim gives
    \[
        |s(u^{ij})|_1  \geq \frac{k}{10T_\Delta}.
    \]

    Therefore, using multiplicativity and imbalance, we have
    \begin{align*}
        \relrk_u(C_i)
        &= \prod_{j=1}^{t_i} \relrk_{u^{ij}}(C_{i,j}) \\
        &\leq \prod_{j=1}^{t_i} 2^{-|s(u^{ij})|_1/2}   \\
        &\leq 2^{-kt_i/(20T_\Delta)} \leq 2^{-km/(20T_\Delta^2)} \leq 2^{-kq/20},
    \end{align*}
    where the last inequality follows from $ qT_\Delta^2\leq m$.

    In the theory,
    \begin{align*}
        \relrank(u,m,N,M_u(C_i))
        &\leq \eta_N^{t_i} \prod_{j=1}^{t_i} \relrank \left(u^{ij}, |S_{i,j}|, N, M_{u^{ij}}(C_{i,j})\right)\\
        &\leq \eta_N^{2t_i} \prod_{j=1}^{t_i}\frac{1}{P_\times(|s(u^{ij})|_1/2, N, N)}\\
        &\leq \frac{\eta_N^{3t_i}}{P_\times(\frac{1}{2} \sum_{j=1}^{t_i}|s(u^{ij})|_1, N, N)}.
    \end{align*}
    Since
    \[
        \frac12 \sum_{j=1}^{t_i}|s(u^{ij})|_1 \geq \frac{kt_i}{20T_\Delta},
    \]
    we obtain
    \[
        \relrank(u,m,N,M_u(C_i)) \leq \frac{\eta_N^{3t_i+1}}{P_\times(kt_i/(20T_\Delta),N,N)}.
    \]
    Moreover,
    \[
        \frac{kt_i}{20T_\Delta} > \frac{km}{20T_\Delta^2} \geq \frac{kq}{20}.
    \]
    Hence
    \[
        \relrank(u,m,N,M_u(C_i)) \leq \frac{\eta_N^{3t_i+2}}{P_\times(kq/20, N, N)} \leq \frac{\eta_N^{5(\Delta+1)m}}{P_\times(kq/20, N, N)}.
    \]
\end{itemize}

Finally, by subadditivity,
\[
    \relrk_u(C) \leq s2^{-kq/20},
\]
and
\[
    \relrank(u,m,N,M_u(C)) \leq \frac{s\eta_N^{5(\Delta+1)m}}{P_\times(kq/20, N, N)}.
\]
Since $q=D_{\Delta+1}(m),$ this completes the induction.
\end{proof}

We now prove \cref{lemma: set multilinear circuit lower bound}.

\begin{proof}[Proof of \cref{lemma: set multilinear circuit lower bound}]
    Let $C$ be a size-$s$ set-multilinear circuit of product-depth $\Delta$ computing $P_w$. Unfolding $C$ into a formula gives a set-multilinear formula of the same product-depth and size at most $s^{2\Delta}$. Since $w$ is $k$-unbiased, we have $|s(w)|_1\leq k$. Since the coefficient matrix of the word polynomial has high rank, we have
    \[
        \relrk_w(P_w) = 2^{-|s(w)|_1/2} \geq 2^{-k/2}.
    \]
    Applying \cref{lemma: relative rank upper bound} gives $2^{-k/2} \leq s^{2\Delta} 2^{-kD_\Delta(d)/20}.$

    Therefore, we have $ s^{2\Delta} \geq 2^{k(D_\Delta(d)/20-1/2)}.$ For all sufficiently large $N$, $D_\Delta(d)\geq20$, and hence $s \geq 2^{\Omega \left(\frac{kD_\Delta(d)}{\Delta} \right)}$. Since $k=10d=\Theta(\log n)$, we obtain
    \[
        s \geq n^{\Omega\left(\frac{D_\Delta(d)}{\Delta}\right)} =
        n^{\Omega\left(\frac{d^{1/(2^\Delta-1)}}{\Delta}\right)}.
    \]

    We now give the corresponding formal argument in $\VNCTwo$. From the high-rank statement proved in the preceding section, we have
    \[
        \relrank \left(w, d, N, M_w(\WordP(w,d,N)) \right) \geq \frac{\eta_N^{-1}}{P_\times(k/2,N,N)}.
    \]
    Applying the formal part of \cref{lemma: relative rank upper bound} to the formula, we have
    \[
        \relrank\left(w, d, N, M_w(\WordP(w,d,N))\right) 
        \leq \frac{s^{2\Delta}\eta_N^{5\Delta d}}{P_\times(kD_\Delta(d)/20, N, N)}.
    \]
    Consequently,
    \[
        s^{2\Delta} \geq \eta_N^{-(5\Delta d+1)} \frac{P_\times(kD_\Delta(d)/20, N, N)}{P_\times(k/2, N, N)}.
    \]
    By the approximation property of $P_\times$, we have
    \[
        \frac{P_\times(kD_\Delta(d)/20, N, N)}{P_\times(k/2, N, N)}
        \geq \eta_N^{-1} 2^{k(D_\Delta(d)/20-1/2)}.
    \]
    Hence
    \[
        s^{2\Delta} \geq \eta_N^{-(5\Delta d+2)} 2^{k(D_\Delta(d)/20-1/2)}.
    \]

    Since $\eta_N =\frac{1+2^{-N}}{1-2^{-N}}$ and $d=O(\log\log N)$, there is a fixed standard constant $N_1$ such that $\VNCTwo$ proves, for every $N\geq N_1$,
    \[
        \eta_N^{5\Delta d+2}\leq2.
    \]
    For example, for sufficiently large $N$ we have $\eta_N\leq1+3\cdot2^{-N}$ and $6(5\Delta d+2)\leq2^N$. The desired estimate then follows from the Binomial Theorem. Thus the rational approximation loss does not affect the asymptotic lower bound.
\end{proof}

Finally, we combine the set-multilinear lower bound with the
field-independent set-multilinearization theorem from the previous
subsection.

Let $s(n)$ be a fixed polynomial. Suppose that a circuit $C$ of
size at most $s(n)$, depth at most $2\Delta$, and product-depth at
most $\Delta$ satisfies
\[
    \CoeffT(C,d(N),N) = \WordP(w,d,N).
\]
Since $\WordP(w,d,N)$ is supported on the $[d]$-set-multilinear
monomials, it follows that $\pi_{sm}(C)=P_w$.

Indeed, equality of the truncated coefficient vectors implies equality
of all $[d]$-set-multilinear coefficients, and every such monomial has
degree exactly $d$.

By the field-independent set-multilinearization theorem, there is a
set-multilinear circuit $C'$ computing $P_w$, of product-depth at most
$2\Delta$, and size $\poly(s(n),d^d)$.

Since $s(n)$ is polynomial and $d=\Theta(\log n)$, we have
\[
    \poly(s(n),d^d) = n^{O(\log\log n)}.
\]

On the other hand, \cref{lemma: set multilinear circuit lower bound}, applied with product-depth $2\Delta$, gives
\[
    \Size(C') \geq n^{\Omega\left(\frac{D_{2\Delta}(d)}{2\Delta} \right)}.
\]

For fixed $\Delta$,
\[
    D_{2\Delta}(d) = \Theta\left(d^{1/(2^{2\Delta}-1)}\right),
\]
and therefore
\[
    D_{2\Delta}(d) \gg \log\log n.
\]
Hence these two bounds are incompatible for all sufficiently large
$N$.

Consequently, there is a fixed standard constant $N_2$ such that
$\overline{\VNCTwo}$ proves that, when $s(n)$ is any fixed polynomial,

\[
    \begin{aligned}
        \forall N\geq N_2\; \forall w,V,G,I,E\;\Bigl(
        &w\in A(k(N))^{d(N)} \land \unbiased(w,d(N),k(N))\\
        &\land \Circuit_w(V,G,I,E) \land \Size(V,G,I,E)<s(n(N))\\
        &\land \Depth_{2\Delta}(V,G,I,E) \land \ProdDepth_\Delta(V,G,I,E) \\
        &\longrightarrow \CoeffT(V,G,I,E,d(N),N) \neq \WordP(w,d(N),N) \Bigr).
    \end{aligned}
\]

Here $\Circuit_w(V,G,I,E)$ means that $(V,G,I,E)$ is a valid algebraic circuit whose variables are the variables in $\overline{X}(w)$, and $\CoeffT$ and $\WordP$ use the same fixed $N$-coordinate monomial encoding.

All number and string parameters occurring above have encodings of length polynomially bounded in $N$. Since $\Delta$ and $c$ are fixed standard constants, the functions $n(N)$, $d(N)$, $k(N)$, $a(k)$, and $D_\delta$ are in $\FNCTwo$. Hence the matrix of the final displayed statement is a $\Sigma_0^B(\calL_{\FNCTwo})$ formula in $\overline{\VNCTwo}$, and by conservativity the corresponding statement is provable in $\VNCTwo$.

\subsubsection{Lower bounds for permanent}

Now, we show that $\VNCTwo$ is able to prove superpolynomial lower bounds for the Permanent polynomial against constant-depth algebraic circuits. We show that each word polynomial is a p-projection of a Permanent polynomial and that the projection can be generated uniformly and efficiently.

First, we recall the construction from~\cite{LST25}: for every $b$-unbiased word $w\in A^d$, an algebraic branching program (ABP) of width $2^b$ computing $P_w$ can be constructed explicitly and efficiently. Here, we give the proof sketch in~\cite{LST25} to keep the exposition self-contained.

\begin{lemma}[\cite{LST25}]
    \label{def: P_w < IMM }
    Let $w \in A^d$ be any word that is $b$-unbiased. Then there is a set-multilinear ABP of width $2^b$ computing $P_w\in\F_{sm}[\overline{X}(w)]$. Moreover,
    \[
        \relrk_w(P_w)\geq2^{-b/2}.
    \]
\end{lemma}

\begin{proof}
    At each layer $i\in\{0,\ldots,d\}$, before identifying the vertices in the last layer, the ABP has exactly $2^{|w_{[i]}|_1}\leq2^b$ vertices. Assume that $w_{[i]}\geq0$. Then, for any monomial $m$ in variable sets $X_1,\ldots,X_i$, the string $\sigma(m_+)$ has length exactly $w_{[i]}$ more than that of $\sigma(m_-)$. Assuming that $\sigma(m_+)$ agrees with $\sigma(m_-)$ on all but its last $w_{[i]}$ bits, i.e., $\sigma(m_+)=\sigma(m_-)\tau$ for $\tau\in\{0,1\}^{w_{[i]}}$, the vertex of the ABP keeps track of the string $\tau$. More formally, for each $\tau\in\{0,1\}^{|w_{[i]}|}$ we have a vertex $v_\tau$ in the $i$th layer of the ABP, where the polynomial $P_{v_\tau}$ computed from the source node to $v_\tau$ is the sum over all monomials $m$ over $X_1,\ldots,X_i$ such that $\sigma(m_+)=\sigma(m_-)\tau$ if $w_{[i]}\geq0$, and $\sigma(m_-)=\sigma(m_+)\tau$ if $w_{[i]}<0$.

    Given a vertex $u_\tau$ on layer $i+1$, we have
    \[
        P_{u_\tau} = \sum_{v_\rho}P_{v_\rho}\cdot L_{\rho,\tau}
    \]
    for a suitable linear polynomial $L_{\rho,\tau}\in\F[X_{i+1}]$, where the sum runs over all vertices $v_\rho$ in the $i$th layer. More precisely,
    \begin{align*}
        L_{\rho,\tau} =
        \begin{cases}
            0 & \parbox[t]{0.70\linewidth}{if $\sgn(w_{[i+1]})=\sgn(w_{[i]})$, $|w_{[i+1]}|\geq|w_{[i]}|$, and $\rho$ is not a prefix of $\tau$,}\\[1mm]
            x_{\rho'} & \parbox[t]{0.70\linewidth}{if $\sgn(w_{[i+1]})=\sgn(w_{[i]})$, $|w_{[i+1]}|\geq|w_{[i]}|$, and $\tau=\rho\rho'$,}\\[1mm]
            0 & \parbox[t]{0.70\linewidth}{if $\sgn(w_{[i+1]})=\sgn(w_{[i]})$, $|w_{[i+1]}|<|w_{[i]}|$, and $\tau$ is not a suffix of $\rho$,}\\[1mm]
            x_{\tau'} & \parbox[t]{0.70\linewidth}{if $\sgn(w_{[i+1]})=\sgn(w_{[i]})$, $|w_{[i+1]}|<|w_{[i]}|$, and $\rho=\tau'\tau$,}\\[1mm]
            x_{\rho\tau} & \parbox[t]{0.70\linewidth}{if $\sgn(w_{[i+1]})\neq\sgn(w_{[i]})$.}
        \end{cases}
    \end{align*}
    Here each $x_\gamma$ denotes the variable of $X_{i+1}$ labelled by the Boolean string $\gamma$.

    Finally, we identify all the vertices in layer $d$ into a single sink. The polynomial computed from the source to this sink is exactly $P_w$.

    The high-relative-rank property
    \[
        \relrk_w(P_w)\geq2^{-b/2}
    \]
    was proved in the preceding subsection.
\end{proof}

In the parameter regime used below, where $b=k(N)$ and $2^b$ is polynomially bounded in $N$, the ABP and its edge labels can be generated uniformly by an $\FNCTwo$ function from $w,d$, and $N$.

Let
\[
    m(w)  \coloneqq 2+\sum_{i=1}^{d-1}2^{|w_{[i]}|_1}.
\]
After identifying the last layer, the ABP constructed above has exactly $m(w)$ vertices. Since $w$ is $b$-unbiased,
\[
    m(w)\leq2+(d-1)2^b.
\]

Let $G=(V,E)$ be an edge-weighted graph. A cycle cover $C$ of $G$ is a set of node-disjoint directed cycles which contains every vertex of $G$. The weight $w(C)$ of $C$ is the product of the weights of the edges in $C$. Then the Permanent of the weighted adjacency matrix $M_G$ is equal to the sum of the weights of the cycle covers of $G$, i.e.,
\[
    \perm(M_G) = \sum_{\text{cycle cover }C\text{ of }G}w(C).
\]

An ABP with source $s$ and sink $t$ can be converted into a graph $G$ such that
\[
    \perm(M_G) = \sum_{\pi:s\leadsto t} \prod_{e\in\pi}L_e,
\]
where $L_e$ is the label of the edge $e$ in the directed path $\pi$ from $s$ to $t$. We add an edge from $t$ to $s$ and label it by $1$. Also, for each node $v$ except $s$, we add a self-loop $v\to v$ and label it by $1$. The adjacency matrix $M_G$ is therefore
\begin{align*}
    M_G(i,j)
    \coloneqq
    \begin{cases}
        \ell(v_i,v_j)
        &\text{if $v_i\to v_j$ is an edge of the ABP constructed above},\\
        1
        &\text{if $i=j$ and $v_i\neq s$},\\
        1
        &\text{if $(v_i,v_j)=(t,s)$},\\
        0
        &\text{otherwise.}
    \end{cases}
\end{align*}

Indeed, since the ABP is acyclic, every cycle cover of $G$ must contain the edge $t\to s$. The remaining non-loop edges in the cycle cover therefore form a unique directed path from $s$ to $t$, while every vertex outside this path is covered by its self-loop. Thus the cycle covers of $G$ are in bijection with the source-to-sink paths in the ABP.

Let $Y=(y_{i,j})_{i,j\in[m(w)]}$ be the variables of the $m(w)\times m(w)$ Permanent. Define the projection $\rho_w$ by replacing $y_{i,j}$ with the corresponding entry $M_G(i,j)$. Every entry of $M_G$ is either a variable of $\overline{X}(w)$ or one of the constants $0,1$. Hence $\rho_w$ is a p-projection, and
\[
    \rho_w\!\left(\perm_{m(w)}(Y)\right) = P_w.
\]

The projection $\rho_w$ can be generated by an $\FNCTwo$ function. We also define an $\FNCTwo$ function
$$
    \ProjCirc(V,G,I,E,\rho_w)
$$
which, given an algebraic circuit $C$ in the variables
$(y_{i,j})_{i,j\in[m]}$, replaces every input variable according to
$\rho_w$. This substitution does not increase the size or the depth of the circuit. If $C$ computes $\perm_m$, then $\ProjCirc(C,\rho_w)$ computes $P_w$.

Recall that
\[
    n = n(N) \coloneqq \lfloor\log_2N\rfloor,\qquad
    d = d(N) \coloneqq \lfloor c\log_2n\rfloor,\qquad
    k = k(N) \coloneqq10d,
\]
where we choose the fixed constant $c$ so that $0<c<1/20$. Let
\[
    m(N)\coloneqq 2+d(N)2^{k(N)}.
\]

For every $k(N)$-unbiased word $w \in A(k(N))^{d(N)}$, the ABP constructed above has at most $m(N)$ vertices. By adding isolated vertices with self-loops of weight $1$, we may therefore regard $P_w$ as a p-projection of the Permanent polynomial $\perm_{m(N)}$.

Since $d(N)=O(\log n(N))$ and $2^{k(N)}\leq n(N)^{10c}$, we have
\[
    m(N)^2=O\bigl(n(N)^{20c}(\log n(N))^2\bigr).
\]

Thus the ABP, the matrix $M_G$, and the projection $\rho_w$ all have encodings of length polynomially bounded in $N$. Moreover, for every circuit of size polynomially bounded in $m(N)$, the projected circuit also has encoding length polynomially bounded in $N$. The stronger assumption $20c<1$ will be used below to ensure that the relevant coefficient spaces fit within the ambient parameter $N$.

Moreover, we define a predicate
\[
    \SetMultOutput(V,G,I,E,m,N)
\]
which expresses that the polynomial computed by the circuit encoded by $(V,G,I,E)$ is set-multilinear with respect to the row partition
\[
    Y_i=\{y_{i,1},\ldots,y_{i,m}\},\qquad i\in[m].
\]
For the constant-depth, polynomial-size circuits considered below, we can compute a polynomial upper bound $D_C$ on the syntactic degree of the output polynomial. Since the circuit has $m^2$ variables and $m=m(N)=O(n(N)^{10c}\log n(N))$, the assumption $20c<1$ implies that, for sufficiently large $N$,
\[
    \binom{m^2+D_C}{D_C}\leq N.
\]

Hence $\CoeffT(V,G,I,E,D_C,N)$ computes the full coefficient vector of the output polynomial. The predicate $\SetMultOutput$ checks that every coefficient corresponding to a monomial which is not set-multilinear with respect to the row partition is zero. Thus $\SetMultOutput$ is a $\Sigma_0^B(\calL_{\FNCTwo})$ predicate.

Consequently,
\[
    \SetMultOutput(V,G,I,E,m(N),N)
    \land
    \CoeffT(V,G,I,E,m(N),N)=\PermT(m(N),N)
\]
implies that the polynomial computed by $(V,G,I,E)$ is exactly $\perm_{m(N)}$. Indeed, if the output polynomial is set-multilinear with respect to the $m(N)$ row blocks, then every monomial occurring in it has total degree exactly $m(N)$. Hence $\CoeffT(V,G,I,E,m(N),N)$ is its full coefficient vector, and equality with $\PermT(m(N),N)$ implies equality of the two polynomials.

Now suppose, for contradiction, that both
\[
    \SetMultOutput(V,G,I,E,m(N),N)
\]
and
\[
    \CoeffT(V,G,I,E,m(N),N)=\PermT(m(N),N)
\]
hold. By the preceding observation, the circuit $C$ encoded by $(V,G,I,E)$ computes $\perm_{m(N)}$. Therefore $\ProjCirc(C,\rho_w)$ has the same size and depth as $C$ and computes
\[
    \rho_w(\perm_{m(N)})=P_w,
\]
contradicting the lower bound for $P_w$.

\begin{theorem}[Provability of superpolynomial lower bound for the Permanent polynomial against constant-depth algebraic circuits]
    \label{theorem: Provability of superpolynomial lower bound for the Permanent polynomial against constant-depth algebraic circuits}
    Let $s$ be a polynomial and let $\Delta$ be a constant. There is a fixed standard constant $N_1$ such that the following formula is provable in $\VNCTwo$:
    \begin{align*}
        \forall N\geq N_1\; \forall V,G,I,E\;\Bigl(
        &\Circuit_{\perm}(V,G,I,E,m(N)) \\
        &\land \SetMultOutput(V,G,I,E,m(N),N) \\
        &\land \Size(V,G,I,E)<s(m(N)) \\
        &\land \Depth_{2\Delta}(V,G,I,E) \\
        &\longrightarrow
        \CoeffT(V,G,I,E,m(N),N)
        \neq
        \PermT(m(N),N)
        \Bigr).
    \end{align*}
\end{theorem}

The matrix of the displayed formula is a $\Sigma_0^B(\calL_{\FNCTwo})$ formula. Every string occurring in the proof has length polynomially bounded in $N$. Hence the Cook--Nguyen translation applies. By conservativity, the corresponding statement is provable in $\VNCTwo$.

Using the translation between bounded arithmetic theories and propositional proofs as shown in~\cite{CN10}, we obtain the following theorem.

\begin{theorem}
    \label{theorem: NC^2-Frege upper bound}
    For every fixed polynomial $s$ and every fixed constant $\Delta$, the propositional translations asserting that the Permanent polynomial does not admit algebraic circuits of size at most $s(m)$ and depth at most $2\Delta$ have polynomial-size $\NCTwo$-Frege proofs over the fixed coefficient field.
\end{theorem}


\let\etalchar\relax 



\section*{Acknowledgements}
We are indebted to the anonymous reviewers of this paper for many important comments.

\section*{AI notice}
AI tools were not used in this work, except to improve the exposition of some parts of Section~\ref{sec: VNCtwo} and \Cref{tab:three-level-picture}, and to suggest local fixes once the manuscript had been completed.

\small

\bibliographystyle{alpha}
\bibliography{reference}




\end{document}